\documentclass[%
 reprint,
superscriptaddress,
 amsmath,amssymb,
 aps,
floatfix,
]{revtex4-2}

\usepackage{graphicx}
\usepackage{dcolumn}
\usepackage{bm}

\usepackage{amsmath,amssymb,amsfonts}
\usepackage{xcolor}
\newtheorem{theorem}{Theorem}

\newtheorem{prop}{Proposition}
\usepackage{hyperref}
\usepackage{braket}

\begin{document}

\preprint{APS/123-QED}

\title{Quantum Fisher kernel for mitigating the vanishing similarity issue}

\author{Yudai Suzuki}
\affiliation{Department of Mechanical Engineering, Keio University, Hiyoshi 3‑14‑1, Kohoku, Yokohama 223‑8522, Japan}
\author{Hideaki Kawaguchi}
\affiliation{Quantum Computing Center, Keio University, Hiyoshi 3‑14‑1, Kohoku,Yokohama 223‑8522, Japan}
\author{Naoki Yamamoto}
\affiliation{Quantum Computing Center, Keio University, Hiyoshi 3‑14‑1, Kohoku,Yokohama 223‑8522, Japan}
\affiliation{Department of Applied Physics and Physico‑Informatics, Keio University, Hiyoshi 3‑14‑1, Kohoku, Yokohama 223‑ 8522, Japan}

\begin{abstract}
Quantum kernel method is a machine learning model exploiting quantum computers to calculate the quantum kernels (QKs) that measure the similarity between data. 
Despite the potential quantum advantage of the method, the commonly used 
fidelity-based QK suffers from a detrimental issue, which we call the vanishing 
similarity issue; detecting the difference between data becomes hard with the 
increase of the number of qubits, due to the exponential decrease of the 
expectation and the variance of the QK. 
This implies the need to design QKs alternative to the fidelity-based one. 
In this work, we propose a new class of QKs called the quantum Fisher kernels (QFKs) that take into account the geometric structure of the data source.
We analytically and numerically demonstrate that the QFK based on the anti-symmetric logarithmic derivatives (ALDQFK) can avoid the issue when the alternating layered ansatzs (ALAs) are used, while the fidelity-based QK cannot even with the ALAs.
Moreover, the Fourier analysis numerically elucidates that the ALDQFK can have 
expressivity comparable to that of the fidelity-based QK. 
These results indicate that the QFK paves the way for practical applications of quantum machine learning with possible quantum advantages. 
\end{abstract}
%
\maketitle
%
%
\section*{Introduction}
Quantum computers have potential to enhance existing machine learning models in terms 
of performance and computational speed. 
Thus far, there have been several proposals of quantum machine learning (QML) algorithms
that outperform the classical counterparts for certain classes of problems 
\cite{farhi2001quantum,rebentrost2014quantum,biamonte2017quantum,liu2021rigorous}. 
An example of QML with possible quantum advantage is the quantum kernel method 
that utilizes quantum computing in the classical kernel methods \cite{havlivcek2019supervised,schuld2019quantum}. 
It has been demonstrated that the quantum kernel methods in combination with classical linear classifiers such as support vector machines, successfully classify some data that cannot be efficiently separated by classical models.
For instance, a synthesized dataset inspired by the discrete logarithmic problem 
(DLP) has been proposed \cite{liu2021rigorous}. 
In addition, a recently proposed procedure that could screen the intrinsic quantum advantages 
of the method \cite{huang2021power} has led to explorations of real-world datasets
\cite{enos2021synthetic,krunic2022quantum}.
Also, the relationship between quantum kernel methods and the so-called quantum 
neural networks has been discussed, e.g., in Ref. \cite{schuld2021supervised}, emphasizing the importance of the method in supervised QML frameworks.

The quantum kernel methods potentially have quantum advantages because the corresponding Hilbert space, which is considered to be hard for classical computers to access efficiently, is used as the feature space for machine learning tasks. 
On the other hand, the use of the large Hilbert space hinders the performance and implementation of the method.
The quantum kernel methods measure the similarity between a pair of data 
$\bm{x},\bm{x'}$ using a function called the quantum kernel (QK), defined as the 
fidelity between data-dependent quantum states \cite{havlivcek2019supervised}: 
\begin{equation}
\label{eq:def_qk}
    k_{Q}(\bm{x},\bm{x'}) = \mathrm{Tr}\left[\rho_{\bm{x},\bm{\theta}}\rho_{\bm{x'},\bm{\theta}}\right].
\end{equation}
Here, $\rho_{\bm{x},\bm{\theta}}=U(\bm{x},\bm{\theta})\rho_{0}U^{\dagger}(\bm{x},\bm{\theta})$ is the density operator representation of the quantum state generated by the input- and parameter-dependent unitary $U(\bm{x},\bm{\theta})$ with the initial state $\rho_0$. 
Then, the Gram matrix composed of QKs given all data pairs is applied to machine learning tasks such as regression and classification. 
However, the fidelity-based QK in Eq.~\eqref{eq:def_qk} has a detrimental issue, which we call the 
\textit{vanishing similarity issue}, stating that all off-diagonal elements of the 
Gram matrix (similarity between different data pairs) significantly vanish as the number of qubits is increased. 
Namely, the expectation and the variance of those elements decrease exponentially 
with respect to the number of qubits. 
This means that an exponential number of measurement shots is required to precisely estimate the QK, 
which erases the possible quantum advantage. 
Thus the realistic number of measurements on real quantum hardware yields a Gram matrix close to the identity matrix; consequently, overfitting happens and the generalization performance of classifiers or regressors using the Gram matrix could be poor.

\begin{figure*}[tb]
    \centering
    \includegraphics[keepaspectratio,scale=0.8]{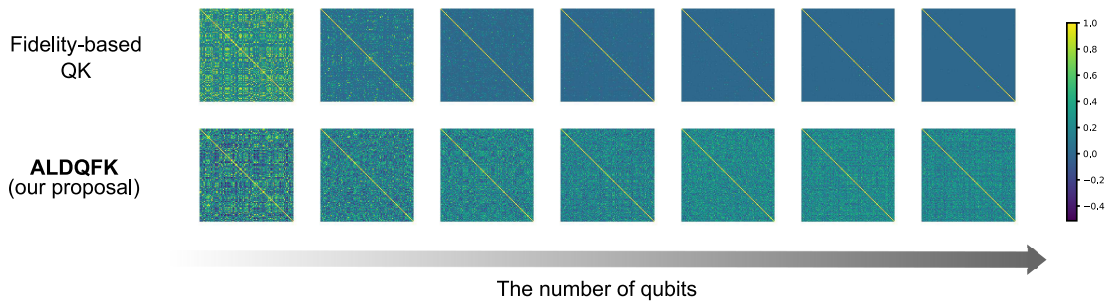}
    \caption{Summary of this work. 
    (a) The vanishing similarity issue for the fidelity-based QK, where the expectation and the variance of the QK decrease exponentially as the number of qubits is increased. Hence, estimating each entry of the Gram matrix on a real quantum device requires an exponential number of measurement shots with respect to the number of qubits. Equivalently, the Gram matrix using the large number of qubits gets close to the identity matrix, resulting in overfitting and poor generalization performance. 
    (b) Comparison of the Gram matrices given by the fidelity-based QK and the 
    proposed ALDQFK, where the ALA is used for both cases.
    The matrices calculated using randomly generated 100 data points ($N=100$) witness that the ALDQFK with ALA can avoid the issue, while the 
    fidelity-based QK cannot.}
    \label{fig:summary}
\end{figure*}

A concept equivalent to the vanishing similarity issue was first introduced in Ref. \cite{huang2021power}, followed by some attempts to analytically understand the phenomenon \cite{kubler2021inductive,canatar2022bandwidth,thanasilp2022exponential}. 
However, this issue has not been resolved yet. 
On the other hand, an analogy of the vanishing similarity issue in the variational quantum algorithms---the barren plateau problem \cite{mcclean2018barren}---can be mitigated by considering the cost function design \cite{cerezo2021cost,khatri2019quantum,larose2019variational} and the structure of the parameterized quantum circuits (PQCs) such as the so-called {\it alternating layered ansatz (ALA)} \cite{cerezo2021cost}.
This gives us insight into a circumventing approach for the vanishing similarity issue; that is, we should design a QK that takes into account the data source structure through the feature map $U(\bm{x},\bm{\theta})$, instead of the fidelity-based one.

In this work, we propose a novel QK called the {\it quantum Fisher kernel (QFK)}, 
as a quantum extension of the classical Fisher kernel \cite{jaakkola1998exploiting}. 
The Fisher kernel is constructed using the information-geometric distance of the data source (i.e., the logarithmic derivatives of the generative model), which as a result incorporates the data structure into the kernel design \cite{jaakkola1998exploiting,tsuda2004asymptotic,hofmann2008kernel}. 
Thus we derive the QFKs that utilizes $\rho_{\bm{x},\bm{\theta}}=U(\bm{x},\bm{\theta})\rho_{0}U^{\dagger}(\bm{x},\bm{\theta})$ as a generative model constituting a set of density-operator-valued data $\{\rho_{\bm{x_i},\bm{\theta}}\}$.
Specifically, we examine the symmetric logarithmic derivative (SLD) 
\cite{helstrom1967minimum} and the anti-symmetric logarithmic derivative (ALD) 
\cite{fujiwara1995quantum}.

Here, with a focus on the vanishing similarity issue, we calculate the expectation and the variance of the fidelity-based QK and the ALD-based QFK (ALDQFK), assuming the quantum circuits satisfying the property of a 2-design \cite{harrow2018approximate,renes2004symmetric,klappenecker2005mutually}.
To be specific, we work on two types of quantum circuits: (1) random quantum circuit acting on all qubits and (2) the ALA. 
We find that the variance of the ALDQFK does not depend on the number of qubits, but on the size of the unitary blocks in the ALA and the depth of the corresponding unitary block, while the same issue arises in the case of random quantum circuit.
That is, the ALDQFK can avoid the issue when the ALA with shallow depth is used.
However, the fidelity-based QK suffers from the vanishing similarity issue for both cases, regardless of the depth.
These results are also confirmed by numerical simulations. 
Hence, according to the results, our proposed ALDQFK with the ALA can avoid the vanishing similarity issue and possibly show better performance than the fidelity-based QK when the large number of qubits is used. 
Figure~\ref{fig:summary} summarizes the aforementioned results. 

Moreover, we numerically show the Fourier representation of the ALDQFK and the fidelity-based QK to demonstrate that they have comparable expressivity. 
We then perform classification tasks using one-dimensional synthesized datasets to validate the Fourier analysis.

An approach to avoid the vanishing similarity issue is to use the projected QK \cite{huang2021power} that reduces the effective dimension by projecting the data-embedded quantum states onto a low-dimensional space. 
However we will not go deep into the approach in this study, as our motivation is to exploit the large Hilbert space using the method inspired by the established classical Fisher kernel. 
We also mention that even the projected QK may suffer from the vanishing similarity issue \cite{kubler2021inductive,thanasilp2022exponential}, although there is more room for further exploration.

\section*{Results}

\subsection*{Motivating examples}

To begin with, we numerically demonstrate the vanishing similarity issue for 
the fidelity-based QK in Eq.~\eqref{eq:def_qk} configured with two types of 
quantum circuits $U(\bm{x},\bm{\theta})$: 
the tensor-product quantum circuits and the instantaneous quantum polynomial 
(IQP) type quantum circuits \cite{havlivcek2019supervised} with depth $L=2$.
Here each layer is composed of an input-embedded circuit and a PQC, and the data re-uploading technique is employed \cite{perez2020data}.
The details are provided in Section $\mathrm{I}\hspace{-1.2pt}\mathrm{V}$.A of Supplementary Information (SI Sec.$\mathrm{I}\hspace{-1.2pt}\mathrm{V}$.A). 
Figure~\ref{fig:motivating_example} shows the expectation and the variance of the fidelity-based QK where each element of input data $\bm{x}$ and parameters in PQC $\bm{\theta}$ is randomly chosen from the range $[-\pi,\pi)$. 
As shown in Figure~\ref{fig:motivating_example}, they both decrease exponentially fast with respect to the number of qubits. 
Note that the IQP-type quantum circuit was first applied to the QK in Ref.~\cite{havlivcek2019supervised} with the motivation that this circuit is conjectured to be hard for classical computers to simulate efficiently \cite{bremner2016average,goldberg2017complexity}. 
Hence, the presence of the vanishing similarity issue in this case is critical, as such possible quantum advantage can be erased. 
Even worse, the tensor-product quantum circuit, which is efficiently simulatable 
by classical means, witnesses the vanishing similarity issue, indicating that even less-expressive quantum circuits with hopeless quantum advantage can suffer from the same issue.

\begin{figure}[tb]
    \centering
    \includegraphics[keepaspectratio,scale=0.45]{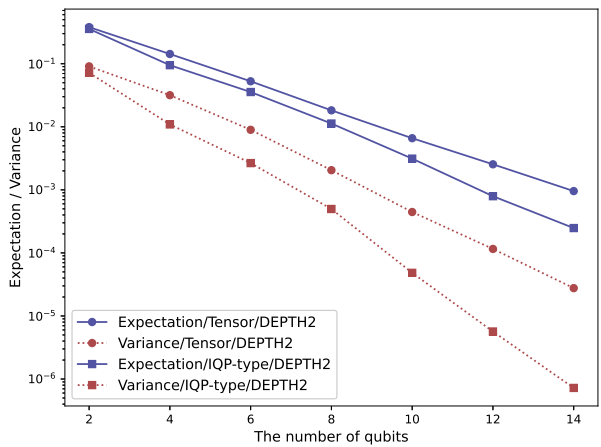}
    \caption{
    The expectation and the variance of the QK in Eq.~\eqref{eq:def_qk}, where each element of input data $\bm{x}$ and parameters in PQC $\bm{\theta}$ is randomly chosen from the range $[-\pi,\pi)$. 
    The tensor-product and the IQP-type quantum circuits with depth $L=2$ are used as $U(\bm{x},\bm{\theta})$.}
    \label{fig:motivating_example}
\end{figure}

\subsection*{Vanishing similarity issue in the fidelity-based QK}

We state the vanishing similarity issue in the fidelity-based QK in general settings. 
Namely, we analytically obtain the expectation and the variance of the QK for two types of quantum circuits: (1) the random quantum circuit acting on all qubits and (2) the ALA, as shown in Figure~\ref{fig:ald_structure} (a) and (b), respectively. 
The ALA is a brick-like quantum circuit with alternating layers of $m$-qubits local unitary blocks, which could mitigate the barren plateau issue thanks to the circuit configuration \cite{cerezo2021cost}.
We here suppose that the total depth is $L$ and the number of unitary blocks in 
each layer is $\kappa$, satisfying $n=m\kappa$ for the total number of qubits 
$n$. 
Also, the $k$-th unitary block in the $d$-th layer, $W_{k,d}(\bm{x},\bm{\theta}_{k,d})$, 
is composed of three types of gates: data-dependent gates $\{R_{B_{k,d}^{\alpha}}(\phi_{k,d}^{\alpha}(\bm{x}))\}$ with a function $\phi_{k,d}^{\alpha}$, parameter-dependent gates $\{R_{{B'}_{k,d}^{\alpha}}(\theta_{\alpha})\}$, and data- and parameter-independent gate $\{S_{k,d}^{\alpha}\}$, where $R_{\sigma}(\theta)=\exp(-i\theta \sigma/2)$ and $B_{k,d}^{\alpha},{B'}_{k,d}^{\alpha}\in\{X,Y,Z\}$ are the Pauli operators on the $\alpha$-th rotation gate. 
See SI Sec.$\mathrm{I}$.B for the details of the circuit settings.

\begin{figure*}[t]
    \centering
    \includegraphics[scale=0.8]{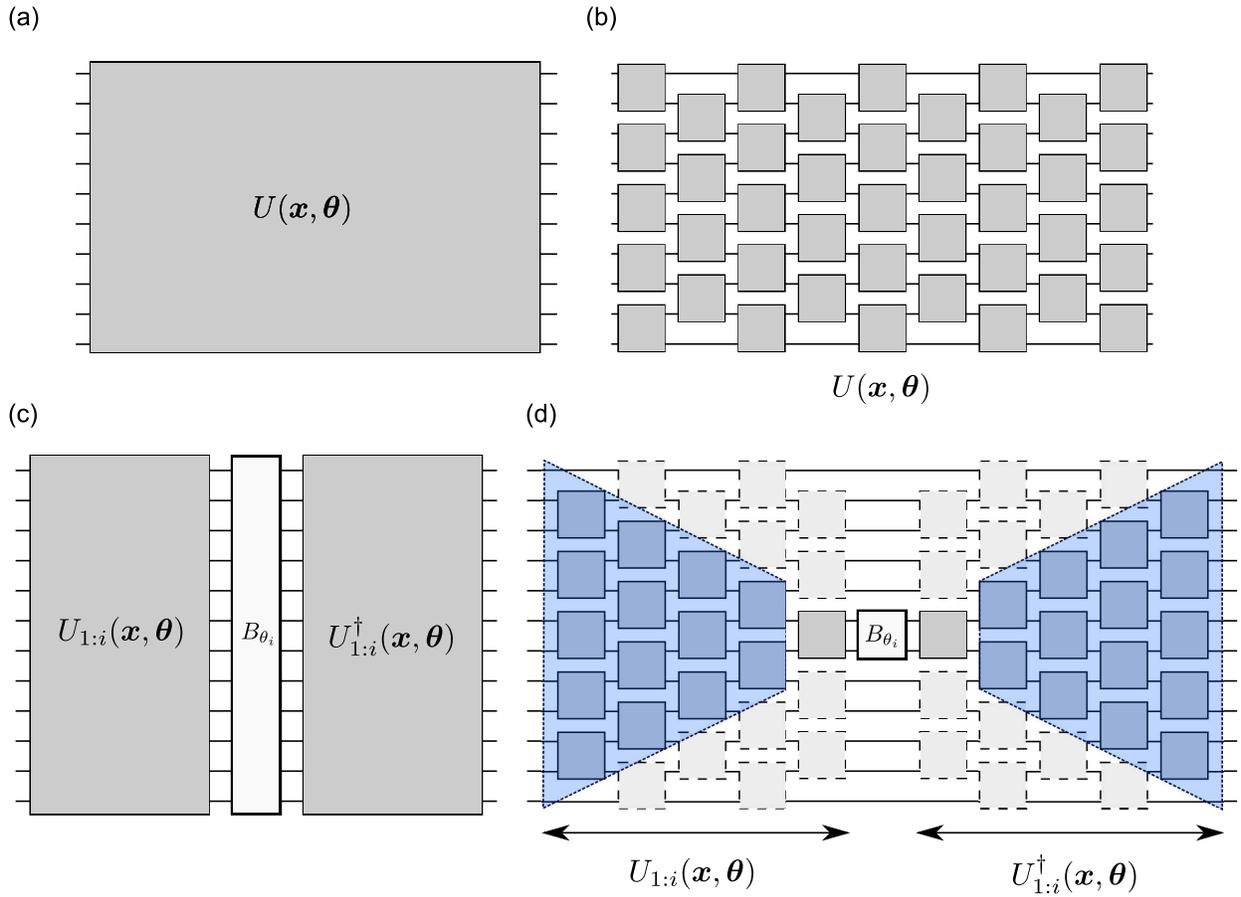}
    \caption{Quantum circuits used in our analysis and quantum circuit representations of the structure used in the ALDQFK, $\tilde{B}_{\bm{x},\theta_{i}}=U_{1:i}^{\dagger}(\bm{x},\bm{\theta})B_{\theta_{i}}U_{1:i}(\bm{x},\bm{\theta})$. Panels (a) and (b) show the random quantum circuit acting on all qubits and the ALA, respectively. Similarly, panels (c) and (d) represent $\tilde{B}_{\bm{x},\theta_{i}}$ for the random quantum circuit and the ALA. 
    For the ALA in panel (d), the thick gray unitary block adjacent to $B_{\theta_{i}}$ is represent by $\tilde{W}_{k,d}(\bm{x},\theta_{i})$ and the shaded region represent $V_{r}(\bm{x},\bm{\theta})$.  }
    \label{fig:ald_structure}
\end{figure*}

Our analysis is based on the assumption that the random quantum circuits and the 
local unitary blocks in the ALA are 2-designs. 
Recall that the $t$-design is an ensemble of unitaries that have the same 
statistical properties as the unitary group with respect to the Haar measure 
up to the $t$-th moment \cite{harrow2018approximate,renes2004symmetric,klappenecker2005mutually}.
The details are given in SI Sec.$\mathrm{I}$.A. 
Then we have the following result (the proof is given in SI Sec.$\mathrm{I}\hspace{-1.2pt}\mathrm{I}$). 

\begin{prop}
\label{prop:prop_qk}
Let the expectation and the variance of the $n$-qubit fidelity-based QK defined in Eq.~\eqref{eq:def_qk} be $\braket{k_{Q}}$ and $Var\left[k_{Q} \right]$, respectively.
Also, let the initial state $\rho_0$ be an arbitrary pure state.

$(1)$ When $U(\bm{x},\bm{\theta})$ and $U(\bm{x'},\bm{\theta})$ are the random 
quantum circuits acting on all qubits, and at least either $U(\bm{x},\bm{\theta})$ 
or $U(\bm{x'},\bm{\theta})$ is a $t$-design with $t\ge2$, the expectation and the 
variance are given by 
\begin{equation}
\label{eq:thm1_1e}
   \braket{k_{Q}} = \frac{1}{2^{n}},
\end{equation}
\begin{equation}
\label{eq:thm1_1v}
    Var\left[k_{Q} \right] = \frac{2^{n}-1}{2^{2n}\left(2^{n}+1\right)} \approx \frac{1}{2^{2n}}.
\end{equation}

$(2)$ Let $U(\bm{x},\bm{\theta})$ and $U(\bm{x'},\bm{\theta})$ be the ALAs, and let $m$-qubit local unitary blocks 
in either $U(\bm{x},\bm{\theta})$ or $U(\bm{x'},\bm{\theta})$ be $t$-designs with 
$t\ge2$. 
Then the expectation and the upper bound of the variance are given by \begin{equation}
\label{eq:thm1_2e}
   \braket{k_{Q}} = \frac{1}{2^{n}},
\end{equation}
\begin{equation}
\label{eq:thm1_2v}
    Var\left[k_{Q} \right] \le \frac{2^{\kappa}}{\left(2^{2m}-1\right)^{\kappa}}-\frac{1}{2^{2m}} \approx \frac{1}{2^{n\left(2-\frac{1}{m}\right)}}.
\end{equation}
\end{prop}

Proposition~\ref{prop:prop_qk} states that, for both types of circuits, all off-diagonal elements of the Gram matrix given by the fidelity-based QK take almost zeros if the number of qubits $n$ is large. 
Note that every diagonal element, i.e., $k_{Q}(\bm{x},\bm{x})$, is $1$ for any input $\bm{x}$ regardless of the number of qubits used. 
Also, the implication of Proposition~\ref{prop:prop_qk} is that devising the circuit structure will not circumvent the issue, as long as the fidelity is used as the metric. 
We remark that the class of ALAs includes the tensor-product quantum circuit 
studied in the motivating example, as a special case (i.e., $L=1$, $m=1$ and 
$\kappa=n$).
Moreover, the initial state $\rho_0$ can be extended to a mixed state with slight 
modification; see SI Sec.$\mathrm{I}\hspace{-1.2pt}\mathrm{I}$.
We should mention that the similar result for the case (1) can be seen in Ref. \cite{kubler2021inductive,thanasilp2022exponential}, but the one for the case (2) has not been reported to the best of our knowledge.

\subsection*{Quantum Fisher kernels}

We define the QFKs by taking into account the structure of the quantum 
data $\rho_{\bm{x},\bm{\theta}}$ used for QML frameworks. 
As described in Introduction, the QFKs are based on the classical Fisher kernel 
\cite{jaakkola1998exploiting} defined as 
\begin{equation}
\begin{split}
    \label{eq:fkernel}
    k_{F}(\bm{x},\bm{x'}) &=\left\langle\bm{g}(\bm{x},\bm{\theta}),\bm{g}(\bm{x'},\bm{\theta})\right\rangle_{\mathcal{I}^{-1}}\\
    &= \bm{g}^\top(\bm{x},\bm{\theta}) \mathcal{I}^{-1}  \bm{g}(\bm{x'},\bm{\theta}),\\
\end{split}
\end{equation}
where $\bm{g}(\bm{x},\bm{\theta})=\nabla_{\bm{\theta}} \log P(\bm{x},\bm{\theta})$ 
is the logarithmic derivative of the generative model 
$P(\bm{x},\theta)$ (called the Fisher score) and $\mathcal{I}$ is the Fisher 
information matrix. 
The idea behind the Fisher kernel is to construct a powerful classifier from 
the probabilistic generative model, using the Fisher score as a natural feature 
vector in the space of probability distributions. 
The classical Fisher kernel has been applied in several areas such as computer 
vision, thanks to its specialized expressivity to data and the performance for 
some tasks 
\cite{perronnin2010improving,sanchez2013image,Sydorov_2014_CVPR,gudovskiy2020deep}.

To define the QFKs, we should first note that there are multiple definitions for the quantum Fisher score \cite{petz1996monotone}. 
In this work, we focus on the symmetric logarithmic derivative (SLD) 
\cite{helstrom1967minimum} and the anti-symmetric logarithmic derivative (ALD) 
\cite{fujiwara1995quantum}.
The SLD $L_{\bm{x},\theta_{l}}^{S}$ and the ALD $L_{\bm{x},\theta_{l}}^{A}$ with respect to the $l$-th parameter $\theta_l$ for $\rho_{\bm{x},\bm{\theta}}=U(\bm{x},\bm{\theta})\rho_0 U^{\dagger}(\bm{x},\bm{\theta})$ are defined as solutions of the following equations, respectively; 
\begin{eqnarray}
\label{eq:sld}
    \partial_{\theta_l} \rho_{\bm{x},\bm{\theta}} &= \frac{1}{2}\left(\rho_{\bm{x},\bm{\theta}}L_{\bm{x},\theta_{l}}^{S}+L_{\bm{x},\theta_{l}}^{S}\rho_{\bm{x},\bm{\theta}}\right),
\\
\label{eq:ald}
    \partial_{\theta_l} \rho_{\bm{x},\bm{\theta}} &= \frac{1}{2}\left(\rho_{\bm{x},\bm{\theta}}L_{\bm{x},\theta_{l}}^{A}-L_{\bm{x},\theta_{l}}^{A}\rho_{\bm{x},\bm{\theta}}\right).
\end{eqnarray}
Here $\partial_{\theta_l}$ represents the partial derivative with respect to the parameter $\theta_{l}$, i.e., $\partial_{\theta_l}\equiv\partial/\partial \theta_l$.
It is known that these quantities are not uniquely determined. 
However, if the initial state is pure, one of the solutions for the SLD equation can be expressed as 
\begin{equation}
\label{eq:sld_pure}
    L_{\bm{x},\theta_{l}}^{S} = 2\partial_{\theta_l} \rho_{\bm{x},\bm{\theta}}.
\end{equation}
Also, a solution of the ALD equation for unitary process can be obtained as
\begin{equation}
\label{eq:ald_unitary_process}
    L_{\bm{x},\theta_{l}}^{A} = i \left(B_{\bm{x},\theta_l} - \mathrm{Tr}\left[\rho_{\bm{x},\bm{\theta}}B_{\bm{x},\theta_l}\right] \right),
\end{equation}
with $B_{\bm{x},\theta_l}=2i(\partial_{\theta_l}U(\bm{x},\bm{\theta}))U^{\dagger}(\bm{x},\bm{\theta})$.
Then, as in the definition of the Fisher kernel in Eq.~\eqref{eq:fkernel}, 
we define the QFK as an inner product of these logarithmic derivatives; 
namely, using $\bm{L}_{\bm{x},\bm{\theta}}^{\gamma}=[L_{\bm{x},\theta_1}^{\gamma},L_{\bm{x},\theta_2}^{\gamma},\ldots]$ with $\gamma = \{S,A\}$, 
\begin{equation}
\label{eq:def_A_S_QFK}
\begin{split}
        k_{QF}^{\gamma}(\bm{x},\bm{x'}) &= \left\langle \bm{L}_{\bm{x},\bm{\theta}}^{\gamma},\bm{L}_{\bm{x'},\bm{\theta}}^{\gamma} \right\rangle_{\mathcal{F}_{\gamma}^{-1}} \\
        &= \sum_{i,j}\mathcal{F}_{\gamma,i,j}^{-1} \left\langle L_{\bm{x},\theta_i}^{\gamma},L_{\bm{x'},\theta_j}^{\gamma} \right\rangle_{\rho},
\end{split}
\end{equation}
where ${\mathcal{F}_{S}}$ (${\mathcal{F}_{A}}$) is the SLD-based (ALD-based) quantum Fisher information matrix. 
In the second line of Eq.~\eqref{eq:def_A_S_QFK}, we introduce the inner product for operators \cite{fujiwara1995quantum}, 
$\langle A,A'\rangle_{\rho}=\frac{1}{2}\mathrm{Tr}[\rho(A'A^{\dagger}+A^{\dagger}A')]$ using certain quantum state $\rho$.

Eq.~\eqref{eq:def_A_S_QFK} can be further specified using properties of the SLD and the ALD.
First, for the case of the SLD with pure states, we have 
$\langle L_{\bm{x},\theta_i}^{S},L_{\bm{x},\theta_j}^{S}\rangle_{\rho_{\bm{x},\bm{\theta}}}=\mathrm{Tr}[L_{\bm{x},\theta_i}^{S} L_{\bm{x},\theta_j}^{S}]$. 
Hence, the SLD-based QFK (SLDQFK) for the pure initial states can be expressed 
in terms of the Hilbert-Schmidt inner product as
\begin{eqnarray}
\label{eq:def_sldqfk}
    k_{QF}^{S}(\bm{x},\bm{x'}) &=  \sum_{i,j} \mathcal{F}_{S,i,j}^{-1} \mathrm{Tr}\left[ L_{\bm{x},\theta_i}^{S} L_{\bm{x'},\theta_j}^{S} \right].
\end{eqnarray}
However, this means that the SLDQFK for the pure state case has a similar structure to 
the fidelity-based QK; 
that is, using the parameter shift rule \cite{mitarai2018quantum,schuld2019evaluating}, we have $\mathrm{Tr}[ L_{\bm{x},\theta_i}^{S} L_{\bm{x'},\theta_j}^{S}] = \mathrm{Tr}[ (\rho_{\bm{x},\theta_{l}^{+}}-\rho_{\bm{x},\theta_{l}^{-}})(\rho_{\bm{x'},\theta_{l}^{+}}-\rho_{\bm{x'},\theta_{l}^{-}})]$, where $\rho_{\bm{x},\theta_{l}^{\pm}}$ is the $\rho_{\bm{x},\theta}$ whose $\theta_l$ is changed to $\theta_l\pm\pi/2$. 
This indicates that the SLDQFK also suffers from the vanishing similarity issue.
Even so, we note that the SLDQFK has an interesting connection with the quantum neural tangent kernels \cite{liu2022representation}; see SI Sec.$\mathrm{V}$.

Next, for the case of the ALD under unitary process, we exploit $\langle L_{\bm{x},\theta_i}^{A},L_{\bm{x},\theta_j}^{A}\rangle_{\rho_{\bm{x},\bm{\theta}}}=-\frac{1}{2} \mathrm{Tr}[\rho_{0}\{L_{\bm{x},\theta_i}^{A, eff}, L_{\bm{x},\theta_j}^{A,eff}\}]$, 
where $\{\cdot,\cdot\}$ is the anti-commutator and 
$L_{\bm{x},\theta_i}^{A,eff}=U^{\dagger}(\bm{x},\bm{\theta})L_{\bm{x},\theta_j}^{A}U(\bm{x},\bm{\theta})$ is the effective ALD operator, of which the SLD version is introduced in the QFIM for unitary process \cite{liu2019quantum}. 
As a result, the ALD-based QFK (ALDQFK) is given by 
\begin{equation}
\label{eq:def_aldqfk}
    k_{QF}^{A}(\bm{x},\bm{x'}) =  -\frac{1}{2}\sum_{i,j} \mathcal{F}_{A,i,j}^{-1} \mathrm{Tr}\left[ \rho_{0} \left\{L_{\bm{x},\theta_i}^{A,eff}, L_{\bm{x'},\theta_j}^{A,eff}\right\} \right].
\end{equation}

In the main text below, the quantum Fisher information matrix in the ALDQFK of Eq.~\eqref{eq:def_aldqfk} is set as the identity matrix, i.e., $\mathcal{F}=\mathbb{I}$, as in the classical case \cite{jaakkola1998exploiting}; this is because the quantum Fisher information matrix is computationally demanding and less significant compared to the Fisher score. 
Moreover, the term $\mathrm{Tr}[\rho_{\bm{x},\bm{\theta}}B_{\bm{x},\theta_l}]$ in the ALD of Eq.~\eqref{eq:ald_unitary_process} is ignored so that the diagonal element of Gram matrix of the ALDQFK, $k_{QF}^{A}(\bm{x},\bm{x})$, is constant for any $\bm{x}$.

\subsection*{Main results}

Here we examine if the ALDQFK could avoid the vanishing similarity issue, under
the assumption that $U(\bm{x},\bm{\theta})$ is (1) the random quantum circuit 
acting on all qubits or (2) the ALA. 
Recall that the fidelity-based QK is subjected to the vanishing similarity issue 
for both cases, as shown in Proposition \ref{prop:prop_qk}. 

We first note that the ALDQFK can be rewritten as 
\begin{equation}
\begin{split}
\label{eq:aldqfk_rewritten}
    k_{QF}^{A}(\bm{x},\bm{x'})  &=  -\frac{1}{2}\sum_{i} \mathrm{Tr}\left[ \rho_{0} \left\{L_{\bm{x},\theta_i}^{A,eff}, L_{\bm{x'},\theta_i}^{A,eff}\right\} \right]\\
    &= \frac{1}{2}\sum_{i} \mathrm{Tr}\left[ \rho_{0} \left\{\tilde{B}_{\bm{x},\theta_{i}} ,\tilde{B}_{\bm{x'},\theta_{i}}\right\} \right],\\
\end{split}
\end{equation}
where $\tilde{B}_{\bm{x},\theta_{i}}=U_{1:i}^{\dagger}(\bm{x},\bm{\theta})B_{\theta_{i}}U_{1:i}(\bm{x},\bm{\theta})$. 
Here, $U_{i:j}(\bm{x},\bm{\theta})$ denotes
a bunch of unitary gates from $U_{i}(\bm{x},\theta_{i})$ to $U_{j}(\bm{x},\theta_{j})$, in the circuit representation $U(\bm{x},\bm{\theta})=U_{D}(\bm{x},\theta_{D})\cdots U_{2}(\bm{x},\theta_{2}) U_{1}(\bm{x},\theta_{1})$. 
We remind that $\theta_{i}$ appears in the form $\exp(-i\theta \sigma/2)$ with $\sigma\in\{X,Y,Z\}$. 
In this section, we focus on the quantity $\mathrm{Tr}[ \rho_{0} \{\tilde{B}_{\bm{x},\theta_{i}} ,\tilde{B}_{\bm{x'},\theta_{i}}\} ]/2$ of Eq.~\eqref{eq:aldqfk_rewritten}.

For ease of analysis, we assume that the random quantum circuit $U_{1:i}(\bm{x},\bm{\theta})$ is a 2-design for arbitrary $i$.
For the ALAs, we assume that the $i$-th parameter $\theta_{i}$ is located in the $k$-th unitary block in the $d$-th layer of the circuits, $W_{k,d}(\bm{x},\bm{\theta}_{k,d})$.
In addition, we decompose the ALA as  $U_{1:i}(\bm{x},\bm{\theta}) = \tilde{W}_{k,d}(\bm{x},\theta_{i})V_{r}(\bm{x},\bm{\theta})$, where $\tilde{W}_{k,d}(\bm{x},\theta_{i})$ is the all gates up to the one containing $i$-th parameter within $W_{k,d}(\bm{x},\bm{\theta}_{k,d})$.
Also, $V_{r}(\bm{x},\bm{\theta})$ is the all unitary blocks up to the $k$-th blocks in the $d$-th layer, which are in the light-cone of $W_{k,d}(\bm{x},\bm{\theta}_{k,d})$.
Then we assume $\tilde{W}_{k,d}(\bm{x},\theta_{i})$ for any $k$ and $d$, and unitary blocks in $V_{r}(\bm{x},\bm{\theta})$ are 2-designs. 
Figure~\ref{fig:ald_structure} (c) and (d) show quantum circuit representations of 
$\tilde{B}_{\bm{x},\theta_{i}}$ for the cases of random quantum circuit and the ALA. 
Under the above setting, we have the following result; the proof is given in SI Sec.$\mathrm{I}\hspace{-1.2pt}\mathrm{I}
\hspace{-1.2pt}\mathrm{I}$.

\begin{theorem}
\label{thm:thm_qfk}
Let the expectation and the variance of the quantity $\mathrm{Tr}[ \rho_{0} \{\tilde{B}_{\bm{x},\theta_{i}} ,\tilde{B}_{\bm{x'},\theta_{j}}\} ]/2$ in the $n$-qubit ALDQFK defined in Eq.~\eqref{eq:aldqfk_rewritten} be $\braket{k_{QF}^{A}}$ and $Var\left[k_{QF}^{A}\right]$, respectively.
Also, let the initial state $\rho_0$ be a pure state.

$(1)$ When $U(\bm{x},\bm{\theta})$ and $U(\bm{x'},\bm{\theta})$ are the random 
quantum circuits acting on all qubits, and both $U_{1:i}(\bm{x},\bm{\theta})$ 
and $U_{1:i}(\bm{x'},\bm{\theta})$ are $t$-designs with $t\ge2$, then we have
\begin{equation}
\label{eq:thm2_1e}
   \braket{k_{QF}^{A}} = 0,
\end{equation}
\begin{equation}
\label{eq:thm2_1v}
    Var\left[k_{QF}^{A} \right] = \frac{2^{n}}{2\left(2^{2n}-1\right)}\left(1+\frac{2^n-2}{2^n\left(2^n+1\right)} \right)
    \approx \frac{1}{2^{n+1}}. 
\end{equation}
$(2)$ Let $U(\bm{x},\bm{\theta})$ and $U(\bm{x'},\bm{\theta})$ be the ALAs, and let $\tilde{W}_{k,d}(\bm{x},\theta_{i})$, $\tilde{W}_{k,d}(\bm{x'},\theta_{i})$ and unitary blocks in $V_{r}(\bm{x},\bm{\theta})$ and $V_{r}(\bm{x'},\bm{\theta})$ be $t$-designs with $t\ge2$. 
Then the expectation is given by
\begin{equation}
\label{eq:thm2_2e}
   \braket{k_{QF}^{A}} = 0.
\end{equation}
Additionally, we assume the initial state $\rho_{0}$ is represented as the tensor product of arbitrary single-qubit pure states $\{\rho_{0,i}\}_{i=1}^{n}$, i.e., $\rho_{0} = \rho_{0,1} \otimes \rho_{0,2} \otimes \ldots \otimes \rho_{0,i} \otimes \ldots \otimes \rho_{0,n}$. Then the lower bound of the variance is given by 
\begin{equation}
\label{eq:thm2_2v}
    Var\left[k_{QF}^{A}\right]  \ge \frac{2^{2md} \left(2^{md} -1\right)}{2\left(2^{2m}-1\right)^{2}\left(2^{m}+1\right)^{4(d-1)}}.
\end{equation}
\end{theorem}

We first remark that the assumption on the initial state for the variance calculation in the case (2) is not too severe from the practical perspective, since the tensor product state is a common choice for the initial state preparation.
Further, we can derive the lower bound of the variance of the ALDQFK for the larger class of the initial states.
The details are provided in SI Sec.$\mathrm{I}\hspace{-1.2pt}\mathrm{I}
\hspace{-1.2pt}\mathrm{I}$.B.

Theorem~\ref{thm:thm_qfk} together with Proposition~\ref{prop:prop_qk} implies that the variance of the ALDQFK is less likely to decrease with respect to the number of qubits than that of the fidelity-based QK. 
In the case (1), although the ALDQFK shows the exponential decrease as well, there is nearly quadratic difference in the scaling compared to the fidelity-based QK. 
Moreover, in the case (2) where the ALA is used, the variance does not depend on the total number of qubits, but the size of 
the unitary block $m$ and the depth $d$ of the local unitary block $W_{k,d}(\bm{x},\bm{\theta}_{k,d})$. 
In other words, the quantity $\mathrm{Tr}[ \rho_{0} \{\tilde{B}_{\bm{x},\theta_{i}} ,\tilde{B}_{\bm{x'},\theta_{i}}\} ]/2$ of the ALDQFK with the ALA can avoid the vanishing similarity issue up to the shallow depth of the quantum circuit.
We remind that the ALDQFK is the summation of $\mathrm{Tr}[ \rho_{0} \{\tilde{B}_{\bm{x},\theta_{i}} ,\tilde{B}_{\bm{x'},\theta_{i}}\} ]/2$ over all $i$, as in Eq.~\eqref{eq:aldqfk_rewritten}.
Thus, even if the ALA with large depth $L$ is used, the variance of the ALDQFK would not be small due to the quantities in the shallow layer.

\subsection*{Numerical experiments}
\begin{figure}[t]
    \centering
    \includegraphics[scale=0.45]{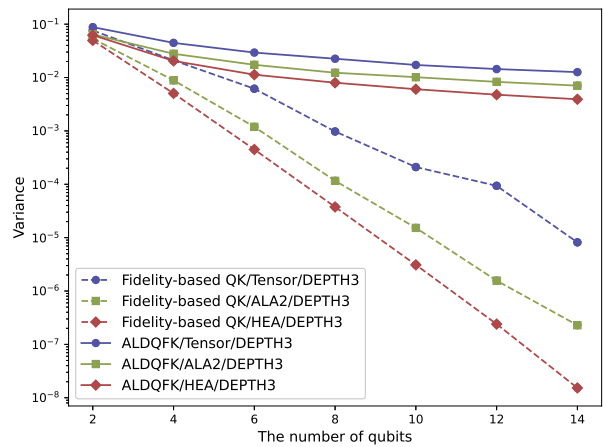}
    \caption{Variance of the fidelity-based QK and the ALDQFK with respect to 
    the number of qubits, $n=2,4,6,8,10,12,14$. 
    Three types of quantum circuits (tensor-product quantum circuits, ALAs with $2$-qubit local unitary blocks and HEAs) with depth $L=3$ are used.}
    \label{fig:main_results}
\end{figure}

We perform numerical simulations to verify Proposition~\ref{prop:prop_qk} and
Theorem~\ref{thm:thm_qfk}. 
Here, the variance of the fidelity-based QK and the ALDQFK are numerically computed for three types of circuits: tensor-product quantum circuits, ALAs with 
$2$-qubit local unitary blocks, and hardware efficient ansatzs (HEAs). 
We use fixed entangling gates and single-qubit rotation gates whose rotation axes, i.e., the Pauli operators $\{X,Y,Z\}$, are randomly chosen, to make the HEAs serve as the random circuits. 
However, we mention that the HEAs and the unitary blocks in the ALA may not form $t$-designs with $t\ge2$. 
See SI Sec.$\mathrm{I}\hspace{-1.2pt}\mathrm{V}$.B for the details.

In the experiments, we randomly generate 100 data points $\{\bm{x}_i\}_{i=1}^{100}$ and a set of PQC parameters $\bm{\theta}$. 
We then calculate the QKs for all different pairs of data, $\bm{x}\neq\bm{x'}$. 
We repeat this process 25 times with different sets of parameters and inputs to obtain the variance for each QK. 
As for the ALDQFK, we use a normalized version, i.e., 
\begin{equation}
\label{eq:normalized_ALDQFK}
    \bar{k}_{QF}^{A}(\bm{x},\bm{x'})= k_{QF}^{A}(\bm{x},\bm{x'})/p, 
\end{equation}
where $p$ is the number of parameters, so that the trace of the Gram matrix 
coincides with the number of the data points as in the fidelity-based QK. 
For the numerical simulation, we use Cirq, a software library for quantum 
computing \cite{cirq_developers_2022_6599601}. 
The detailed settings of the numerical experiments are provided in SI Sec.$\mathrm{I}\hspace{-1.2pt}\mathrm{V}$.B.

Figure~\ref{fig:main_results} shows a semi-log plot of the variance of the QKs against the number of qubits. 
As expected from Proposition~\ref{prop:prop_qk}, the variance of the fidelity-based QK decreases exponentially fast regardless of the type of quantum circuits.
On the other hand, the ALDQFK does not show such exponential decrease of variance for every circuit.
Indeed, the variance of the ALDQFK with the ALA vanish with respect to the number of qubits in Figure~\ref{fig:main_results}, contrary to Theorem~\ref{thm:thm_qfk} (2).
However, this is attributed to the fact that a normalization factor for the ALDQFK, $p$, linearly depends on the number of qubits in our setting.
We remark that the variance of the ALDQFK with HEA does not decrease exponentially fast, which looks inconsistent to Theorem~\ref{thm:thm_qfk} (1). 
This is because the assumption on the $2$-design property of quantum circuits 
is not satisfied due to the insufficient expressivity of the quantum circuits $U_{1:i}(\bm{x},\bm{\theta})$ for any $i$. 
Overall, the numerical experiments confirm the vanishing similarity issue in the 
Fidelity-based QK in Proposition~\ref{prop:prop_qk}, and the validity of Theorem~\ref{thm:thm_qfk} in the main result.

\subsection*{Expressivity comparison of the fidelity-based QK and the ALDQFK}

\begin{figure}[tb]
    \centering
    \includegraphics[scale=0.45]{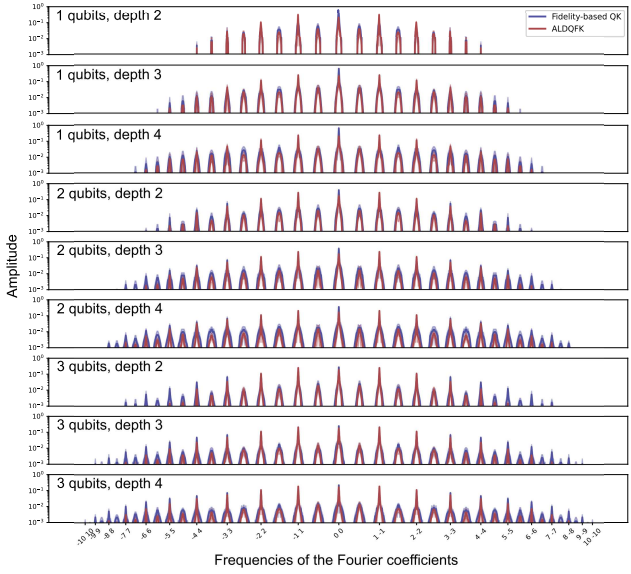}
    \caption{
    Amplitudes of the Fourier coefficients $c_{\omega,\omega'}$ for the 
    fidelity-based QK (blue) and the ALDQFK (red) using the ALAs. The horizontal axis 
    denotes the index $(\omega,\omega')$ with $\omega,\omega'\in\{-12,-11,\ldots,11,12\}$. 
    The number of qubits is $n=1,2,3$ and the circuit depth is $L=2,3,4$. 
    The standard deviations over 10 trials are depicted by the shaded regions.}
    \label{fig:fourier_analysis}
\end{figure}

\begin{figure}[htb]
    \centering
    \includegraphics[scale=0.66]{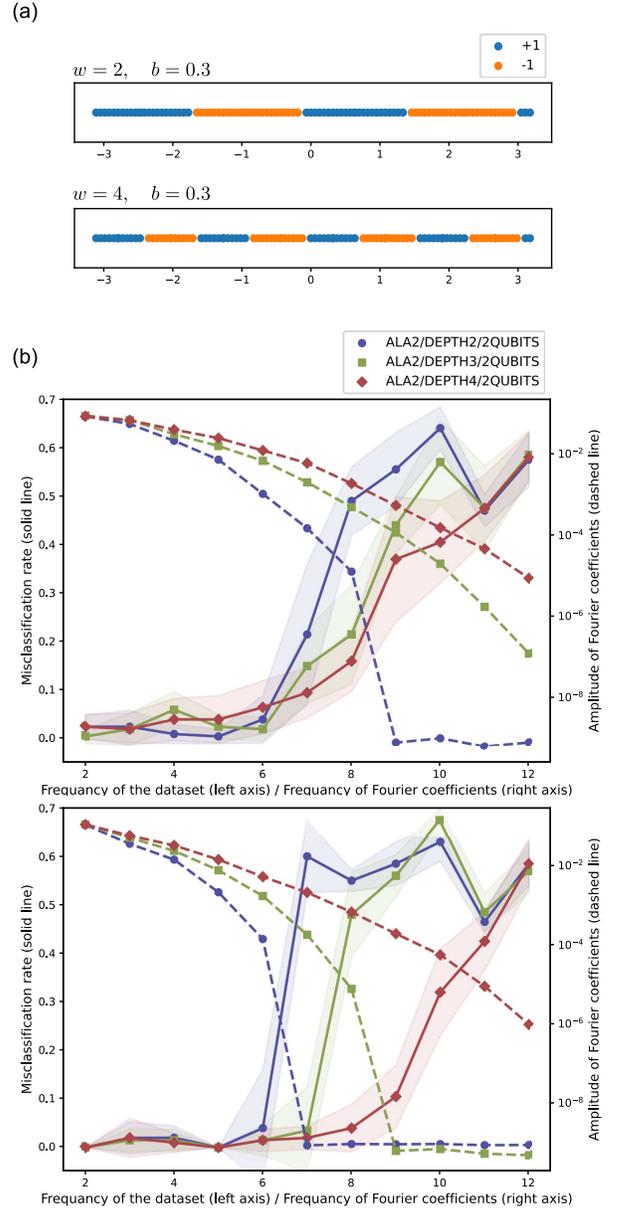}
    \caption{Comparison of misclassification rate for the synthesized datasets 
    and the amplitudes of the Fourier coefficients. 
    (a) The synthesized datasets with $w=2,b=0.3$ and $w=4,b=0.3$ are illustrated. 
    (b) The misclassification rates of the QKs using the ALAs for the datasets with frequencies $w\in \{2,\ldots,12\}$ (solid lines), and the amplitudes of the Fourier coefficients $c_{\omega,-\omega},\omega\in \{2,\ldots,12\}$ given by the corresponding QKs (dashed lines) are demonstrated. 
    The upper and the lower panels show the result for the fidelity-based QK and 
    the ALDQFK, respectively. }
    \label{fig:1d-classification}
\end{figure}

We have addressed that the ALDQFK can be free from the vanishing similarity issue, utilizing the structure of the ALA.
However, the ALDQFK may lose the expressivity, which is another important property that any machine learning models should fulfill.
That is, it is unclear how large class of functions can be approximated by the ALDQFK.
Thus, we here perform numerical analysis to show that the ALDQFK has almost the same level of expressivity as that of the fidelity-based QK.

In general, the QK can be represented as an inner product of two Fourier series \cite{schuld2021supervised}; 
\begin{equation}
\label{eq:fourier_repr}
    k(\bm{x},\bm{x'}) = \sum_{\omega,\omega'\in \Omega} e^{i\omega \bm{x}}e^{i\omega' \bm{x}} c_{\omega,\omega'},
\end{equation}
where $c_{\omega,\omega'}\in\mathbb{C}$ is the Fourier coefficient satisfying  
$c_{\omega,\omega'}=c_{-\omega,-\omega'}^{*}$ and $\Omega$ is the set of integer-valued frequencies.
Hence, we can gauge the expressivity by numerically calculating the magnitude 
of $c_{\omega,\omega'}$ over the effective frequency set. 
Since the Fourier decomposition becomes computationally challenging for the case of high-dimensional data, we use 100 one-dimensional data points and the ALAs with $n=1,2$ and depth $L=2,3,4$ in the experiments. 
Also, we truncate the set of frequency to $\tilde{\Omega}\in\{-12,-11,\ldots,11,12\}$. 
The details are given in SI Sec.$\mathrm{I}\hspace{-1.2pt}\mathrm{V}$.C.

Figure~\ref{fig:fourier_analysis} shows the amplitudes of all Fourier coefficients 
$c_{\omega,\omega'}$ (i.e., totally 625 coefficients) for the fidelity-based QK and the normalized ALDQFK, in several settings of the number of qubits and the 
circuit depth. 
Here, we use the ``curve\_fit" function in Scipy \cite{2020SciPy-NMeth} to obtain the coefficients by fitting each QK to its Fourier representation. 
As shown in Figure~\ref{fig:fourier_analysis}, the ALDQFK has almost the same non-zero Fourier coefficients as those of the fidelity-based QK, in both amplitudes and frequencies. 
Thus, the ALDQFK has the expressivity comparable to that of the fidelity-based QK, while only the former is free from the vanishing similarity issue.

Furthermore, we perform classification tasks using one-dimensional synthesized 
datasets to validate the expressivity analysis. 
The synthesized datasets $\{ (x_{i},y_{i}) \}_{i=1}^{100}$ are composed of one-dimensional input data $x_i \in [-\pi,\pi)$ and the label $y_i \in \{+1,-1\}$ which is determined according to the sign of the sine function $y_i = \text{sign} (\sin (w x_i + b))$ with fixed frequency $w$ and phase $b$. 
This dataset can be used to test if the QKs would have the non-zero 
frequency component corresponding to $w$. Examples of the datasets are shown in Figure~\ref{fig:1d-classification} (a). 
As for the classifiers, we use the support vector machines implemented by SVC provided in scikit-learn \cite{buitinck2013api}. 
The details of the experiment setup are provided in SI Sec.$\mathrm{I}\hspace{-1.2pt}\mathrm{V}$.C.

Figure~\ref{fig:1d-classification} (b) shows the comparisons of the amplitudes of Fourier coefficients and the classification performance between the fidelity-based QK and the ALDQFK.
We use $c_{\omega,-\omega}$ to depict the expressivity over the frequencies, which can be interpreted as the power spectrum over frequencies.
Also, the classification performances are shown using the misclassification rate, which is defined as the number of misclassification over the total number of the test data points.
The range of the horizontal axis in Figure~\ref{fig:1d-classification} (b) shows that the frequency of the Fourier coefficient, $\omega$, and the frequency of the dataset, $w$, are from $2$ to $12$ when the phase is fixed to $b=0.3$.  
Figure~\ref{fig:1d-classification} (b) shows that, for a chosen frequency $w=\omega$, the QK with bigger amplitude of $c_{\omega,-\omega}$ has smaller misclassification rate.
Namely, the QK with higher expressivity results in better classification performance.
As shown in SI Sec.$\mathrm{I}\hspace{-1.2pt}\mathrm{V}$.C, the similar tendency is observed for the case where the HEAs are used. 
In summary, these results support the relevance of using the Fourier analysis 
to execute the performance comparison, and our conclusion is that the ALDQFK 
has the expressivity comparable to the fidelity-based QK.

\section*{Conclusion}

To demonstrate a potential of quantum advantage in machine learning tasks, it has been well recognized that the data structure should be incorporated into the models \cite{liu2021rigorous}.
From this perspective, it is important to design a QK that takes into account the data structure, instead of the fidelity-based one that would suffer from the vanishing similarity issue.
In this work, we propose the QFK as a quantum extension of the classical Fisher kernel, which has already been developed in the classical regime to incorporate the structure of data (generative models) into the kernel design.
We show that the ALDQFK with the ALA may avoid the vanishing similarity issue in the quantum kernel methods, even for the large quantum systems.
Also, the Fourier analysis indicates that the ALDQFK has almost the same expressivity as the fidelity-based QK.

Although the classical Fisher kernel is not as well-known as the Gaussian kernel, the QFK might hold an essentially important position in quantum machine learning filed due to the aforementioned desirable features. In addition, the ALDQFK can be related to quantum dynamics. The structure in the ALDQFK, i.e., $UBU^{\dagger}$, can be regarded as the extent to which the unitary cancellation process is affected by $B$. 
Indeed, the similar structures appear in quantities such as the out-of-ordered 
correlator function \cite{hashimoto2017out, swingle2018unscrambling} and Loschmidt 
Echo \cite{gorin2006dynamics,goussev2012loschmidt}, both of which are used to 
investigate quantum chaos and quantum information scrambling. Therefore, the performance of the ALDQFK can be investigated from the viewpoint of quantum dynamics, which hopefully suggests us some pathway toward quantum advantage.

\begin{acknowledgments}
This work was supported by Grant-in-Aid for JSPS Fellows 22J14183, 
and MEXT Quantum Leap Flagship Program Grants No. JPMXS0118067285 and 
No. JPMXS0120319794. 
\end{acknowledgments}

\bibliographystyle{naturemag}
\bibliography{apssamp}

\end{document}


\title{Supplemental Information for ``Quantum Fisher kernel for mitigating \\ the vanishing similarity issue"}
\author{Yudai Suzuki}
\affiliation{Department of Mechanical Engineering, Keio University, Hiyoshi 3‑14‑1, Kohoku, Yokohama 223‑8522, Japan}
\author{Hideaki Kawaguchi}
\affiliation{Quantum Computing Center, Keio University, Hiyoshi 3‑14‑1, Kohoku,Yokohama 223‑8522, Japan}
\author{Naoki Yamamoto}
\affiliation{Quantum Computing Center, Keio University, Hiyoshi 3‑14‑1, Kohoku,Yokohama 223‑8522, Japan}
\affiliation{Department of Applied Physics and Physico‑Informatics, Keio University, Hiyoshi 3‑14‑1, Kohoku, Yokohama 223‑ 8522, Japan}

\begin{abstract}
This Supplementary Information describes the proof of the Proposition and Theorem, and the setting of the numerical experiments in the main text. 
Also, we present the connection between the symmetric logarithmic derivative-based quantum Fisher kernel (SLDQFK) and the quantum neural tangent kernels.
\end{abstract}

\maketitle

\section{Preliminaries}
We exploit integrals over Haar random unitaries to analytically calculate the expectation and the variance of the quantum kernels (QKs), assuming the quantum circuits are $t$-designs.
Thus, we here present the Lemmas regarding the integrals used for the proof of the Proposition and the Theorem in the main text. 
We also describe the setup of the quantum circuits used in the proof.

\subsection{Formulas of integrals over Haar random unitaries}
For ease of analysis on the expectation and the variance of QKs, we assume that the quantum circuits are $t$-designs. 
The $t$-design is an ensemble of unitaries that have the same 
statistical properties as the unitary group with respect to the Haar measure up to the $t$-th moment \cite{bremner2016average,goldberg2017complexity,harrow2018approximate}. 
Specifically, when the unitary ensemble $\{p_i, W_i\}$ (i.e., a unitary operator $W_i$ is sampled with probability $p_i$) is a $t$-design, the following equality holds;

\begin{equation}
\label{eq:t_design_def}
    \sum_{i} p_i W_{i}^{\otimes t}\rho (W_{i}^{\dagger})^{\otimes t} =  \int d\mu(W) W^{\otimes t}\rho (W^{\dagger})^{\otimes t},
\end{equation}
where the right-hand side represents the integral over the unitary group with respect to Haar measure $d\mu(W)$.
Then the integrals over the ensemble of unitaries forming a $t$-design with $t\ge2$, $\{p_i, W_i\}$, holds

\begin{equation}
\label{eq:1d_int}
    \int d\mu(W) w_{i,j}w^{*}_{l,k} = \frac{\delta_{i,l}\delta_{j,k}}{d},
\end{equation}
\begin{equation}
\label{eq:2d_int}
\begin{split}
    \int d\mu(W) w_{i,j}w^{*}_{l,k}w_{i',j'}w^{*}_{l',k'} &=\frac{\delta_{i,l}\delta_{i',l'}\delta_{j,k}\delta_{j',k'}+\delta_{i,l'}\delta_{i',l}\delta_{j,k'}\delta_{j',k}}{d^2-1}\\
    & \qquad -\frac{\delta_{i,l}\delta_{i',l'}\delta_{j,k'}\delta_{j',k}+\delta_{i,l'}\delta_{i',l}\delta_{j,k}\delta_{j',k'}}{d\left(d^2-1\right)}.
\end{split}
\end{equation}
Here, $W$ is a unitary operator that acts on the $d$-dimensional Hilbert space $\mathcal{H}_{w}$.
In addition, $\delta_{i,j}$ represents the Kronecker delta.

In the main text, we consider two different types of quantum circuits for the calculation of the expectation and the variance.
Specifically, we use the random quantum circuit acting on all qubits and the alternating layered ansatz (ALA). 
Thus, for the sake of clarity, we show the five Lemmas derived and shown in Supplementary Information of Ref. \cite{cerezo2021cost} below. 
In these Lemmas, a unitary operator $W$ acting on the Hilbert space $\mathcal{H}_w$ and $W'$ acting on the bipartite system $\mathcal{H}_{w_1}\otimes\mathcal{H}_{w_2}$ can be written as follows.

\begin{equation}
    \label{eq:expr_mat}
    W = \sum_{i,j} w_{i,j}\ket{i}\bra{j}, \quad W'=\sum_{i,j,i',j'} w'_{ij,i'j'} \ket{ii'}\bra{jj'}.
\end{equation}

\begin{lemma}
\label{lem1}
Let the ensemble of unitaries $\{p_i, W_i\}$ acting on the $d$-dimensional Hilbert space $\mathcal{H}_w$ be a $t$-design with $t\ge1$.
Then, for arbitrary operators $A,B: \mathcal{H}_w\to \mathcal{H}_w$, we have
\begin{equation}
\label{eq:lem1}
    \sum_i p_i \mathrm{Tr}\left[W_i A W_i^{\dagger} B\right] = \int d\mu(W) \mathrm{Tr}\left[WAW^{\dagger}B\right] = \frac{\mathrm{Tr}\left[A\right]\mathrm{Tr}\left[B\right]}{d}.
\end{equation}
\end{lemma}

\begin{proof}
Since the unitary ensemble is a $t$-design, we have
\begin{equation}
\label{eq:lem1_proof}
\begin{split}
    \int d\mu(W) \mathrm{Tr}\left[WAW^{\dagger}B\right] &=  \int d\mu(W) \sum_{i,j,k,l} w_{i,j}a_{j,k}w_{l,k}^{*}b_{l,i} \\
    &= \frac{1}{d} \sum_{i,j} a_{j,j} b_{i,i} \\
    &= \frac{\mathrm{Tr}\left[A\right]\mathrm{Tr}\left[B\right]}{d},\\
\end{split}
\end{equation}
where we use Eq.~\eqref{eq:1d_int}.
\end{proof}

\begin{lemma}
\label{lem2}
Let the ensemble of unitaries $\{p_i, W_i\}$ acting on the $d$-dimensional Hilbert space $\mathcal{H}_w$ be a $t$-design with $t\ge2$.
Then, for arbitrary operators $A,B,C,D: \mathcal{H}_w\to \mathcal{H}_w$, we have

\begin{equation}
\label{eq:lem2}
\begin{split}
    \sum_i p_i \mathrm{Tr}\left[W_i A W_i^{\dagger} B W_i C W_i^{\dagger} D\right] &= \int d\mu(W) \mathrm{Tr}\left[WAW^{\dagger}BWCW^{\dagger}D\right]\\
    &= \frac{1}{d^2-1}\left(\mathrm{Tr}\left[A\right]\mathrm{Tr}\left[C\right]\mathrm{Tr}\left[BD\right]+\mathrm{Tr}\left[AC\right]\mathrm{Tr}\left[B\right]\mathrm{Tr}\left[D\right]\right)\\
    & \qquad -\frac{1}{d\left(d^2-1\right)}\left(\mathrm{Tr}\left[A\right]\mathrm{Tr}\left[B\right]\mathrm{Tr}\left[C\right]\mathrm{Tr}\left[D\right]+\mathrm{Tr}\left[AC\right]\mathrm{Tr}\left[BD\right]\right).
\end{split}
\end{equation}
\end{lemma}

\begin{proof}
Since the unitary ensemble is a $t$-design, we have
\begin{equation}
\label{eq:lem2_proof}
\begin{split}
    \int d\mu(W) \mathrm{Tr}\left[WAW^{\dagger}BWCW^{\dagger}D\right] &=  \int d\mu(W) \sum_{i,j,k,l,i',j',k',l'} w_{i,j}a_{j,k}w_{l,k}^{*}b_{l,i'}w_{i',j'}c_{j',k'}w_{l',k'}^{*}d_{l',i} \\
    &= \frac{1}{d^2-1}\sum_{i,j,k,l}\left(a_{j,j}b_{l,i}c_{k,k}d_{i,l}+a_{j,k}b_{l,l}c_{k,j}d_{i,i}\right)\\
    & \qquad -\frac{1}{d\left(d^2-1\right)}\sum_{i,j,k,l}\left(a_{j,j}b_{l,l}c_{k,k}d_{i,i}+a_{j,k}b_{i,l}c_{k,j}d_{l,i}\right).\\
    &= \frac{1}{d^2-1}\left(\mathrm{Tr}\left[A\right]\mathrm{Tr}\left[C\right]\mathrm{Tr}\left[BD\right]+\mathrm{Tr}\left[AC\right]\mathrm{Tr}\left[B\right]\mathrm{Tr}\left[D\right]\right)\\
    & \qquad -\frac{1}{d\left(d^2-1\right)}\left(\mathrm{Tr}\left[A\right]\mathrm{Tr}\left[B\right]\mathrm{Tr}\left[C\right]\mathrm{Tr}\left[D\right]+\mathrm{Tr}\left[AC\right]\mathrm{Tr}\left[BD\right]\right),\\
\end{split} 
\end{equation}
where we use Eq.~\eqref{eq:2d_int}.
\end{proof}

\begin{lemma}
\label{lem3}
Let the ensemble of unitaries $\{p_i, W_i\}$ on the $d$-dimensional Hilbert space $\mathcal{H}_w$ be a $t$-design with $t\ge2$.
Then, for arbitrary operators $A,B,C,D: \mathcal{H}_w\to \mathcal{H}_w$, we have
\begin{equation}
\label{eq:lem3}
\begin{split}
    \sum_i p_i \mathrm{Tr}\left[W_i A W_i^{\dagger} B\right]\mathrm{Tr}\left[ W_i C W_i^{\dagger} D\right] &= \int d\mu(W) \mathrm{Tr}\left[WAW^{\dagger}B\right]\mathrm{Tr}\left[WCW^{\dagger}D\right]\\
    &= \frac{1}{d^2-1}\left(\mathrm{Tr}\left[A\right]\mathrm{Tr}\left[B\right]\mathrm{Tr}\left[C\right]\mathrm{Tr}\left[D\right]+\mathrm{Tr}\left[AC\right]\mathrm{Tr}\left[BD\right]\right)\\
    & \qquad -\frac{1}{d\left(d^2-1\right)}\left(\mathrm{Tr}\left[A\right]\mathrm{Tr}\left[C\right]\mathrm{Tr}\left[BD\right]+\mathrm{Tr}\left[AC\right]\mathrm{Tr}\left[B\right]\mathrm{Tr}\left[D\right]\right).
\end{split}
\end{equation}
\end{lemma}

\begin{proof}
As is shown in the proof of Lemma \ref{lem2}, we have
\begin{equation}
\label{eq:lem3_proof}
\begin{split}
    \int d\mu(W) \mathrm{Tr}\left[WAW^{\dagger}B\right]\mathrm{Tr}\left[WCW^{\dagger}D\right] &=  \int d\mu(W) \sum_{i,j,k,l,i',j',k',l'} w_{i,j}a_{j,k}w_{l,k}^{*}b_{l,i}w_{i',j'}c_{j',k'}w_{l',k'}^{*}d_{l',i'} \\
    &= \frac{1}{d^2-1}\sum_{i,j,k,l}\left(a_{j,j}b_{l,l}c_{k,k}d_{i,i}+a_{j,k}b_{i,l}c_{k,j}d_{l,i}\right)\\
    & \qquad -\frac{1}{d\left(d^2-1\right)}\sum_{i,j,k,l}\left(a_{j,j}b_{l,i}c_{k,k}d_{i,l}+a_{j,k}b_{l,l}c_{k,j}d_{i,i}\right).\\
    &= \frac{1}{d^2-1}\left(\mathrm{Tr}\left[A\right]\mathrm{Tr}\left[B\right]\mathrm{Tr}\left[C\right]\mathrm{Tr}\left[D\right]+\mathrm{Tr}\left[AC\right]\mathrm{Tr}\left[BD\right]\right)\\
    & \qquad -\frac{1}{d\left(d^2-1\right)}\left(\mathrm{Tr}\left[A\right]\mathrm{Tr}\left[C\right]\mathrm{Tr}\left[BD\right]+\mathrm{Tr}\left[AC\right]\mathrm{Tr}\left[B\right]\mathrm{Tr}\left[D\right]\right).\\
\end{split} 
\end{equation}
where we use Eq.~\eqref{eq:2d_int} in the second equality.
\end{proof}

\begin{lemma}
\label{lem4}
Let the ensemble of unitaries $\{p_i, W_i\}$ on the $d_{w}$-dimensional Hilbert space $\mathcal{H}_w$ be a $t$-design with $t\ge2$.
Also, let $\mathcal{H}=\mathcal{H}_{\bar{w}}\otimes \mathcal{H}_{w}$ be $d_{w}d_{\bar{w}}$-dimensional. 
Then, for arbitrary operators $A,B: \mathcal{H}\to \mathcal{H}$, we have
\begin{equation}
\label{eq:lem4_1}
    \sum_i p_i (\mathbb{I}_{\bar{w}}\otimes W_i) A (\mathbb{I}_{\bar{w}}\otimes W_i^{\dagger}) B = \int d\mu(W) (\mathbb{I}_{\bar{w}}\otimes W) A (\mathbb{I}_{\bar{w}}\otimes W^{\dagger}) B = \frac{\mathrm{Tr}_{w}\left[A\right]\otimes\mathbb{I}_{w}}{d_{w}}B,
\end{equation}
and
\begin{equation}
\label{eq:lem4_2}
    \sum_i p_i \mathrm{Tr}\left[(\mathbb{I}_{\bar{w}}\otimes W_i) A (\mathbb{I}_{\bar{w}}\otimes W_i^{\dagger}) B\right] = \int d\mu(W) \mathrm{Tr}\left[(\mathbb{I}_{\bar{w}}\otimes W) A (\mathbb{I}_{\bar{w}}\otimes W_i^{\dagger}) B\right] = \frac{1}{d_{w}}\mathrm{Tr}\left[\mathrm{Tr}_{w}\left[A\right]\mathrm{Tr}_{w}\left[B\right]\right].
\end{equation}
Here, $\mathbb{I}_{w}$($\mathbb{I}_{\bar{w}}$) represents the identity matrix acting on the Hilbert space $\mathcal{H}_{w}$($\mathcal{H}_{\bar{w}}$) and the partial trace over $\mathcal{H}_{w}$($\mathcal{H}_{\bar{w}}$) is denoted as $\mathrm{Tr}_{w}$($\mathrm{Tr}_{\bar{w}}$). 
Also $\bar{A}$ denotes the complement of $A$. 
\end{lemma}

\begin{proof}
First, following Eq.~\eqref{eq:1d_int}, we have
\begin{equation}
\label{eq:lem4_1_proof}
\begin{split}
    \int d\mu(W) (\mathbb{I}_{\bar{w}}\otimes W) A (\mathbb{I}_{\bar{w}}\otimes W^{\dagger}) B 
    &=\int d\mu(W) \sum_{i,j,k,l} w_{i,j}w^{*}_{l,k} (\mathbb{I}_{\bar{w}}\otimes\ket{i}\bra{j}) A (\mathbb{I}_{\bar{w}}\otimes\ket{k}\bra{l}) B \\
    &= \frac{1}{d_{w}} \sum_{i,j} (\mathbb{I}_{\bar{w}}\otimes\ket{i}\bra{j}) A (\mathbb{I}_{\bar{w}}\otimes\ket{j}\bra{i}) B\\
    &= \frac{\mathrm{Tr}_{w}\left[A\right]\otimes\mathbb{I}_{w}}{d_{w}}B. \\
\end{split} 
\end{equation}
Moreover, according to Eq.\eqref{eq:lem4_1_proof},
\begin{equation}
\label{eq:lem4_2_proof}
\begin{split}
    \int d\mu(W) \mathrm{Tr}\left[(\mathbb{I}_{\bar{w}}\otimes W) A (\mathbb{I}_{\bar{w}}\otimes W_i^{\dagger}) B\right] &= \frac{1}{d_{w}} \mathrm{Tr}\left[\left(\mathrm{Tr}_{w}\left[A\right]\otimes\mathbb{I}_{w}\right)B\right]\\
    &= \frac{1}{d_{w}}\mathrm{Tr}\left[\mathrm{Tr}_{w}\left[A\right]\mathrm{Tr}_{w}\left[B\right]\right].\\
\end{split}
\end{equation}
\end{proof}

\begin{lemma}
\label{lem5}
Let W be a unitary operator acting on the $d_{w}$-dimensional Hilbert space $\mathcal{H}_w$.
Also, let $\mathcal{H}=\mathcal{H}_{\bar{w}}\otimes \mathcal{H}_{w}$ be $d_{w}d_{\bar{w}}$-dimensional with $d_{w}=2^m$ and $d_{\bar{w}}=2^{n-m}$. 
Then, for arbitrary operators $A,B: \mathcal{H}\to \mathcal{H}$, we have
\begin{equation}
\label{eq:lem5}
    \mathrm{Tr}\left[(\mathbb{I}_{\bar{w}}\otimes W) A (\mathbb{I}_{\bar{w}}\otimes W_i^{\dagger}) B\right] = \sum_{\bm{p,q}} \mathrm{Tr} \left[WA_{\bm{qp}},W^{\dagger}B_{\bm{pq}}\right],
\end{equation}
where 
\begin{equation}
    A_{\bm{qp}} = \mathrm{Tr}_{\bar{w}}\left[\left(\ket{\bm{p}}\bra{\bm{q}}\otimes\mathbb{I}_{w}\right)A\right],\quad B_{\bm{pq}} = \mathrm{Tr}_{\bar{w}}\left[\left(\ket{\bm{q}}\bra{\bm{p}}\otimes\mathbb{I}_{w}\right)B\right].
\end{equation}
Here $\bm{q}$ and $\bm{p}$ represent bit-strings of length $n-m$.
\end{lemma}

\begin{proof}
The left-hand side of Eq.~\eqref{eq:lem5} can be expanded as follows;
\begin{equation}
\label{eq:lem5_proof}
\begin{split}
    \mathrm{Tr}\left[(\mathbb{I}_{\bar{w}}\otimes W) A (\mathbb{I}_{\bar{w}}\otimes W_i^{\dagger}) B\right] &= \sum_{i,j,k,l,p,q} w_{i,j} a_{qp,jk} w_{l,k}^{*} a_{pq,li} \\
    &= \sum_{\bm{p},\bm{q}} \mathrm{Tr} \left[ W \mathrm{Tr}_{\bar{w}}\left[\left(\ket{\bm{p}}\bra{\bm{q}}\otimes\mathbb{I}_{\bar{w}}\right)A\right] W^{\dagger} \mathrm{Tr}_{\bar{w}}\left[\left(\ket{\bm{q}}\bra{\bm{p}}\otimes\mathbb{I}_{\bar{w}}\right)B\right] \right] \\
    &= \sum_{\bm{p,q}} \mathrm{Tr} \left[WA_{\bm{qp}},W^{\dagger}B_{\bm{pq}}\right].\\
\end{split}
\end{equation}
\end{proof}

\subsection{The details of the quantum circuits used}
In our analysis, two types of quantum circuits are considered.
One is the random quantum circuit that acts on all $n$ qubits and forms a $t$-design. 
The other is the alternating layered ansatz (ALA) whose $m$-qubit local unitary blocks are $t$-designs. 
More specifically, the ALA can be written as 
\begin{equation}
\label{eq:ALA_unitary}
\begin{split}
    U(\bm{x},\bm{\theta}) &= \prod_{d=1}^{L} V_{d}(\bm{x},\bm{\theta}) \\
    &= \prod_{d=1}^{L} \left(\prod_{k=1}^{\kappa} W_{k,d}(\bm{x},\bm{\theta}_{k,d})\right),
\end{split}
\end{equation}
where the total number of the depth is $L$ and the number of unitary blocks in each layer $\kappa$ satisfying $n=m\kappa$ with the total number of qubits $n$.
Here a unitary block overlaps each $m/2$ qubits on which the corresponding two unitary blocks of one previous layer act; see Figure 3 (b) in the main text.
Note that, when $m=1$, we assume the ALA can be regarded as the tensor-product quantum circuit.
Also, the $k$-th unitary block in the $d$-th layer $W_{k,d}(\bm{x},\bm{\theta}_{k,d})$ can be expressed as
\begin{equation}
    W_{k,d}(\bm{x},\bm{\theta}_{k,d}) = \prod_{\alpha=1}^{n_{(k,d)}} R_{B_{k,d}^{\alpha}}(\theta_{k,d}^{\alpha})R_{{B'}_{k,d}^{\alpha}}(x_{\alpha})R_{k,d}^{\alpha}.
\end{equation}
with data-dependent gates $\{R_{B_{k,d}^{\alpha}}(\phi_{k,d}^{\alpha}(\bm{x}))\}$ with a function $\phi_{k,d}^{\alpha}$, parameter-dependent gates $\{R_{{B'}_{k,d}^{\alpha}}(\theta_{k,d}^{\alpha})\}$ and data- and parameter-independent gate $\{R_{k,d}^{\alpha}\}$. Here $B_{k,d}^{\alpha},{B'}_{k,d}^{\alpha}\in\{X,Y,Z\}$ are the Pauli operators on the $\alpha$-th rotation gate and $ n_{(k,d)}$ is the number of gates composed of these three types of gates in $ W_{k,d}(\bm{x},\bm{\theta}_{k,d})$.
Note that each rotation gate is represented as $R_{\sigma}(\theta)=\exp(-i\theta \sigma/2)$ with the Pauli operator $\sigma$.

\section{Proof of Proposition}
In this section, we derive the expectation and the variance of the fidelity-based QK defined as
\begin{equation}
\label{eq:def_qk}
    k_{Q}(\bm{x},\bm{x'}) = \mathrm{Tr}\left[\rho_{\bm{x},\bm{\theta}}\rho_{\bm{x'},\bm{\theta}}\right],
\end{equation}
where $\rho_{\bm{x},\bm{\theta}}=U(\bm{x},\bm{\theta})\rho_{0}U^{\dagger}(\bm{x},\bm{\theta})$ is the density operator representation of the quantum state with the input- and parameter-dependent unitary $U(\bm{x},\bm{\theta})$ and the initial state $\rho_0$.
Especially, we only focus on the case where $\rho_0$ is an arbitrary pure state, while it can be straightforwardly extended to the mixed state.

\subsection{Case (1): The random quantum circuit acting on all $n$ qubits for the fidelity-based QK}
We here calculate the expectation $\braket{k_{Q}}$ and the variance $Var\left[k_{Q}\right]$ of the fidelity-based QK defined in Eq.~\eqref{eq:def_qk}, assuming either of the random quantum circuits acting on all $n$ qubits, i.e., $ U(\bm{x},\bm{\theta})$ or $U(\bm{x’},\bm{\theta})$ is a $t$-design.

First, we derive the expectation of the QK, assuming either $U(\bm{x},\bm{\theta})$ or $U(\bm{x’},\bm{\theta})$ is a $t$-design with $t\ge1$.
Without loss of generality, we assume only $U(\bm{x},\bm{\theta})$ is a $t$-design with $t\ge1$, due to the symmetry of the fidelity-based QK in Eq.~\eqref{eq:def_qk}.
Then the expectation of the QK over the Haar random unitary, $\braket{k_{Q}}_{ U(\bm{x},\bm{\theta})}$, is calculated as follows;

\begin{equation}
\begin{split}
    \braket{k_{Q}}_{ U(\bm{x},\bm{\theta})}&=\left\langle\mathrm{Tr}\left[U(\bm{x},\bm{\theta})\rho_{0} U^{\dagger}(\bm{x},\bm{\theta})\rho_{\bm{x'},\bm{\theta}}\right]\right\rangle_{ U(\bm{x},\bm{\theta})} \\
    &= \frac{1}{2^{n}} \mathrm{Tr} \left[\rho_{0}\right] \mathrm{Tr} \left[\rho_{\bm{x'},\bm{\theta}}\right]\\
    &= \frac{1}{2^{n}},\\
\end{split}
\end{equation}
where Lemma \ref{lem1} and the property of the density matrix, i.e., $\mathrm{Tr}\left[\rho\right]=1$, are utilized.

Next, we calculate the variance.
The variance $Var\left[k_{Q}\right]$ is expressed as $Var\left[k_{Q}\right]=\braket{k_{Q}^{2}}-\braket{k_{Q}}^2$.
Since we have already had $\braket{k_{Q}}^2=1/2^{2n}$, we focus on $\braket{{k_{Q}}^{2}}$.
Here we assume that $ U(\bm{x},\bm{\theta})$ is a $t$-design with $t\ge2$. 
We remind that it is enough to show the case for $ U(\bm{x},\bm{\theta})$ because of the symmetry.
Then the expectation $\braket{k_{Q}^2}$ can be obtained as 

\begin{equation}
\begin{split}
    \braket{k_{Q}^{2}}_{U(\bm{x},\bm{\theta})}&=\left\langle\mathrm{Tr}\left[U(\bm{x},\bm{\theta})\rho_{0} U^{\dagger}(\bm{x},\bm{\theta})\rho_{\bm{x'},\bm{\theta}}\right]\mathrm{Tr}\left[U(\bm{x},\bm{\theta})\rho_{0} U^{\dagger}(\bm{x},\bm{\theta})\rho_{\bm{x'},\bm{\theta}}\right]\right\rangle_{U(\bm{x},\bm{\theta})}\\
    &= \frac{1}{2^{2n}-1}\left(\mathrm{Tr}\left[\rho_{0}\right]\mathrm{Tr}\left[\rho_{\bm{x'},\bm{\theta}}\right]\mathrm{Tr}\left[\rho_{0}\right]\mathrm{Tr}\left[\rho_{\bm{x'},\bm{\theta}}\right]+\mathrm{Tr}\left[\rho_{0}^{2}\right]\mathrm{Tr}\left[\rho_{\bm{x'},\bm{\theta}}^{2}\right]\right)\\
    & \qquad -\frac{1}{2^{n}\left(2^{2n}-1\right)}\left(\mathrm{Tr}\left[\rho_{0}\right]\mathrm{Tr}\left[\rho_{0}\right]\mathrm{Tr}\left[\rho_{\bm{x'},\bm{\theta}}^{2}\right]+\mathrm{Tr}\left[\rho_{0}^{2}\right]\mathrm{Tr}\left[\rho_{\bm{x'},\bm{\theta}}\right]\mathrm{Tr}\left[\rho_{\bm{x'},\bm{\theta}}\right]\right)\\
    &= \frac{2}{2^{n}\left(2^{n}+1\right)}.\\
\end{split}
\end{equation}
Here we utilize Lemma \ref{lem3} and the property of the pure state, i.e., $\mathrm{Tr}\left[\rho\right]=\mathrm{Tr}\left[\rho^2\right]=1$.

Thus, we have
\begin{equation}
\begin{split}
    Var[k_{Q}] &= \braket{k_{Q}^{2}}_{U(\bm{x},\bm{\theta})}-\braket{k_{Q}}_{U(\bm{x},\bm{\theta})}^2\\
    &= \frac{2}{2^{n}\left(2^{n}+1\right)} - \frac{1}{2^{2n}}\\
    &= \frac{2^{n}-1}{2^{2n}\left(2^{n}+1\right)}\\
\end{split}
\end{equation}

We also remark that $\frac{2^{n}-1}{2^{2n}\left(2^{n}+1\right)}$ is the upper bound of the variance for the case where $\rho_0$ is the mixed state, because we utilize the inequality for the purity, $1/d\le\mathrm{Tr}\left[\rho^2\right]\le1$ with the $d$-dimensional quantum state $\rho$.

\subsection{Case (2): The single-layer ALA for the fidelity-based QK}
We here calculate the expectation $\braket{k_{Q}}$ and the upper bound of the variance $Var\left[k_{Q} \right]$ of the fidelity-based QK, considering the ALA in Eq.~\eqref{eq:ALA_unitary} whose $m$-qubit local unitary blocks are $t$-designs.

We again work on the expectation first. 
We notice that the expectation $\braket{k_{Q}}_{U(\bm{x},\bm{\theta})}$ can be obtained by integrating the quantity over every unitary block, as $\braket{k_{Q}}_{W_{1,1}(\bm{x},\bm{\theta}),W_{2,1}(\bm{x},\bm{\theta}),\ldots W_{\kappa,L}(\bm{x},\bm{\theta}_{\kappa,L})}$.
Thus, we start with the integration over the $\kappa$-th unitary blocks in the last layer, $ W_{\kappa,L}(\bm{x},\bm{\theta}_{\kappa,L})$.

The expectation of the QK over $W_{\kappa,L}(\bm{x},\bm{\theta}_{\kappa,L})$ is calculated in the following way.

\begin{equation}
\label{eq:4-2e_iter1}
\begin{split}
    \braket{k_{Q}}_{W_{\kappa,L}(\bm{x},\bm{\theta}_{\kappa,L})}&=\left\langle\mathrm{Tr}\left[W_{\kappa,L}(\bm{x},\bm{\theta}_{\kappa,L})\rho_{0}^{(\kappa,L)}W_{\kappa,L}^{\dagger}(\bm{x},\bm{\theta}_{\kappa,L})\rho_{\bm{x'},\bm{\theta}}\right]\right\rangle_{W_{\kappa,L}(\bm{x},\bm{\theta}_{\kappa,L})}\\
    &=\left\langle\mathrm{Tr}\left[\left(\mathbb{I}_{\bar{S}_{(\kappa,L)}}\otimes W_{\kappa,L}(\bm{x},\bm{\theta}_{\kappa,L})\right)\rho_{0}^{(\kappa,L)}\left(\mathbb{I}_{\bar{S}_{(\kappa,L)}}\otimes W_{\kappa,L}^{\dagger}(\bm{x},\bm{\theta}_{\kappa,L})\right)\rho_{\bm{x'},\bm{\theta}}\right]\right\rangle_{W_{\kappa,L}(\bm{x},\bm{\theta}_{\kappa,L})}\\
    &= \frac{1}{2^{m}}\mathrm{Tr}\left[\left(\mathrm{Tr}_{S_{(\kappa,L)}}\left[\rho_{0}^{(\kappa,L)}\right]\otimes \mathbb{I}_{S_{(\kappa,L)}}  \right)\rho_{\bm{x'},\bm{\theta}}\right],\\
\end{split}
\end{equation}
where $\rho_{0}^{(a,b)}=U_{a,b}\rho_{0} U_{a,b}^{\dagger}$ with $U_{a,b}= (\prod_{k'=1}^{a-1} W_{k',b}(\bm{x},\bm{\theta}_{k',b}))(\prod_{d=1}^{b-1} V_{d}(\bm{x},\bm{\theta}))$ and $\mathrm{Tr}_{S_{(k,d)}}$($\mathbb{I}_{S_{(k,d)}}$) is the partial trace (the identity operator) over the subspace $S_{(k,d)}$ of the qubits on which $W_{k,d}(\bm{x},\bm{\theta}_{k,d})$ acts. 
Note that $U_{a,b}$ means the all gates up to the $(a-1)$-th blocks in the $b$-th layer.
We also utilize Lemma \ref{lem4} here.

Subsequently, we calculate the expectation over $W_{\kappa-1,L}(\bm{x},\bm{\theta}_{\kappa-1,L})$.
Since $W_{\kappa-1,L}(\bm{x},\bm{\theta}_{\kappa-1,L})$ is involved only in $\mathrm{Tr}_{S_{\kappa,L}}[\rho_{0}^{(\kappa,L)}]$ of Eq.\eqref{eq:4-2e_iter1}, we integrate the quantity only. Then we have
\begin{equation}
\begin{split}
    \left\langle\mathrm{Tr}_{S_{\kappa,L}}\left[\rho_{0}^{(\kappa,L)}\right]\right\rangle_{W_{\kappa-1,L}(\bm{x},\bm{\theta}_{\kappa-1,L})}  &= \left\langle\mathrm{Tr}_{S_{(\kappa,L)}}\left[W_{\kappa-1,L}(\bm{x},\bm{\theta}_{\kappa-1,L})\rho_{0}^{(\kappa-1,L)}W_{\kappa-1,L}^{\dagger}(\bm{x},\bm{\theta}_{\kappa-1,L})\right]\right\rangle_{W_{\kappa-1,L}(\bm{x},\bm{\theta}_{\kappa-1,L})}  \\
    &= \frac{1}{2^{m}} \mathrm{Tr}_{S_{(\kappa,L)}}\left[\mathrm{Tr}_{S_{(\kappa-1,L)}}\left[\rho_{0}^{(\kappa-1,L)}\right]\otimes \mathbb{I}_{S_{(\kappa-1,L)}} \right],\\
\end{split}
\end{equation}
where Lemma \ref{lem4} is used again.
Similarly, we iterate the integration of the quantity for all unitary blocks in the ALA, and then we obtain

\begin{equation}
\label{eq:4-2_iter_all}
\begin{split}
    &\braket{k_{Q}}_{U(\bm{x},\bm{\theta})} \\
    &= \frac{1}{\left(2^{m}\right)^{\kappa L}}\mathrm{Tr}\left[\left(\mathrm{Tr}_{S_{(\kappa,L)}}\left[\mathrm{Tr}_{S_{(\kappa-1,L)}}\left[\ldots \mathrm{Tr}_{S_{(2,1)}}\left[\mathrm{Tr}_{S_{(1,1)}}\left[\rho_{0}^{(1,1)}\right] \otimes \mathbb{I}_{S_{(1,1)}}  \right] \otimes \mathbb{I}_{S_{(2,1)}} \right] \ldots \otimes \mathbb{I}_{S_{(\kappa-1,L)}} \right]\otimes \mathbb{I}_{S_{(\kappa,L)}}  \right)\rho_{x',\theta}\right].
\end{split}
\end{equation}
Here, $\rho_{0}^{(1,1)}=\rho_{0}$ by definition.
Intuitively, the operation $ \mathrm{Tr}_{S_{(k,d)}}[\rho] \otimes \mathbb{I}_{S_{(k,d)}} $ in Eq.\eqref{eq:4-2_iter_all} means a partial trace of the quantum state over subspace $S_{(k,d)}$, in which the state is then replaced with the identity operator $\mathbb{I}_{S_{(k,d)}}$. 
Thus let $\rho_{0} = \sum_{\alpha,\alpha'} c_{\alpha}c_{\alpha’}^{*} \ket{\alpha}\bra{\alpha'}$ be an arbitrary initial state where $\alpha$ and $\alpha'$ are bit-strings, and $c_{\alpha},c_{\alpha’} \in \mathbb{C}$ satisfying $\sum_{\alpha,\alpha'} c_{\alpha}c_{\alpha}^{*}=1$. 
Then the quantity up to the first layer, i.e., $\mathrm{Tr}_{S_{(\kappa,1)}}[\cdots] \otimes \mathbb{I}_{S_{(\kappa,1)}}$ in Eq.\eqref{eq:4-2_iter_all}, can be written as
\begin{equation}
\begin{split}
\label{eq:ptace_replace_identity}
    & \qquad \mathrm{Tr}_{S_{(\kappa,1)}}\left[\mathrm{Tr}_{S_{(\kappa-1,1)}}\left[\ldots \mathrm{Tr}_{S_{(2,1)}}\left[\mathrm{Tr}_{S_{(1,1)}}\left[\rho_{0}\right] \otimes \mathbb{I}_{S_{(1,1)}}  \right] \otimes \mathbb{I}_{S_{(2,1)}} \right] \ldots \otimes \mathbb{I}_{S_{(\kappa-1,1)}} \right]\otimes \mathbb{I}_{S_{(\kappa,1)}} \\ &=  \sum_{\alpha,\alpha'} c_{\alpha}c_{\alpha’}^{*} \left(\prod_{k=1}^{\kappa}\delta_{(\alpha,\alpha')_{S_k}}\right) \times \left(\mathbb{I}_{S_{(1,1)}} \otimes \mathbb{I}_{S_{(2,1)}} \otimes \ldots \otimes \mathbb{I}_{S_{(\kappa,1)}}\right) \\
    &= \mathbb{I}.
\end{split}
\end{equation}
Here we remind every subspace $S_{(k,1)}$ in $\{S_{(k,1)}\}_{k=1}^{\kappa}$ has no overlap with one another.
Consequently, by substituting Eq.~\eqref{eq:ptace_replace_identity} into Eq.~\eqref{eq:4-2_iter_all}, we get
\begin{equation}
    \braket{k_{Q}}_{U(\bm{x},\bm{\theta})} = \frac{\left(2^{m}\right)^{\kappa (L-1)}}{\left(2^{m}\right)^{\kappa L}}= \frac{1}{2^n}.\\
\end{equation}
We note that the numerator in the first equality comes from the trace of the identity operators over the whole system by $L-1$ times.
Also $n=m\kappa$ is used here.

Lastly, we calculate the upper bound of the variance of the QK using the ALA.
As is shown, the variance $Var\left[k_{Q}\right]$ can be described by $Var\left[k_{Q}\right]=\braket{k_{Q}^{2}}-\braket{k_{Q}}^2$.
Thus we focus on $\braket{k_{Q}^{2}}$ because $\braket{k_{Q}}$ has been calculated.
Again, we here assume only $U(\bm{x},\bm{\theta})$ is a $t$-design with $t\ge2$. 

Analogously to the calculation for the expectation, we integrate the quantity over all local unitary blocks in the ALA.
First, the expectation over $W_{\kappa,L}(\bm{x},\bm{\theta}_{\kappa,L})$ is calculated as follows;

\begin{equation}
\label{eq:qk_var_ala_1}
    \begin{split}
    &\braket{k_{Q}^{2}}_{W_{\kappa,L}(\bm{x},\bm{\theta}_{\kappa,L})} \\
    &=\left\langle\mathrm{Tr}\left[W_{\kappa,L}(\bm{x},\bm{\theta}_{\kappa,L})\rho_{0}^{(\kappa,L)}W_{\kappa,L}^{\dagger}(\bm{x},\bm{\theta}_{\kappa,L})\rho_{\bm{x'},\bm{\theta}}\right]\mathrm{Tr}\left[W_{\kappa,L}(\bm{x},\bm{\theta}_{\kappa,L})\rho_{0}^{(\kappa,L)}W_{\kappa,L}^{\dagger}(\bm{x},\bm{\theta}_{\kappa,L})\rho_{\bm{x'},\bm{\theta}}\right]\right\rangle_{W_{\kappa,L}(\bm{x},\bm{\theta}_{\kappa,L})}\\
    &=\left\langle\sum_{\bm{p},\bm{q},\bm{p'},\bm{q'}}\mathrm{Tr}\left[W_{\kappa,L}(\bm{x},\bm{\theta}_{\kappa,L})\rho_{0,\bm{q}\bm{p}}^{(\kappa,L)}W_{\kappa,L}^{\dagger}(\bm{x},\bm{\theta}_{\kappa,L})\rho_{\bm{x'},\bm{\theta},\bm{p}\bm{q}}\right]\mathrm{Tr}\left[W_{\kappa,L}(\bm{x},\bm{\theta}_{\kappa,L})\rho_{0,\bm{q'}\bm{p'}}^{(\kappa,L)}W_{\kappa,L}^{\dagger}(\bm{x},\bm{\theta}_{\kappa,L})\rho_{\bm{x'},\bm{\theta},\bm{p'}\bm{q'}}\right]\right\rangle_{W_{\kappa,L}(\bm{x},\bm{\theta}_{\kappa,L})}\\
    &= \frac{1}{2^{2m}-1} \sum_{\bm{p},\bm{q},\bm{p'},\bm{q'}} \Biggl( \mathrm{Tr}\left[\rho_{0,\bm{q}\bm{p}}^{(\kappa,L)}\right]\mathrm{Tr}\left[\rho_{0,\bm{q'}\bm{p'}}^{(\kappa,L)}\right] \left(  \mathrm{Tr}\left[\rho_{\bm{x'},\bm{\theta},\bm{p}\bm{q}}\right]\mathrm{Tr}\left[\rho_{\bm{x'},\bm{\theta},\bm{p'}\bm{q'}}\right]-\frac{1}{2^{m}}\mathrm{Tr}\left[\rho_{\bm{x'},\bm{\theta},\bm{p}\bm{q}}\rho_{\bm{x'},\bm{\theta},\bm{p'}\bm{q'}}\right]\right)\\
    & \qquad \qquad \qquad \qquad \qquad + \mathrm{Tr}\left[\rho_{0,\bm{q}\bm{p}}^{(\kappa,L)}\rho_{0,\bm{q'}\bm{p'}}^{(\kappa,L)}\right] \left(\mathrm{Tr}\left[\rho_{\bm{x'},\bm{\theta},\bm{p}\bm{q}}\rho_{\bm{x'},\bm{\theta},\bm{p'}\bm{q'}}\right]-\frac{1}{2^{m}}\mathrm{Tr}\left[\rho_{\bm{x'},\bm{\theta},\bm{p}\bm{q}}\right]\mathrm{Tr}\left[\rho_{\bm{x'},\bm{\theta},\bm{p'}\bm{q'}}\right]\right) \Biggr),\\
\end{split}
\end{equation}
where $\rho_{0,\bm{q}\bm{p}}^{(\kappa,L)}=\mathrm{Tr}_{\bar{S}_{(\kappa,L)}}\left[\left(\ket{\bm{p}}\bra{\bm{q}}\otimes \mathbb{I}_{S_{(\kappa,L)}}\right)\rho_{0}^{(\kappa,L)}\right]$ and $\rho_{\bm{x'},\bm{\theta},\bm{p}\bm{q}}=\mathrm{Tr}_{\bar{S}_{(\kappa,L)}}\left[\left(\ket{\bm{q}}\bra{\bm{p}}\otimes \mathbb{I}_{S_{(\kappa,L)}}\right)\rho_{\bm{x'},\bm{\theta}}\right]$.
Here we utilize Lemmas \ref{lem4} and \ref{lem5}.

Next, we integrate the quantity over $W_{\kappa-1,L}(\bm{x},\bm{\theta}_{\kappa-1,L})$.
Since the unitary $W_{\kappa-1,L}(\bm{x},\bm{\theta}_{\kappa-1,L})$ is involved in $\mathrm{Tr}[\rho_{0,\bm{q}\bm{p}}^{(\kappa,L)}]\mathrm{Tr}[\rho_{0,\bm{q'}\bm{p'}}^{(\kappa,L)}]$ and $\mathrm{Tr}[\rho_{0,\bm{q}\bm{p}}^{(\kappa,L)}\rho_{0,\bm{q'}\bm{p'}}^{(\kappa,L)}]$ of Eq.~\eqref{eq:qk_var_ala_1}, we calculate these quantities.
Then the expectation for each quantity can be expressed as follows;

\begin{equation}
\begin{split}
    &\Big\langle \mathrm{Tr}\left[\rho_{0,\bm{q}\bm{p}}^{(\kappa,L)}\right]\mathrm{Tr}\left[\rho_{0,\bm{q'}\bm{p'}}^{(\kappa,L)}\right] \Big\rangle_{W_{\kappa-1,L}(\bm{x},\bm{\theta}_{\kappa-1,L})} \\
    &= \Big\langle \mathrm{Tr}\left[\mathrm{Tr}_{\bar{S}_{(\kappa,L)}}\left[\left(\ket{\bm{p}}\bra{\bm{q}}\otimes \mathbb{I}_{S_{(\kappa,L)}}\right)\rho_{0}^{(\kappa,L)}\right]\right]\mathrm{Tr}\left[\mathrm{Tr}_{\bar{S}_{(\kappa,L)}}\left[\left(\ket{\bm{p'}}\bra{\bm{q'}}\otimes \mathbb{I}_{S_{(\kappa,L)}}\right)\rho_{0}^{(\kappa,L)}\right]\right] \Big\rangle_{W_{\kappa-1,L}(\bm{x},\bm{\theta}_{\kappa-1,L})}\\
    &= \Big\langle \mathrm{Tr}\left[\left(\ket{\bm{p}}\bra{\bm{q}}\otimes \mathbb{I}_{S_{(\kappa,L)}}\right)\rho_{0}^{(\kappa,L)}\right]\mathrm{Tr}\left[\left(\ket{\bm{p'}}\bra{\bm{q'}}\otimes \mathbb{I}_{S_{(\kappa,L)}}\right)\rho_{0}^{(\kappa,L)}\right] \Big\rangle_{W_{\kappa-1,L}(\bm{x},\bm{\theta}_{\kappa-1,L})}\\
    &=  \Big\langle \mathrm{Tr}\left[\left(\ket{\bm{p}}\bra{\bm{q}}\otimes \mathbb{I}_{S_{(\kappa,L)}}\right)W_{\kappa-1,L}(\bm{x},\bm{\theta}_{\kappa-1,L})\rho_{0}^{(\kappa-1,L)}W_{\kappa-1,L}^{\dagger}(\bm{x},\bm{\theta}_{\kappa-1,L})\right] \\
    & \qquad \qquad \qquad \qquad \times \mathrm{Tr}\left[\left(\ket{\bm{p'}}\bra{\bm{q'}}\otimes \mathbb{I}_{S_{(\kappa,L)}}\right)W_{\kappa-1,L}(\bm{x},\bm{\theta}_{\kappa-1,L})\rho_{0}^{(\kappa-1,L)}W_{\kappa-1,L}^{\dagger}(\bm{x},\bm{\theta}_{\kappa-1,L})\right] \Big\rangle_{W_{\kappa-1,L}(\bm{x},\bm{\theta}_{\kappa-1,L})}\\
    &= \frac{1}{2^{2m}-1}\Bigg(\mathrm{Tr} \left[ \left(\ket{\bm{p}}\bra{\bm{q}}_{\bar{S}_{(\kappa-1,L)}}\otimes \mathbb{I}_{S_{(\kappa-1:\kappa,L)}}\right)\rho_{0}^{(\kappa-1,L)}\right]\mathrm{Tr} \left[ \left(\ket{\bm{p'}}\bra{\bm{q'}}_{\bar{S}_{(\kappa-1,L)}}\otimes \mathbb{I}_{S_{(\kappa-1:\kappa,L)}}\right)\rho_{0}^{(\kappa-1,L)}\right] \\
    & \qquad \qquad \qquad \qquad \times \bigg( \delta_{(\bm{p}\bm{q})_{S_{(\kappa-1,L)}}}\delta_{(\bm{p'}\bm{q'})_{S_{(\kappa-1,L)}}} - \frac{1}{2^m} \delta_{(\bm{p}\bm{q'})_{S_{(\kappa-1,L)}}}\delta_{(\bm{p'}\bm{q})_{S_{(\kappa-1,L)}}} \bigg) \\
    &\qquad \qquad + \mathrm{Tr}\left[\mathrm{Tr}_{\bar{S}_{(\kappa-1,L)}}\left[\left(\ket{\bm{p}}\bra{\bm{q}}_{\bar{S}_{(\kappa-1,L)}}\otimes \mathbb{I}_{S_{(\kappa-1:\kappa,L)}}\right)\rho_{0}^{(\kappa-1,L)}\right] \mathrm{Tr}_{\bar{S}_{(\kappa-1,L)}}\left[\left(\ket{\bm{p'}}\bra{\bm{q'}}_{\bar{S}_{(\kappa-1,L)}}\otimes \mathbb{I}_{S_{(\kappa-1:\kappa,L)}}\right)\rho_{0}^{(\kappa-1,L)}\right] \right] \\
    & \qquad \qquad \qquad \qquad \times \bigg( \delta_{(\bm{p}\bm{q'})_{S_{(\kappa-1,L)}}}\delta_{(\bm{p'}\bm{q})_{S_{(\kappa-1,L)}}} - \frac{1}{2^m} \delta_{(\bm{p}\bm{q})_{S_{(\kappa-1,L)}}}\delta_{(\bm{p'}\bm{q'})_{S_{(\kappa-1,L)}}}  \bigg) \Bigg), \\
\end{split}
\end{equation}
\begin{equation}
\begin{split}
    &\Big\langle \mathrm{Tr}\left[\rho_{0,\bm{q}\bm{p}}^{(\kappa,L)}\rho_{0,\bm{q'}\bm{p'}}^{(\kappa,L)}\right] \Big\rangle_{W_{\kappa-1,L}(\bm{x},\bm{\theta}_{\kappa-1,L})} \\
    &= \Big\langle \mathrm{Tr}\left[\mathrm{Tr}_{\bar{S}_{(\kappa,L)}}\left[\left(\ket{\bm{p}}\bra{\bm{q}}\otimes \mathbb{I}_{S_{(\kappa,L)}}\right)\rho_{0}^{(\kappa,L)}\right]\mathrm{Tr}_{\bar{S}_{(\kappa,L)}}\left[\left(\ket{\bm{p'}}\bra{\bm{q'}}\otimes \mathbb{I}_{S_{(\kappa,L)}}\right)\rho_{0}^{(\kappa,L)}\right]\right] \Big\rangle_{W_{\kappa-1,L}(\bm{x},\bm{\theta}_{\kappa-1,L})}\\
    &=  \Big\langle \mathrm{Tr}\Big[\mathrm{Tr}_{\bar{S}_{(\kappa,L)}}\left[\left(\ket{\bm{p}}\bra{\bm{q}}\otimes \mathbb{I}_{S_{(\kappa,L)}}\right)W_{\kappa-1,L}(\bm{x},\bm{\theta}_{\kappa-1,L})\rho_{0}^{(\kappa-1,L)}W_{\kappa-1,L}^{\dagger}(\bm{x},\bm{\theta}_{\kappa-1,L})\right] \\
    & \qquad \qquad \qquad \qquad \times \mathrm{Tr}_{\bar{S}_{(\kappa,L)}}\left[\left(\ket{\bm{p'}}\bra{\bm{q'}}\otimes \mathbb{I}_{S_{(\kappa,L)}}\right)W_{\kappa-1,L}(\bm{x},\bm{\theta}_{\kappa-1,L})\rho_{0}^{(\kappa-1,L)}W_{\kappa-1,L}^{\dagger}(\bm{x},\bm{\theta}_{\kappa-1,L})\right]\Big] \Big\rangle_{W_{\kappa-1,L}(\bm{x},\bm{\theta}_{\kappa-1,L})}\\
    &= \frac{1}{2^{2m}-1}\Bigg(\mathrm{Tr}\left[\mathrm{Tr}_{\bar{S}_{(\kappa,L)}}\left[\left(\ket{\bm{p}}\bra{\bm{q}}_{\bar{S}_{(\kappa-1,L)}}\otimes \mathbb{I}_{S_{(\kappa-1:\kappa,L)}}\right)\rho_{0}^{(\kappa-1,L)}\right] \mathrm{Tr}_{\bar{S}_{(\kappa,L)}}\left[\left(\ket{\bm{p'}}\bra{\bm{q'}}_{\bar{S}_{(\kappa-1,L)}}\otimes \mathbb{I}_{S_{(\kappa-1:\kappa,L)}}\right)\rho_{0}^{(\kappa-1,L)}\right] \right] \\
    & \qquad \qquad \qquad \qquad \times \bigg( \delta_{(\bm{p}\bm{q})_{S_{(\kappa-1,L)}}}\delta_{(\bm{p'}\bm{q'})_{S_{(\kappa-1,L)}}} - \frac{1}{2^m} \delta_{(\bm{p}\bm{q'})_{S_{(\kappa-1,L)}}}\delta_{(\bm{p'}\bm{q})_{S_{(\kappa-1,L)}}} \bigg) \\
    &\qquad \qquad + \mathrm{Tr}\left[\mathrm{Tr}_{\bar{S}_{(\kappa-1:\kappa,L)}}\left[\left(\ket{\bm{p}}\bra{\bm{q}}_{\bar{S}_{(\kappa-1,L)}}\otimes \mathbb{I}_{S_{(\kappa-1:\kappa,L)}}\right)\rho_{0}^{(\kappa-1,L)}\right] \mathrm{Tr}_{\bar{S}_{(\kappa-1):\kappa,L}}\left[\left(\ket{\bm{p'}}\bra{\bm{q'}}_{\bar{S}_{(\kappa-1,L)}}\otimes \mathbb{I}_{S_{(\kappa-1:\kappa,L)}}\right)\rho_{0}^{(\kappa-1,L)}\right] \right] \\ 
    & \qquad \qquad \qquad \qquad \times \bigg( \delta_{(\bm{p}\bm{q'})_{S_{(\kappa-1,L)}}}\delta_{(\bm{p'}\bm{q})_{S_{\kappa-1,L}}} - \frac{1}{2^m} \delta_{(\bm{p}\bm{q})_{S_{(\kappa-1,L)}}}\delta_{(\bm{p'}\bm{q'})_{S_{(\kappa-1,L)}}}  \bigg) \Bigg), \\
\end{split}
\end{equation}
where $S_{(i:j,d)}$ is the subspace from $S_{(i,d)}$ to $S_{(j,d)}$ and  $\delta_{(\bm{p}\bm{q})_{S_{(k,d)}}}$ is the Kronecker delta for $\bm{p}$ and $\bm{q}$ in $S_{(k,d)}$. 
Also, $\ket{\bm{p}}\bra{\bm{q }}_{S_{(k,d)}}$ represents the state $\ket{\bm{p}}\bra{\bm{q}}$ in $S_{(k,d)}$.

Here we repeat this operation for the rest of the unitary blocks in the $L$-th layer.
Fortunately, since the unitary blocks, $W_{k,L}(\bm{x},\bm{\theta}_{k,L})$ with $k\in\{1,\ldots,\kappa-2\}$, are not involved in the delta function of the bit-strings $\bm{p,q,p',q'}$, we only need to focus on the remaining terms.
That is, the following quantities should be integrated over the rest of the unitary blocks:
\begin{equation*}
\begin{split}
&\mathrm{Tr} \left[ \left(\ket{\bm{p}}\bra{\bm{q}}_{\bar{S}_{(\kappa-1,L)}}\otimes \mathbb{I}_{S_{(\kappa-1:\kappa,L)}}\right)\rho_{0}^{(\kappa-1,L)}\right]\mathrm{Tr} \left[ \left(\ket{\bm{p'}}\bra{\bm{q'}}_{\bar{S}_{(\kappa-1,L)}}\otimes \mathbb{I}_{S_{(\kappa-1:\kappa,L)}}\right)\rho_{0}^{(\kappa-1,L)}\right],\\
&\mathrm{Tr}\left[\mathrm{Tr}_{\bar{S}_{(\kappa-1,L)}}\left[\left(\ket{\bm{p}}\bra{\bm{q}}_{\bar{S}_{(\kappa-1,L)}}\otimes \mathbb{I}_{S_{(\kappa-1:\kappa,L)}}\right)\rho_{0}^{(\kappa-1,L)}\right] \mathrm{Tr}_{\bar{S}_{(\kappa-1,L)}}\left[\left(\ket{\bm{p'}}\bra{\bm{q'}}_{\bar{S}_{(\kappa-1,L)}}\otimes \mathbb{I}_{S_{(\kappa-1:\kappa,L)}}\right)\rho_{0}^{(\kappa-1,L)}\right] \right],\\
&\mathrm{Tr}\left[\mathrm{Tr}_{\bar{S}_{(\kappa,L)}}\left[\left(\ket{\bm{p}}\bra{\bm{q}}_{\bar{S}_{(\kappa-1,L)}}\otimes \mathbb{I}_{S_{(\kappa-1:\kappa,L)}}\right)\rho_{0}^{(\kappa-1,L)}\right] \mathrm{Tr}_{\bar{S}_{(\kappa,L)}}\left[\left(\ket{\bm{p'}}\bra{\bm{q'}}_{\bar{S}_{(\kappa-1,L)}}\otimes \mathbb{I}_{S_{(\kappa-1:\kappa,L)}}\right)\rho_{0}^{(\kappa-1,L)}\right] \right],\\
&\mathrm{Tr}\left[\mathrm{Tr}_{\bar{S}_{(\kappa-1:\kappa,L)}}\left[\left(\ket{\bm{p}}\bra{\bm{q}}_{\bar{S}_{(\kappa-1,L)}}\otimes \mathbb{I}_{S_{(\kappa-1:\kappa,L)}}\right)\rho_{0}^{(\kappa-1,L)}\right] \mathrm{Tr}_{\bar{S}_{(\kappa-1:\kappa,L)}}\left[\left(\ket{\bm{p'}}\bra{\bm{q'}}_{\bar{S}_{(\kappa-1,L)}}\otimes \mathbb{I}_{S_{(\kappa-1:\kappa,L)}}\right)\rho_{0}^{(\kappa-1,L)}\right] \right].
\end{split}
\end{equation*}
Suppose we integrate these quantities over the unitary blocks $W_{k,L}(\bm{x},\bm{\theta}_{k,L})$ in the descending order with respect to $k$.
Then the quantities after integration over $W_{k,L}(\bm{x},\bm{\theta}_{k,L})$ fall into two types; 
\begin{equation}
\begin{split}
    &\mathrm{Tr} \left[ \left(\ket{\bm{p}}\bra{\bm{q}}_{\bar{S}_{(k:\kappa-1,L)}}\otimes \mathbb{I}_{S_{(k:\kappa,L)}}\right)\rho_{0}^{(k,L)}\right]\mathrm{Tr} \left[ \left(\ket{\bm{p'}}\bra{\bm{q'}}_{\bar{S}_{(k:\kappa-1,L)}}\otimes \mathbb{I}_{S_{(k:\kappa,L)}}\right)\rho_{0}^{(k,L)}\right],\\
    &\mathrm{Tr}\left[\mathrm{Tr}_{S_{(k')}}\left[\left(\ket{\bm{p}}\bra{\bm{q}}_{\bar{S}_{(k:\kappa-1,L)}}\otimes \mathbb{I}_{S_{(k:\kappa,L)}}\right)\rho_{0}^{(k,L)}\right] \mathrm{Tr}_{S_{(k')}}\left[\left(\ket{\bm{p'}}\bra{\bm{q'}}_{\bar{S}_{(k:\kappa-1,L)}}\otimes \mathbb{I}_{S_{(k:\kappa,L)}}\right)\rho_{0}^{(k,L)}\right] \right],
\end{split}
\end{equation}
where $S_{(k')}$ denotes certain subspace of the whole systems.
Also, for $k\in\{2,\ldots,\kappa-2\}$, we can integrate the above quantities over $W_{k-1,L}(\bm{x},\bm{\theta}_{k-1,L})$ in the following way.

\begin{equation}
\begin{split}
    &\Big\langle \mathrm{Tr} \left[ \left(\ket{\bm{p}}\bra{\bm{q}}_{\bar{S}_{(k:\kappa-1,L)}}\otimes \mathbb{I}_{S_{(k:\kappa,L)}}\right)\rho_{0}^{(k,L)}\right]\mathrm{Tr} \left[ \left(\ket{\bm{p'}}\bra{\bm{q'}}_{\bar{S}_{(k:\kappa-1,L)}}\otimes \mathbb{I}_{S_{(k:\kappa,L)}}\right)\rho_{0}^{(k,L)}\right] \Big\rangle_{W_{k-1,L}(\bm{x},\bm{\theta})_{k-1,L}} \\
    & \qquad =  \Big\langle \mathrm{Tr} \left[ \left(\ket{\bm{p}}\bra{\bm{q}}_{\bar{S}_{(k:\kappa-1,L)}}\otimes \mathbb{I}_{S_{(k:\kappa,L)}}\right)W_{k-1,L}(\bm{x},\bm{\theta}_{k-1,L})\rho_{0}^{(k-1,L)}W_{k-1,L}^{\dagger}(\bm{x},\bm{\theta}_{k-1,L})\right]\\
    & \qquad \qquad \times \mathrm{Tr} \left[ \left(\ket{\bm{p'}}\bra{\bm{q'}}_{\bar{S}_{(k:\kappa-1,L)}}\otimes \mathbb{I}_{S_{(k:\kappa,L)}}\right)W_{k-1,L}(\bm{x},\bm{\theta}_{k-1,L})\rho_{0}^{(k-1,L)}W_{k-1,L}^{\dagger}(\bm{x},\bm{\theta}_{k-1,L})\right] \Big\rangle_{W_{k-1,L}(\bm{x},\bm{\theta}_{k-1,L})}\\
    & \qquad = \frac{1}{2^{2m}-1}\Bigg(\mathrm{Tr} \left[ \left(\ket{\bm{p}}\bra{\bm{q}}_{\bar{S}_{(k-1:\kappa-1,L)}}\otimes \mathbb{I}_{S_{(k-1:\kappa,L)}}\right)\rho_{0}^{(k-1,L)}\right]\mathrm{Tr} \left[ \left(\ket{\bm{p'}}\bra{\bm{q'}}_{\bar{S}_{(k-1:\kappa-1,L)}}\otimes \mathbb{I}_{S_{(k-1:\kappa,L)}}\right)\rho_{0}^{(k-1,L)}\right]\\
    & \qquad \qquad \qquad \qquad \times \bigg( \delta_{(\bm{p}\bm{q})_{S_{(k-1,L)}}}\delta_{(\bm{p'}\bm{q'})_{S_{(k-1,L)}}} - \frac{1}{2^m} \delta_{(\bm{p}\bm{q'})_{S_{(k-1,L)}}}\delta_{(\bm{p'}\bm{q})_{S_{(k-1,L)}}} \bigg) \\
    &\qquad \qquad  \quad+ \mathrm{Tr}\Biggl[\mathrm{Tr}_{\bar{S}_{(k-1,L)}}\left[\left(\ket{\bm{p}}\bra{\bm{q}}_{\bar{S}_{(k-1:\kappa-1,L)}}\otimes \mathbb{I}_{S_{(k-1:\kappa,L)}}\right)\rho_{0}^{(k-1,L)}\right] \\
    & \qquad \qquad \qquad \qquad \times \mathrm{Tr}_{\bar{S}_{(k-1,L)}}\left[\left(\ket{\bm{p'}}\bra{\bm{q'}}_{\bar{S}_{(k-1:\kappa-1,L)}}\otimes \mathbb{I}_{S_{(k-1:\kappa,L)}}\right)\rho_{0}^{(k-1,L)}\right] \Biggr] \\
    & \qquad \qquad \qquad \qquad \qquad \times \bigg( \delta_{(\bm{p}\bm{q'})_{S_{(k-1,L)}}}\delta_{(\bm{p'}\bm{q})_{S_{(k-1,L)}}} - \frac{1}{2^m} \delta_{(\bm{p}\bm{q})_{S_{(k-1,L)}}}\delta_{(\bm{p'}\bm{q'})_{S_{(k-1,L)}}}  \bigg) \Bigg), \\
\end{split}
\end{equation}
\begin{equation}
\begin{split}
    &\Big\langle \mathrm{Tr}\left[\mathrm{Tr}_{S_{(k')}}\left[\left(\ket{\bm{p}}\bra{\bm{q}}_{\bar{S}_{(k:\kappa-1,L)}}\otimes \mathbb{I}_{S_{(k:\kappa,L)}}\right)\rho_{0}^{(k,L)}\right] \mathrm{Tr}_{S_{(k')}}\left[\left(\ket{\bm{p'}}\bra{\bm{q'}}_{\bar{S}_{(k:\kappa-1,L)}}\otimes \mathbb{I}_{S_{(k:\kappa,L)}}\right)\rho_{0}^{(k,L)}\right] \right] \Big\rangle_{W_{k-1,L}(\bm{x},\bm{\theta}_{k-1,L})} \\
    & \qquad =  \Big\langle \mathrm{Tr}\Big[\mathrm{Tr}_{S_{(k')}}\left[\left(\ket{\bm{p}}\bra{\bm{q}}_{\bar{S}_{(k:\kappa-1,L)}}\otimes \mathbb{I}_{S_{(k:\kappa,L)}}\right)W_{k-1,L}(\bm{x},\bm{\theta}_{k-1,L})\rho_{0}^{(k-1,L)}W_{k-1,L}(\bm{x},\bm{\theta}_{k-1,L})^{\dagger}\right]\\
    & \qquad \qquad \qquad \times \mathrm{Tr}_{S_{(k')}}\left[\left(\ket{\bm{p'}}\bra{\bm{q'}}_{\bar{S}_{(k:\kappa-1,L)}}\otimes \mathbb{I}_{S_{(k:\kappa,L)}}\right)W_{k-1,L}(\bm{x},\bm{\theta}_{k-1,L})\rho_{0}^{(k-1,L)}W_{k-1,L}^{\dagger}(\bm{x},\bm{\theta}_{k-1,L})\right] \Big] \Big\rangle_{W_{k-1,L}(\bm{x},\bm{\theta}_{k-1,L})}\\
    & \qquad = \frac{1}{2^{2m}-1}\Bigg(\mathrm{Tr}\Biggl[\mathrm{Tr}_{S_{(k')}}\left[\left(\ket{\bm{p}}\bra{\bm{q}}_{\bar{S}_{(k-1:\kappa-1,L)}}\otimes \mathbb{I}_{S_{(k-1:\kappa-1,L)}}\right)\rho_{0}^{(k-1,L)}\right] \\
    & \qquad \qquad \qquad \qquad \times \mathrm{Tr}_{S_{(k')}}\left[\left(\ket{\bm{p'}}\bra{\bm{q'}}_{\bar{S}_{(k-1:\kappa-1,L)}}\otimes \mathbb{I}_{S_{(k-1:\kappa-1,L)}}\right)\rho_{0}^{(k-1,L)}\right] \Biggr]\\
    & \qquad \qquad \qquad \qquad \qquad \times \bigg( \delta_{(\bm{p}\bm{q})_{S_{(k-1,L)}}}\delta_{(\bm{p'}\bm{q'})_{S_{(k-1,L)}}} - \frac{1}{2^m} \delta_{(\bm{p}\bm{q'})_{S_{(k-1,L)}}}\delta_{(\bm{p'}\bm{q})_{S_{(k-1,L)}}} \bigg) \\
    &\qquad  \quad+ \mathrm{Tr}\Biggl[\mathrm{Tr}_{S_{(k')}/S_{(k-1,L)}}\left[\left(\ket{\bm{p}}\bra{\bm{q}}_{\bar{S}_{(k-1:\kappa-1,L)}}\otimes \mathbb{I}_{S_{(k-1:\kappa,L)}}\right)\rho_{0}^{(k-1,L)}\right]\\
    & \qquad \qquad \qquad \qquad \times \mathrm{Tr}_{S_{(k')}/S_{(k-1,L)}}\left[\left(\ket{\bm{p'}}\bra{\bm{q'}}_{\bar{S}_{(k-1:\kappa-1,L)}}\otimes \mathbb{I}_{S_{(k-1:\kappa,L)}}\right)\rho_{0}^{(k-1,L)}\right] \Biggr] \\
    & \qquad \qquad \qquad \qquad \qquad \times \bigg( \delta_{(\bm{p}\bm{q'})_{S_{(k-1,L)}}}\delta_{(\bm{p'}\bm{q})_{S_{(k-1,L)}}} - \frac{1}{2^m} \delta_{(\bm{p}\bm{q})_{S_{(k-1,L)}}}\delta_{(\bm{p'}\bm{q'})_{S_{(k-1,L)}}}  \bigg) \Bigg). \\
\end{split}
\end{equation}
Therefore, by applying these equations iteratively, the expectation for the unitary blocks in the last layer, $\braket{k_{Q}^{2}}_{V_{L}(\bm{x},\bm{\theta})}$, is calculated as follows.

\begin{equation}
\begin{split}
   & \braket{k_{Q}^{2}}_{V_{L}(\bm{x},\bm{\theta})} = \frac{1}{\left(2^{2m}-1\right)^{\kappa}} \times \\
   &  \sum_{\bm{p},\bm{q},\bm{p'},\bm{q'}}  \sum_{S_k\in P(S^{(1:\kappa-1,L)})}  \prod_{h\in \bar{S}_{k}\cap S^{(1:\kappa-1,L)}} \left(\delta_{(\bm{p}\bm{q})_h}\delta_{(\bm{p'}\bm{q'})_h}-\frac{1}{2^{m}}\delta_{(\bm{p}\bm{q'})_h}\delta_{(\bm{p'}\bm{q})_h}\right) \prod_{h\in S_{k}} \left(\delta_{(\bm{p}\bm{q'})_h}\delta_{(\bm{p'}\bm{q})_h}-\frac{1}{2^{m}}\delta_{(\bm{p}\bm{q})_h}\delta_{(\bm{p'}\bm{q'})_h}\right) \\
   & \qquad  \qquad \times \Bigg( \mathrm{Tr}\left[\mathrm{Tr}_{\bar{S}_{k}}\left[\rho_{0}^{(1,L)}\right] \mathrm{Tr}_{\bar{S}_{k}}\left[\rho_{0}^{(1,L)}\right] \right] \left(  \mathrm{Tr}\left[\rho_{\bm{x'},\bm{\theta},\bm{p}\bm{q}}\right]\mathrm{Tr}\left[\rho_{\bm{x'},\bm{\theta},\bm{p'}\bm{q'}}\right]-\frac{1}{2^{m}}\mathrm{Tr}\left[\rho_{\bm{x'},\bm{\theta},\bm{p}\bm{q}}\rho_{\bm{x'},\bm{\theta},\bm{p'}\bm{q'}}\right]\right)\\
   & \qquad \qquad \quad +  \mathrm{Tr}\left[\mathrm{Tr}_{\overline{S_{k} \cup S_{(\kappa,L)}}}\left[\rho_{0}^{(1,L)}\right] \mathrm{Tr}_{\overline{S_{k} \cup S_{(\kappa,L)}}}\left[\rho_{0}^{(1,L)}\right] \right] \left(\mathrm{Tr}\left[\rho_{\bm{x'},\bm{\theta},\bm{p}\bm{q}}\rho_{x',\theta,\bm{p'}\bm{q'}}\right]-\frac{1}{2^{m}}\mathrm{Tr}\left[\rho_{\bm{x'},\bm{\theta},\bm{p}\bm{q}}\right]\mathrm{Tr}\left[\rho_{\bm{x'},\bm{\theta},\bm{p'}\bm{q'}}\right]\right) \Bigg),\\
\end{split}
\end{equation}
where $P(S^{(1:\kappa-1,L)})=\{\emptyset,\{S_{(1,L)}\},\{S_{(2,L)}\},\ldots,\{S_{(\kappa-1,L)}\},\{S_{(1,L)}, S_{(2,L)}\}, \{S_{(1,L)}, S_{(3,L)}\},\ldots \}$ is the power set of $S^{(1:\kappa-1,L)}=\{S_{(1,L)},S_{(2,L)},\ldots, S_{(\kappa-1,L)}\}$. 
We also define $\prod_{h=\emptyset} (\cdots) \equiv 1$ and $\mathrm{Tr}_{\emptyset}[\rho_{0}]\equiv \rho_{0}$.

Here $\mathrm{Tr}[\mathrm{Tr}_{\bar{S}_{k}}[\rho_{0}^{(1,L)}] \mathrm{Tr}_{\bar{S}_{k}}[\rho_{0}^{(1,L)}]]$ and $\mathrm{Tr}[\mathrm{Tr}_{\overline{S_{k} \cup S_{(\kappa,L)}}}[\rho_{0}^{(1,L)}] \mathrm{Tr}_{\overline{S_{k} \cup S_{(\kappa,L)}}}[\rho_{0}^{(1,L)}]]$ are regarded as the purity of the quantum state $\rho_{0}^{(1,L)}$ which is partially traced over $\bar{S}_{k}$ and $\overline{S_{k} \cup S_{(\kappa,L)}}$, respectively. 
We remind that $\rho_{0}^{(1,L)}$ is the quantum state obtained by applying the ALA up to $L-1$ layer to the initial state, i.e., $\rho_{0}^{(1,L)}=(\prod_{d=1}^{L-1} V_{d}(\bm{x},\bm{\theta})) \rho_{0} (\prod_{d=1}^{L-1} V_{d}^{\dagger}(\bm{x},\bm{\theta}))$. 
Hence, due to the inequality of the purity, i.e., $1/d\le\mathrm{Tr}\left[\rho^2\right]\le1$ with the $d$-dimentional quantum state $\rho$, we have 

\begin{equation}
\begin{split}
\label{eq:qk_var_middle1}
   & \braket{k_{Q}^{2}}_{U(\bm{x},\bm{\theta})} \le \frac{1}{\left(2^{2m}-1\right)^{\kappa}}\times \\
   &  \sum_{\bm{p},\bm{q},\bm{p'},\bm{q'}}  \sum_{S_k\in P(S^{(1:\kappa-1,L)})}  \prod_{h\in \bar{S}_{k}\cap S^{(1:\kappa-1,L)}} \left(\delta_{(\bm{p}\bm{q})_h}\delta_{(\bm{p'}\bm{q'})_h}-\frac{1}{2^{m}}\delta_{(\bm{p}\bm{q'})_h}\delta_{(\bm{p'}\bm{q})_h}\right) \prod_{h\in S_{k}} \left(\delta_{(\bm{p}\bm{q'})_h}\delta_{(\bm{p'}\bm{q})_h}-\frac{1}{2^{m}}\delta_{(\bm{p}\bm{q})_h}\delta_{(\bm{p'}\bm{q'})_h}\right)  \\
   & \qquad \qquad \quad \times \Bigg( \left(  \mathrm{Tr}\left[\rho_{\bm{x'},\bm{\theta},\bm{p}\bm{q}}\right]\mathrm{Tr}\left[\rho_{\bm{x'},\bm{\theta},\bm{p'}\bm{q'}}\right]-\frac{1}{2^{m}}\mathrm{Tr}\left[\rho_{\bm{x'},\bm{\theta},\bm{p}\bm{q}}\rho_{\bm{x'},\theta,\bm{p'}\bm{q'}}\right]\right)\\
   & \qquad \qquad \qquad \qquad \qquad \qquad +  \left(\mathrm{Tr}\left[\rho_{\bm{x'},\bm{\theta},\bm{p}\bm{q}}\rho_{x',\theta,\bm{p'}\bm{q'}}\right]-\frac{1}{2^{m}}\mathrm{Tr}\left[\rho_{\bm{x'},\bm{\theta},\bm{p}\bm{q}}\right]\mathrm{Tr}\left[\rho_{\bm{x'},\bm{\theta},\bm{p'}\bm{q'}}\right]\right) \Bigg).\\
\end{split}
\end{equation}

Further, using $\rho_{\bm{x'},\bm{\theta},\bm{p}\bm{q}}=\mathrm{Tr}_{\bar{S}_{(\kappa,L)}}\left[\left(\ket{\bm{q}}\bra{\bm{p}}\otimes \mathbb{I}_{S_{(\kappa,L)}}\right)\rho_{\bm{x'},\bm{\theta}}\right]$ and the Kronecker delta regarding bit-strings $\bm{p,q,p',q'}$, we can get the following equality.
\begin{equation}
\begin{split}
   &\sum_{\bm{p},\bm{q},\bm{p'},\bm{q'}}   \mathrm{Tr}\left[\rho_{\bm{x'},\bm{\theta},\bm{p}\bm{q}}\right]\mathrm{Tr}\left[\rho_{\bm{x'},\bm{\theta},\bm{p'}\bm{q'}}\right] \delta_{(\bm{p}\bm{q})_{S_{k}}}\delta_{(\bm{p'}\bm{q'})_{S_{k}}}\delta_{(\bm{p}\bm{q'})_{\bar{S}_{k}}}\delta_{(\bm{p'}\bm{q})_{\bar{S}_{k}}} \\
   &=\sum_{\bm{p},\bm{q},\bm{p'},\bm{q'}}   \mathrm{Tr}\left[\left(\ket{\bm{q}}\bra{\bm{p}}\otimes \mathbb{I}_{S_{(\kappa,L)}}\right)\rho_{\bm{x'},\bm{\theta}}\right]\mathrm{Tr}\left[\left(\ket{\bm{q'}}\bra{\bm{p'}}\otimes \mathbb{I}_{S_{(\kappa,L)}}\right)\rho_{\bm{x'},\bm{\theta}}\right] \delta_{(\bm{p}\bm{q})_{S_{k}}}\delta_{(\bm{p'}\bm{q'})_{S_{k}}}\delta_{(\bm{p}\bm{q'})_{\bar{S}_{k}}}\delta_{(\bm{p'}\bm{q})_{\bar{S}_{k}}} \\
   & = \mathrm{Tr}\left[\mathrm{Tr}_{S_k \cup S_{(\kappa, L)}}\left[\rho_{\bm{x'},\bm{\theta}}\right] \mathrm{Tr}_{S_k \cup S_{(\kappa, L)}}\left[\rho_{\bm{x'},\bm{\theta}}\right] \right],
\end{split}
\end{equation}
\begin{equation}
\begin{split}
   &\sum_{\bm{p},\bm{q},\bm{p'},\bm{q'}}   \mathrm{Tr}\left[\rho_{\bm{x'},\bm{\theta},\bm{p}\bm{q}}\rho_{\bm{x'},\bm{\theta},\bm{p'}\bm{q'}}\right] \delta_{(\bm{p}\bm{q})_{S_{k}}}\delta_{(\bm{p'}\bm{q'})_{S_{k}}}\delta_{(\bm{p}\bm{q'})_{\bar{S}_{k}}}\delta_{(\bm{p'}\bm{q})_{\bar{S}_{k}}} \\
   &= \sum_{\bm{p},\bm{q},\bm{p'},\bm{q'}}   \mathrm{Tr}\left[\mathrm{Tr}_{\bar{S}_{(\kappa,L)}}\left[\left(\ket{\bm{q}}\bra{\bm{p}}\otimes \mathbb{I}_{S_{(\kappa,L)}}\right)\rho_{\bm{x'},\bm{\theta}}\right]\mathrm{Tr}_{\bar{S}_{(\kappa,L)}}\left[\left(\ket{\bm{q'}}\bra{\bm{p'}}\otimes \mathbb{I}_{S_{(\kappa,L)}}\right)\rho_{\bm{x'},\bm{\theta}}\right]\right] \delta_{(\bm{p}\bm{q})_{S_{k}}}\delta_{(\bm{p'}\bm{q'})_{S_{k}}}\delta_{(\bm{p}\bm{q'})_{\bar{S}_{k}}}\delta_{(\bm{p'}\bm{q})_{\bar{S}_{k}}} \\
   & = \mathrm{Tr}\left[\mathrm{Tr}_{S_k}\left[\rho_{\bm{x'},\bm{\theta}}\right] \mathrm{Tr}_{S_k}\left[\rho_{\bm{x'},\bm{\theta}}\right] \right].
\end{split}
\end{equation}
This means that Eq.~\eqref{eq:qk_var_middle1} can also be represented using purity of quantum states.
Therefore we have
\begin{equation}
\begin{split}
   \braket{k_{Q}^{2}}_{U(\bm{x},\bm{\theta})} \le \frac{2^{\kappa}}{\left(2^{2m}-1\right)^{\kappa}}, \\
\end{split}
\end{equation}
where we use $\mathrm{Tr}[\rho^2]\le1$.
Also we assume $2^m\gg1$ here.
Thus, the upper bound of the fidelity-based QK using the ALA is described as
\begin{equation}
    Var\left[k_{Q} \right] \le \frac{2^{\kappa}}{\left(2^{2m}-1\right)^{\kappa}}-\frac{1}{2^{2n}} \approx \frac{1}{2^{n\left(2-\frac{1}{m}\right)}}.
\end{equation}
This result is valid for the case where a mixed state is used as the initial state, since the upper bound is derived by the purity of quantum states.

\section{Proof of Theorem}
Here, we derive the expectation and the variance of the quantum Fisher kernel (QFK) based on the anti-symmetric logarithmic derivative (ALD), termed as the ALDQFK.
The ALD is the quantity that can be solved by the equation,
\begin{equation}
\label{eq:ald}
    \partial_{\theta_l} \rho_{\bm{x},\bm{\theta}} = \frac{1}{2}\left(\rho_{\bm{x},\bm{\theta}}L_{\bm{x},\theta_{l}}^{A}-L_{\bm{x},\theta_{l}}^{A}\rho_{\bm{x},\bm{\theta}}\right),
\end{equation}
where $\partial_{\theta_l}\equiv\partial/\partial \theta_l $ and $\rho_{\bm{x},\bm{\theta}}=U(\bm{x},\bm{\theta})\rho_0 U^{\dagger}(\bm{x},\bm{\theta})$.
While the ALD cannot be determined uniquely, one solution of the equation for unitary process can be obtained;
\begin{equation}
\label{eq:ald_unitary_process}
    L_{\bm{x},\theta_{l}}^{A} = i \left(B_{\bm{x},\theta_l} - \mathrm{Tr}\left[\rho_{\bm{x},\bm{\theta}}B_{\bm{x},\theta_l}\right] \right)
\end{equation}
with $B_{\bm{x},\theta_l}=2i(\partial_{\theta_l}U(\bm{x},\bm{\theta}))U^{\dagger}(\bm{x},\bm{\theta})$.
Then, using the ALD, the ALDQFK can be defined as follows;
\begin{equation}
\label{eq:def_aldqfk}
    k_{QF}^{A}(\bm{x},\bm{x'}) =  -\frac{1}{2}\sum_{i,j} \mathcal{F}_{A,i,j}^{-1} \mathrm{Tr}\left[ \rho_{0} \left\{L_{\bm{x},\theta_i}^{A,eff}, L_{\bm{x'},\theta_j}^{A,eff}\right\} \right],
\end{equation}
where $\{\cdot,\cdot\}$ is the anti-commutator, $\mathcal{F}_{A}$ is the ALD-based quantum Fisher information matrix (QFIM) and $L_{\bm{x},\theta_i}^{A,eff}=U^{\dagger}
(\bm{x},\bm{\theta})L_{\bm{x},\theta_j}^{A}U(\bm{x},\bm{\theta}) $ is the effective ALD operator.
Here, we exploit the form of the ALD-based QFIM under unitary process; 
\begin{equation}
\begin{split}
    \langle L_{\bm{x},\theta_i}^{A},L_{\bm{x},\theta_j}^{A}\rangle_{\rho_{\bm{x},\bm{\theta}}}
    &=-\frac{1}{2} \mathrm{Tr}[\rho_{\bm{x},\bm{\theta}}\{L_{\bm{x},\theta_i}^{A}, L_{\bm{x},\theta_j}^{A}\}]\\
    &=-\frac{1}{2} \mathrm{Tr}[\rho_{0}\{L_{\bm{x},\theta_i}^{A, eff}, L_{\bm{x},\theta_j}^{A,eff}\}].
\end{split}
\end{equation}

In this work, we set the QFIM as the identity matrix, i.e.,  $\mathcal{F}=\mathbb{I}$, because the QFIM is computationally demanding and has been suggested to be less significant in Ref. \cite{jaakkola1998exploiting}.
In addition, the term $\mathrm{Tr}[\rho_{\bm{x},\bm{\theta}}B_{\bm{x},\theta_l}]$ in the ALD of Eq.\eqref{eq:ald_unitary_process} is ignored so that the ALDQFK with the same inputs $k_{QF}^{A}(\bm{x},\bm{x})$ is constant for any $\bm{x}$.
Also, since we assume each parameter $\theta$ is in the angle of the rotation gate, $\exp(-i\theta \sigma/2)$ with $\sigma\in\{X,Y,Z\}$, we can rewrite the ALDQFK as
\begin{equation}
\begin{split}
\label{eq:aldqfk_rewritten}
    k_{QF}^{A}(\bm{x},\bm{x'})  &=  -\frac{1}{2}\sum_{i} \mathrm{Tr}\left[ \rho_{0} \left\{L_{\bm{x},\theta_i}^{A,eff}, L_{\bm{x'},\theta_i}^{A,eff}\right\} \right]\\
    &= \frac{1}{2}\sum_{i} \mathrm{Tr}\left[ \rho_{0} \left\{\tilde{B}_{\bm{x},\theta_{i}} ,\tilde{B}_{\bm{x'},\theta_{i}}\right\} \right],\\
\end{split}
\end{equation}
where $\tilde{B}_{\bm{x},\theta_{i}}=U_{1:i}^{\dagger}(\bm{x},\bm{\theta})B_{\theta_{i}}U_{1:i}(\bm{x},\bm{\theta})$ with the Pauli operator $B_{\theta_{l}}$ of the rotation gate containing the $l$-th parameter.
Here, $U_{i:j}(\bm{x},\bm{\theta})$ denotes
a bunch of unitary gates from $U_{i}(\bm{x},\theta_{i})$ to $U_{j}(\bm{x},\theta_{j})$, assuming the quantum circuit can be decomposed as $U(\bm{x},\bm{\theta})=U_{D}(\bm{x},\theta_{D})\ldots U_{2}(\bm{x},\theta_{2}) U_{1}(\bm{x},\theta_{1})$.

In the analysis shown below, we only focus on $\mathrm{Tr}[ \rho_{0} \{\tilde{B}_{\bm{x},\theta_{i}} ,\tilde{B}_{\bm{x'},\theta_{j}}\} ]/2$, since diagnosing the quantity is enough to see the tendency of the vanishing similarity issue in the ALDQFK.
We here note that the initial state $\rho_0$ is a pure state.
Additionally, regardless of the position of the gate, $i\in\{1,\ldots,D\}$ , we here assume $U_{1:i}(\bm{x},\bm{\theta})$ is a $t$-design for the case where the random quantum circuits acting on all $n$ qubits are used. 
As for the case where the ALA is used, assuming the $i$-th parameter $\theta_{i}$ is located in the $k$-th unitary block in the $d$-th layer of the circuits, $W_{k,d}(\bm{x},\bm{\theta}_{k,d})$, we decompose the circuit as
\begin{equation}
    U_{1:i}(\bm{x},\bm{\theta}) = \tilde{W}_{k,d}(\bm{x},\theta_{i})V_{r}(\bm{x},\bm{\theta}),
\end{equation}
where $\tilde{W}_{k,d}(\bm{x},\theta_{i})$ is the all gates up to the one containing $i$-th parameter within $W_{k,d}(\bm{x},\bm{\theta}_{k,d})$, and $V_{r}(\bm{x},\bm{\theta})$ is the all unitary blocks in the light-cone of $W_{k,d}(\bm{x},\bm{\theta}_{k,d})$, as in Figure 3 of the main test.
Then we assume not only the unitary blocks in $V_{r}(\bm{x},\bm{\theta})$ but also $\tilde{W}_{k,d}(\bm{x},\theta_{i})$ for arbitrary $i,k$ and $d$ are $t$-designs.

\subsection{Case (1): The random quantum circuit acting on all $n$ qubits for the ALDQFK}

Here, considering the random quantum circuits acting on all $n$ qubits, we calculate the expectation $\braket{k_{QF}^{A}}$ and the variance $Var\left[k_{QF}^{A}\right]$ of the ALDQFK in Eq.\eqref{eq:def_aldqfk}.
In particular, we focus on $\mathrm{Tr}[ \rho_{0} \{\tilde{B}_{\bm{x},\theta_{i}} ,\tilde{B}_{\bm{x'},\theta_{j}}\} ]/2$ in the ALDQFK, as we stated.

At first, we work on the expectation of the ALDQFK.
In this case, we assume that either $U_{1:i}(\bm{x},\bm{\theta})$ or $U_{1:i}(\bm{x’},\bm{\theta})$ is a $t$-design with $t\ge1$.
However, due to the symmetry of $\tilde{B}_{\bm{x},\theta_{i}}$ and $\tilde{B}_{\bm{x'},\theta_{i}}$, we consider only $U_{1:i}(\bm{x},\bm{\theta})$ is a $t$-design here.
Then the expectation of the ALDQFK is calculated as follows.

\begin{equation}
\begin{split}
    \braket{k_{QF}^{A}}_{U_{1:i}(\bm{x},\bm{\theta})}&=\frac{1}{2}\left\langle\mathrm{Tr}\left[ \rho_{0} U_{1:i}^{\dagger}(\bm{x},\bm{\theta})B_{\theta_{i}}U_{1:i}(\bm{x},\bm{\theta}) \tilde{B}_{\bm{x'},\theta_{i}}  \right]\right\rangle_{U_{1:i}(\bm{x},\bm{\theta})} +\frac{1}{2} \left\langle\mathrm{Tr}\left[ \rho_{0}  \tilde{B}_{\bm{x'},\theta_{i}} U_{1:i}^{\dagger}(\bm{x},\bm{\theta})B_{\theta_{i}}U_{1:i}(\bm{x},\bm{\theta}) \right]\right\rangle_{U_{1:i}(\bm{x},\bm{\theta})} \\
    &=\frac{1}{2}\left\langle\mathrm{Tr}\left[ U_{1:i}^{\dagger}(\bm{x},\bm{\theta})B_{\theta_{i}}U_{1:i}(\bm{x},\bm{\theta}) \tilde{B}_{\bm{x'},\theta_{i}} \rho_{0} \right]\right\rangle_{U_{1:i}(\bm{x},\bm{\theta})} +\frac{1}{2} \left\langle\mathrm{Tr}\left[ U_{1:i}^{\dagger}(\bm{x},\bm{\theta})B_{\theta_{i}}U_{1:i}(\bm{x},\bm{\theta})\rho_{0} \tilde{B}_{\bm{x'},\theta_{i}}  \right]\right\rangle_{U_{1:i}(\bm{x},\bm{\theta})} \\
    &= \frac{1}{2\cdot2^{n}} \mathrm{Tr} \left[B_{\theta_{i}}\right] \mathrm{Tr} \left[\tilde{B}_{\bm{x'},\theta_{i}} \rho_{0}\right] + \frac{1}{2\cdot2^{n}} \mathrm{Tr} \left[B_{\theta_{i}}\right] \mathrm{Tr} \left[\rho_{0} \tilde{B}_{\bm{x'},\theta_{i}} \right]\\
    &= 0,\\
\end{split}
\end{equation}
where Lemma \ref{lem1} and the traceless property of the Pauli operators are utilized.

Then we calculate the variance.
The variance $Var\left[k_{QF}^{A}\right]$ can be obtained by $Var\left[k_{QF}^{A}\right]=\braket{{k_{QF}^{A}}^{2}}-\braket{k_{QF}^{A}}^2$.
Since $\braket{k_{QF}^{A}}=0$, all we need to do is calculate $Var\left[k_{QF}^{A}\right]=\braket{{k_{QF}^{A}}^{2}}$.
Here, we assume that $U_{1:i}(\bm{x},\bm{\theta})$ and $U_{1:i}(\bm{x’},\bm{\theta})$ are $t$-designs with $t\ge2$, and work on the integration over $U_{1:i}(\bm{x},\bm{\theta})$ first.
The expectation $\braket{{k_{QF}^{A}}^{2}}_{ U_{1:i}(\bm{x},\bm{\theta})}$ can be expressed as 

\begin{equation}
\begin{split}
\label{eq:var_aldqfk_rqc_origin}
    \braket{{k_{QF}^{A}}^{2}}_{U_{1:i}(\bm{x},\bm{\theta})}&=\frac{1}{4}\left\langle\left(\mathrm{Tr}\left[ \rho_{0} U_{1:i}^{\dagger}(\bm{x},\bm{\theta})B_{\theta_{i}}U_{1:i}(\bm{x},\bm{\theta}) \tilde{B}_{\bm{x'},\theta_{i}}  \right] +  \mathrm{Tr}\left[ U_{1:i}^{\dagger}(\bm{x},\bm{\theta})B_{\theta_{i}}U_{1:i}(\bm{x},\bm{\theta})\rho_{0} \tilde{B}_{\bm{x'},\theta_{i}}  \right] \right)^2\right\rangle_{U_{1:i}(\bm{x},\bm{\theta})}\\
    &= \frac{1}{4}\left\langle\mathrm{Tr}\left[ \rho_{0} U_{1:i}^{\dagger}(\bm{x},\bm{\theta})B_{\theta_{i}}U_{1:i}(\bm{x},\bm{\theta}) \tilde{B}_{\bm{x'},\theta_{i}}  \right]\mathrm{Tr}\left[ \rho_{0} U_{1:i}^{\dagger}(\bm{x},\bm{\theta})B_{\theta_{i}}U_{1:i}(\bm{x},\bm{\theta}) \tilde{B}_{\bm{x'},\theta_{i}}  \right]\right\rangle_{U_{1:i}(\bm{x},\bm{\theta})}\\
    &\quad +\frac{1}{2} \left\langle\mathrm{Tr}\left[ \rho_{0} U_{1:i}^{\dagger}(\bm{x},\bm{\theta})B_{\theta_{i}}U_{1:i}(\bm{x},\bm{\theta}) \tilde{B}_{\bm{x'},\theta_{i}}  \right] \mathrm{Tr}\left[ U_{1:i}^{\dagger}(\bm{x},\bm{\theta})B_{\theta_{i}}U_{1:i}(\bm{x},\bm{\theta})\rho_{0} \tilde{B}_{\bm{x'},\theta_{i}}  \right]\right\rangle_{U_{1:i}(\bm{x},\bm{\theta})}\\
    &\quad +\frac{1}{4} \left\langle\mathrm{Tr}\left[ U_{1:i}^{\dagger}(\bm{x},\bm{\theta})B_{\theta_{i}}U_{1:i}(\bm{x},\bm{\theta})\rho_{0} \tilde{B}_{\bm{x'},\theta_{i}}  \right]\mathrm{Tr}\left[ U_{1:i}^{\dagger}(\bm{x},\bm{\theta})B_{\theta_{i}}U_{1:i}(\bm{x},\bm{\theta})\rho_{0} \tilde{B}_{\bm{x'},\theta_{i}}  \right]\right\rangle_{U_{1:i}(\bm{x},\bm{\theta})}\\
    &= Var_{r,1} + Var_{r,2} + Var_{r,3},\\ 
\end{split}
\end{equation}
where $Var_{r,i}$ represents the $i$-th term of the right-hand side of the second equality.
Thus, we calculate these terms to get the variance of the ALDQFK.

The first term can be obtained as
\begin{equation}
\begin{split}
    Var_{r,1} &= \frac{1}{4}\left\langle\mathrm{Tr}\left[ \rho_{0} U_{1:i}^{\dagger}(\bm{x},\bm{\theta})B_{\theta_{i}}U_{1:i}(\bm{x},\bm{\theta}) \tilde{B}_{\bm{x'},\theta_{i}}  \right]\mathrm{Tr}\left[ \rho_{0} U_{1:i}^{\dagger}(\bm{x},\bm{\theta})B_{\theta_{i}}U_{1:i}(\bm{x},\bm{\theta}) \tilde{B}_{\bm{x'},\theta_{i}}  \right]\right\rangle_{U_{1:i}(\bm{x},\bm{\theta})}\\
    &= \frac{1}{4\cdot\left(2^{2n}-1\right)}\left(\mathrm{Tr}\left[B_{\theta_{i}}\right]\mathrm{Tr}\left[\tilde{B}_{\bm{x'},\theta_{i}} \rho_{0}\right]\mathrm{Tr}\left[B_{\theta_{i}}\right]\mathrm{Tr}\left[\tilde{B}_{\bm{x'},\theta_{i}} \rho_{0}\right]+\mathrm{Tr}\left[B_{\theta_{i}}^{2}\right]\mathrm{Tr}\left[\left(\tilde{B}_{\bm{x'},\theta_{i}} \rho_{0}\right)^2\right]\right)\\
    & \quad -\frac{1}{4\cdot2^{n}\left(2^{2n}-1\right)}\left(\mathrm{Tr}\left[B_{\theta_{i}}\right]\mathrm{Tr}\left[B_{\theta_{i}}\right]\mathrm{Tr}\left[\left(\tilde{B}_{\bm{x'},\theta_{i}} \rho_{0}\right)^2\right]+\mathrm{Tr}\left[B_{\theta_{i}}^{2}\right]\mathrm{Tr}\left[\tilde{B}_{\bm{x'},\theta_{i}} \rho_{0}\right]\mathrm{Tr}\left[\tilde{B}_{\bm{x'},\theta_{i}} \rho_{0}\right]\right)\\
    &= \frac{2^{n}}{4\cdot\left(2^{2n}-1\right)}\left(\mathrm{Tr}\left[\left(\tilde{B}_{\bm{x'},\theta_{i}} \rho_{0}\right)^2\right]-\frac{1}{2^n}\left(\mathrm{Tr}\left[\tilde{B}_{\bm{x'},\theta_{i}} \rho_{0}\right]\right)^2\right) \\
    &= \frac{2^{n}}{4\cdot\left(2^{2n}-1\right)}\left(1-\frac{1}{2^n}\right)\left(\mathrm{Tr}\left[\tilde{B}_{\bm{x'},\theta_{i}} \rho_{0}\right]\right)^2. \\
\end{split}
\end{equation}
where we exploit Lemma \ref{lem3} and the properties of the Pauli operators, $\mathrm{Tr}[B]=0$ and $\mathrm{Tr}[B^2]=2^{n}$. 
Also, due to the fact that the initial state is a pure state, we use the equality $\mathrm{Tr}[(\tilde{B}_{\bm{x'},\theta_{i}} \rho_{0})^2]=(\mathrm{Tr}[\tilde{B}_{\bm{x'},\theta_{i}} \rho_{0}])^2$.

Similarly, the second and the third terms are calculated in the following way.
\begin{equation}
\begin{split}
    Var_{r,2} &= \frac{1}{2\cdot\left(2^{2n}-1\right)}\left(\mathrm{Tr}\left[B_{\theta_{i}}\right]\mathrm{Tr}\left[\tilde{B}_{\bm{x'},\theta_{i}} \rho_{0}\right]\mathrm{Tr}\left[B_{\theta_{i}}\right]\mathrm{Tr}\left[\rho_{0}\tilde{B}_{\bm{x'},\theta_{i}} \right]+\mathrm{Tr}\left[B_{\theta_{i}}^{2}\right]\mathrm{Tr}\left[\tilde{B}_{\bm{x'},\theta_{i}}\rho_{0}\rho_{0}\tilde{B}_{\bm{x'},\theta_{i}}\right]\right)\\
    & \quad -\frac{1}{2\cdot2^{n}\left(2^{2n}-1\right)}\left(\mathrm{Tr}\left[B_{\theta_{i}}\right]\mathrm{Tr}\left[B_{\theta_{i}}\right]\mathrm{Tr}\left[\tilde{B}_{\bm{x'},\theta_{i}}^2\rho_{0}\right]+\mathrm{Tr}\left[B_{\theta_{i}}^{2}\right]\mathrm{Tr}\left[\rho_{0}\tilde{B}_{\bm{x'},\theta_{i}} \right]\mathrm{Tr}\left[\tilde{B}_{\bm{x'},\theta_{i}} \rho_{0}\right]\right)\\
    &= \frac{2^{n}}{2\cdot\left(2^{2n}-1\right)}\left(\mathrm{Tr}\left[\tilde{B}_{\bm{x'},\theta_{i}}^2\rho_{0}\right]-\frac{1}{2^n}\left(\mathrm{Tr}\left[\tilde{B}_{\bm{x'},\theta_{i}} \rho_{0}\right]\right)^2\right), \\
\end{split}
\end{equation}
\begin{equation}
\begin{split}
    Var_{r,3} &= \frac{1}{4\cdot\left(2^{2n}-1\right)}\left(\mathrm{Tr}\left[B_{\theta_{i}}\right]\mathrm{Tr}\left[\tilde{B}_{\bm{x'},\theta_{i}} \rho_{0}\right]\mathrm{Tr}\left[B_{\theta_{i}}\right]\mathrm{Tr}\left[\tilde{B}_{\bm{x'},\theta_{i}} \rho_{0}\right]+\mathrm{Tr}\left[B_{\theta_{i}}^{2}\right]\mathrm{Tr}\left[\left(\tilde{B}_{\bm{x'},\theta_{i}} \rho_{0}\right)^2\right]\right)\\
    & \quad -\frac{1}{4\cdot2^{n}\left(2^{2n}-1\right)}\left(\mathrm{Tr}\left[B_{\theta_{i}}\right]\mathrm{Tr}\left[B_{\theta_{i}}\right]\mathrm{Tr}\left[\left(\tilde{B}_{\bm{x'},\theta_{i}} \rho_{0}\right)^2\right]+\mathrm{Tr}\left[B_{\theta_{i}}^{2}\right]\mathrm{Tr}\left[\tilde{B}_{\bm{x'},\theta_{i}}\rho_{0} \right]\mathrm{Tr}\left[\tilde{B}_{\bm{x'},\theta_{i}} \rho_{0}\right]\right)\\
    &= \frac{2^{n}}{4\cdot\left(2^{2n}-1\right)}\left(\mathrm{Tr}\left[\left(\tilde{B}_{\bm{x'},\theta_{i}} \rho_{0}\right)^2\right]-\frac{1}{2^n}\left(\mathrm{Tr}\left[\tilde{B}_{\bm{x'},\theta_{i}} \rho_{0}\right]\right)^2\right) \\
    &= \frac{2^{n}}{4\cdot\left(2^{2n}-1\right)}\left(1-\frac{1}{2^n}\right)\left(\mathrm{Tr}\left[\tilde{B}_{\bm{x'},\theta_{i}} \rho_{0}\right]\right)^2. \\
\end{split}
\end{equation}

Thus we have
\begin{equation}
\label{eq:pf_2_1_mid_v}
\begin{split}
    \braket{{k_{QF}^{A}}^{2}}_{U_{1:i}(\bm{x},\bm{\theta})} &= Var_{r,1} + Var_{r,2} + Var_{r,3}  \\
    &= \frac{2^{n}}{2^{2n}-1}\cdot \frac{1}{2}\left(\left(1-\frac{1}{2^n}\right)\left(\mathrm{Tr}\left[\tilde{B}_{\bm{x'},\theta_{i}} \rho_{0}\right]\right)^2 + \left(\mathrm{Tr}\left[\tilde{B}_{\bm{x'},\theta_{i}}^2\rho_{0}\right]-\frac{1}{2^n}\left(\mathrm{Tr}\left[\tilde{B}_{\bm{x'},\theta_{i}} \rho_{0}\right]\right)^2\right) \right).\\
\end{split}
\end{equation}

Next we integrate the quantity over $U_{1:i}(\bm{x'},\bm{\theta})$.
Since $U_{1:i}(\bm{x'},\bm{\theta})$ is involved in $\mathrm{Tr}[\tilde{B}_{\bm{x'},\theta_{i}}^2\rho_{0}]$ and $(\mathrm{Tr}[\tilde{B}_{\bm{x'},\theta_{i}} \rho_{0}])^2$ in Eq.~\eqref{eq:pf_2_1_mid_v}, we consider these terms.
The expectation of these terms are calculated as
\begin{equation}
\begin{split}
   \left\langle \mathrm{Tr}\left[\tilde{B}_{\bm{x'},\theta_{i}}^2\rho_{0}\right]\right\rangle_{U_{1:i}(\bm{x'},\bm{\theta})}&=\left\langle\mathrm{Tr}\left[U_{1:i}^{\dagger}(\bm{x'},\bm{\theta})B_{\theta_{i}}^{2}U_{1:i}(\bm{x'},\bm{\theta})\rho_{0}\right]\right\rangle_{U_{1:i}(\bm{x'},\bm{\theta})}\\
    &=\frac{1}{2^{n}} \mathrm{Tr}\left[B_{\theta_{i}}^{2}\right]\mathrm{Tr}\left[\rho_{0}\right] \\
    &= 1,\\
\end{split}
\end{equation}
\begin{equation}
\begin{split}
    \left\langle\left(\mathrm{Tr}\left[\tilde{B}_{\bm{x'},\theta_{i}} \rho_{0}\right]\right)^2 \right\rangle_{U_{1:i}(\bm{x'},\bm{\theta})}&=\left\langle \mathrm{Tr}\left[U_{1:i}^{\dagger}(\bm{x'},\bm{\theta})B_{\theta_{i}}U_{1:i}(\bm{x'},\bm{\theta})\rho_{0}\right]\mathrm{Tr}\left[U_{1:i}^{\dagger}(\bm{x'},\bm{\theta})B_{\theta_{i}}U_{1:i}(\bm{x'},\bm{\theta})\rho_{0}\right] \right\rangle_{U_{1:i}(\bm{x'},\bm{\theta})}\\
    &= \frac{1}{2^{2n}-1}\left(\mathrm{Tr}\left[B_{\theta_{i}}\right]\mathrm{Tr}\left[\rho_{0}\right]\mathrm{Tr}\left[B_{\theta_{i}}\right]\mathrm{Tr}\left[\rho_{0}\right]+\mathrm{Tr}\left[B_{\theta_{i}}^2\right]\mathrm{Tr}\left[\rho_{0}^2\right]\right)\\
    & \quad -\frac{1}{2^{n}\left(2^{2n}-1\right)}\left(\mathrm{Tr}\left[B_{\theta_{i}}\right]\mathrm{Tr}\left[B_{\theta_{i}}\right]\mathrm{Tr}\left[\rho_{0}^2\right]+\mathrm{Tr}\left[B_{\theta_{i}}^2\right]\mathrm{Tr}\left[\rho_{0}\right]\mathrm{Tr}\left[\rho_{0}\right]\right)\\
    &= \frac{1}{2^{2n}-1}\left(2^{n}-1\right)\\
    &= \frac{1}{2^{n}+1}.
\end{split}
\end{equation}
Here we utilize Lemmas \ref{lem1} and \ref{lem3} and the property of the Pauli operators and the pure state. 
Therefore, substituting the terms into Eq.~\eqref{eq:pf_2_1_mid_v}, we have
\begin{equation}
\begin{split}
    \braket{{k_{QF}^{A}}^{2}} &= \frac{2^{n}}{2\left(2^{2n}-1\right)}\left(1+\frac{2^n-2}{2^n\left(2^n+1\right)} \right) \approx \frac{1}{2^{n+1}}.\\ 
\end{split}
\end{equation}

\subsection{Case (2): The ALA for the ALDQFK}
We here calculate the expectation and the lower bound of the variance for the ALDQFK using the ALA.
Specifically, we assume both $\tilde{W}_{k,d}(\bm{x},\theta_{i})$ and $\tilde{W}_{k,d}(\bm{x'},\theta_{i})$ as well as all unitary blocks are $t$-designs.
We remind that $\tilde{W}_{k,d}(\bm{x},\theta_{i})$ represents the all gates up to the one containing $i$-th parameter within $W_{k,d}(\bm{x},\bm{\theta}_{k,d})$.

Firstly, we calculate the expectation of the ALDQFK.
The expectation for the unitary $U_{1:i}(\bm{x},\bm{\theta})$ can be obtained by integrating the quantity over each unitary block, i.e., $\braket{k_{QF}^{A}}_{U_{1:i}(\bm{x},\bm{\theta})}=\braket{k_{QF}^{A}}_{\tilde{W}_{k,d}(\bm{x},\theta_{i})V_{r}(\bm{x},\bm{\theta})}$.
Hence, we start with the $\tilde{W}_{k,d}(\bm{x},\theta_{i})$.
Then, the expectation can be obtained as follows; 
\begin{equation}
\begin{split}
    \braket{k_{QF}^{A}}_{\tilde{W}_{k,d}(\bm{x},\theta_{i})}&=\frac{1}{2}\left\langle\mathrm{Tr}\left[ \rho_{0} V_{r}^{\dagger}(\bm{x},\bm{\theta})\tilde{W}_{k,d}^{\dagger}(\bm{x},\theta_{i})B_{\theta_{i}}\tilde{W}_{k,d}(\bm{x},\theta_{i})V_{r}(\bm{x},\bm{\theta}) \tilde{B}_{\bm{x'},\theta_{i}}  \right]\right\rangle_{\tilde{W}_{k,d}(\bm{x},\theta_{i})} \\ 
    & \qquad +\frac{1}{2} \left\langle\mathrm{Tr}\left[ \rho_{0}  \tilde{B}_{\bm{x'},\theta_{i}} V_{r}^{\dagger}(\bm{x},\bm{\theta})\tilde{W}_{k,d}^{\dagger}(\bm{x},\theta_{i})B_{\theta_{i}}\tilde{W}_{k,d}(\bm{x},\theta_{i})V_{r}(\bm{x},\bm{\theta}) \right]\right\rangle_{\tilde{W}_{k,d}(\bm{x},\theta_{i})} \\
    &=\frac{1}{2}\left\langle\mathrm{Tr}\left[ \tilde{W}_{k,d}(\bm{x},\theta_{i})V_{r}(\bm{x},\bm{\theta}) \tilde{B}_{\bm{x'},\theta_{i}}  \rho_{0} V_{r}^{\dagger}(\bm{x},\bm{\theta})\tilde{W}_{k,d}^{\dagger}(\bm{x},\theta_{i})B_{\theta_{i}}\right]\right\rangle_{\tilde{W}_{k,d}(\bm{x},\theta_{i})} \\ 
    & \qquad +\frac{1}{2} \left\langle\mathrm{Tr}\left[ \tilde{W}_{k,d}(\bm{x},\theta_{i})V_{r}(\bm{x},\bm{\theta}) \rho_{0}  \tilde{B}_{\bm{x'},\theta_{i}} V_{r}^{\dagger}(\bm{x},\bm{\theta})\tilde{W}_{k,d}^{\dagger}(\bm{x},\theta_{i})B_{\theta_{i}}\right]\right\rangle_{\tilde{W}_{k,d}(\bm{x},\theta_{i})} \\
    &= \frac{1}{2\cdot2^{m}} \mathrm{Tr}\left[ \mathrm{Tr}_{S_{(k,d)}} \left[V_{r}(\bm{x},\bm{\theta}) \rho_{0}  \tilde{B}_{\bm{x'},\theta_{i}} V_{r}^{\dagger}(\bm{x},\bm{\theta})\right]  \mathrm{Tr}_{S_{(k,d)}} \left[B_{\theta_{i}}\right] \right] \\
    & \qquad + \frac{1}{2\cdot2^{m}} \mathrm{Tr}\left[ \mathrm{Tr}_{S_{(k,d)}} \left[V_{r}(\bm{x},\bm{\theta}) \rho_{0}  \tilde{B}_{\bm{x'},\theta_{i}} V_{r}^{\dagger}(\bm{x},\bm{\theta}) \right]  \mathrm{Tr}_{S_{(k,d)}} \left[B_{\theta_{i}}\right] \right]\\
    &= 0, \\
\end{split}
\end{equation}
where $\mathrm{Tr}_{S_{(k,d)}}$ represents a partial trace over the space $S_{(k,d)}$ on which $\tilde{W}_{k,d}(\bm{x},\theta_{i})$ acts.
Also we utilize Lemma \ref{lem4} and the traceless property of the Pauli operators, $\mathrm{Tr}[B]=0$, due to the fact that $B_{\theta_{i}}$ acts on $S_{(k,d)}$.
This means that the expectation $\braket{k_{QF}^{A}}$ is zero irrespective of the remaining unitary blocks in $U_{1:i}(\bm{x},\bm{\theta})$, $\tilde{W}_{k,d}(\bm{x},\theta_{i})$ and $U_{1:i}(\bm{x'},\bm{\theta})$.

Next, we calculate the variance.
The variance $Var\left[k_{QF}^{A}\right]$ can be obtained by calculating $\braket{{k_{QF}^{A}}^{2}}$, because $\braket{k_{QF}^{A}}=0$. 
Here, we first focus on the integration over $\tilde{W}_{k,d}(\bm{x},\theta_{i})$.
Then, we have
\begin{equation}
\begin{split}
\label{eq:var_aldqfk_ala_origin}
    \braket{{k_{QF}^{A}}^{2}}_{\tilde{W}_{k,d}(\bm{x},\theta_{i})}&=\frac{1}{4}\biggl\langle \biggl(\mathrm{Tr}\left[ \rho_{0} V_{r}^{\dagger}(\bm{x},\bm{\theta})\tilde{W}_{k,d}^{\dagger}(\bm{x},\theta_{i})B_{\theta_{i}}\tilde{W}_{k,d}(\bm{x},\theta_{i})V_{r}(\bm{x},\bm{\theta}) \tilde{B}_{\bm{x'},\theta_{i}}  \right] \\
    & \qquad \qquad +  \mathrm{Tr}\left[\rho_{0}  \tilde{B}_{\bm{x'},\theta_{i}} V_{r}^{\dagger}(\bm{x},\bm{\theta})\tilde{W}_{k,d}^{\dagger}(\bm{x},\theta_{i})B_{\theta_{i}}\tilde{W}_{k,d}(\bm{x},\theta_{i})V_{r}(\bm{x},\bm{\theta}) \right] \biggr)^2 \biggr\rangle_{\tilde{W}_{k,d}(\bm{x},\theta_{i})} \\
    &= \frac{1}{4} \biggl\langle \mathrm{Tr}\left[ \rho_{0} V_{r}^{\dagger}(\bm{x},\bm{\theta})\tilde{W}_{k,d}^{\dagger}(\bm{x},\theta_{i})B_{\theta_{i}}\tilde{W}_{k,d}(\bm{x},\theta_{i})V_{r}(\bm{x},\bm{\theta}) \tilde{B}_{\bm{x'},\theta_{i}}  \right] \\
    & \qquad \qquad \times  \mathrm{Tr}\left[ \rho_{0} V_{r}^{\dagger}(\bm{x},\bm{\theta})\tilde{W}_{k,d}^{\dagger}(\bm{x},\theta_{i})B_{\theta_{i}}\tilde{W}_{k,d}(\bm{x},\theta_{i})V_{r}(\bm{x},\bm{\theta}) \tilde{B}_{\bm{x'},\theta_{i}}  \right]  \biggr\rangle_{\tilde{W}_{k,d}(\bm{x},\theta_{i})} \\
    & \qquad + \frac{1}{2} \biggl\langle \mathrm{Tr}\left[ \rho_{0} V_{r}^{\dagger}(\bm{x},\bm{\theta})\tilde{W}_{k,d}^{\dagger}(\bm{x},\theta_{i})B_{\theta_{i}}\tilde{W}_{k,d}(\bm{x},\theta_{i})V_{r}(\bm{x},\bm{\theta})\tilde{B}_{\bm{x'},\theta_{i}}  \right] \\
    & \qquad \qquad \times  \mathrm{Tr}\left[\rho_{0}  \tilde{B}_{\bm{x'},\theta_{i}} V_{r}^{\dagger}(\bm{x},\bm{\theta})\tilde{W}_{k,d}^{\dagger}(\bm{x},\theta_{i})B_{\theta_{i}}\tilde{W}_{k,d}(\bm{x},\theta_{i})V_{r}(\bm{x},\bm{\theta}) \right] \biggr\rangle_{\tilde{W}_{k,d}(\bm{x},\theta_{i})} \\
    & \qquad + \frac{1}{4} \biggl\langle \mathrm{Tr}\left[\rho_{0}  \tilde{B}_{\bm{x'},\theta_{i}} V_{r}^{\dagger}(\bm{x},\bm{\theta})\tilde{W}_{k,d}^{\dagger}(\bm{x},\theta_{i})B_{\theta_{i}}\tilde{W}_{k,d}(\bm{x},\theta_{i})V_{r}(\bm{x},\bm{\theta}) \right] \\
    & \qquad \qquad \times  \mathrm{Tr}\left[\rho_{0}  \tilde{B}_{\bm{x'},\theta_{i}} V_{r}^{\dagger}(\bm{x},\bm{\theta})\tilde{W}_{k,d}^{\dagger}(\bm{x},\theta_{i})B_{\theta_{i}}\tilde{W}_{k,d}(\bm{x},\theta_{i})V_{r}(\bm{x},\bm{\theta}) \right] \biggr\rangle_{\tilde{W}_{k,d}(\bm{x},\theta_{i})}. \\
    &= Var_{a,1} + Var_{a,2} + Var_{a,3},\\
\end{split}
\end{equation}
where $Var_{a,i}$ is the $i$-th term of the right-hand side of the second equality.

We start with the integration of the first term and then we have
\begin{equation}
\begin{split}
    & Var_{a,1} \\
    & = \frac{1}{4} \biggl\langle \mathrm{Tr}\left[ \tilde{W}_{k,d}(\bm{x},\theta_{i})\tilde{\rho}_{0,B_{l}}^{(1)}\tilde{W}_{k,d}^{\dagger}(\bm{x},\theta_{i})B_{\theta_{i}}  \right]  \mathrm{Tr}\left[ \tilde{W}_{k,d}(\bm{x},\theta_{i})\tilde{\rho}_{0,B_{l}}^{(1)}\tilde{W}_{k,d}^{\dagger}(\bm{x},\theta_{i})B_{\theta_{i}} \right]  \biggr\rangle_{\tilde{W}_{k,d}(\bm{x},\theta_{i})} \\
    &= \frac{1}{4}\sum_{\bm{p},\bm{q},\bm{p'},\bm{q'}} \biggl\langle \mathrm{Tr}\left[ \tilde{W}_{k,d}(\bm{x},\theta_{i})\tilde{\rho}_{0,B_{l},\bm{q}\bm{p}}^{(1)}\tilde{W}_{k,d}^{\dagger}(\bm{x},\theta_{i})B_{\theta_{i},\bm{p}\bm{q}}  \right] \mathrm{Tr}\left[ \tilde{W}_{k,d}(\bm{x},\theta_{i})\tilde{\rho}_{0,B_{l},\bm{q'}\bm{p'}}^{(1)}\tilde{W}_{k,d}^{\dagger}(\bm{x},\theta_{i})B_{\theta_{i},\bm{p'}\bm{q'}}  \right]\biggr\rangle_{\tilde{W}_{k,d}(\bm{x},\theta_{i})} \\
    &= \frac{1}{4} \sum_{\bm{p},\bm{q},\bm{p'},\bm{q'}} \Biggl( \frac{1}{2^{2m}-1} \left(\mathrm{Tr}\left[\tilde{\rho}_{0,B_{l},\bm{q}\bm{p}}^{(1)}\right]\mathrm{Tr}\left[B_{\theta_{i},\bm{p}\bm{q}}\right]\mathrm{Tr}\left[\tilde{\rho}_{0,B_{l},\bm{q'}\bm{p'}}^{(1)}\right]\mathrm{Tr}\left[B_{\theta_{i},\bm{p'}\bm{q'}}\right]+\mathrm{Tr}\left[\tilde{\rho}_{0,B_{l},\bm{q}\bm{p}}^{(1)}\tilde{\rho}_{0,B_{l},\bm{q'}\bm{p'}}^{(1)}\right]\mathrm{Tr}\left[B_{\theta_{i},\bm{p}\bm{q}}B_{\theta_{i},\bm{p'}\bm{q'}}\right]\right)\\
    & \qquad -\frac{1}{2^{m}\left(2^{2m}-1\right)}\left(\mathrm{Tr}\left[\tilde{\rho}_{0,B_{l},\bm{q}\bm{p}}^{(1)}\right]\mathrm{Tr}\left[\tilde{\rho}_{0,B_{l},\bm{q'}\bm{p'}}^{(1)}\right]\mathrm{Tr}\left[B_{\theta_{i},\bm{p}\bm{q}}B_{\theta_{i},\bm{p'}\bm{q'}}\right]+\mathrm{Tr}\left[\tilde{\rho}_{0,B_{l},\bm{q}\bm{p}}^{(1)}\tilde{\rho}_{0,B_{l},\bm{q'}\bm{p'}}^{(1)}\right]\mathrm{Tr}\left[B_{\theta_{i},\bm{p}\bm{q}}\right]\mathrm{Tr}\left[B_{\theta_{i},\bm{p'}\bm{q'}}\right]\right) \Biggr) \\
    &= \frac{1}{4}\cdot \frac{1}{2^{2m}-1} \sum_{\bm{p},\bm{q},\bm{p'},\bm{q'}} \mathrm{Tr}\left[B_{\theta_{i},\bm{p}\bm{q}}B_{\theta_{i},\bm{p'}\bm{q'}}\right] \left(\mathrm{Tr}\left[\tilde{\rho}_{0,B_{l},\bm{q}\bm{p}}^{(1)}\tilde{\rho}_{0,B_{l},\bm{q'}\bm{p'}}^{(1)}\right]-\frac{1}{2^m} \mathrm{Tr}\left[\tilde{\rho}_{0,B_{l},\bm{q}\bm{p}}^{(1)}\right]\mathrm{Tr}\left[\tilde{\rho}_{0,B_{l},\bm{q'}\bm{p'}}^{(1)}\right] \right)\\
    &= \frac{1}{4}\cdot\frac{2^m}{2^{2m}-1} \sum_{\bm{p},\bm{q},\bm{p'},\bm{q'}} \delta_{(\bm{p},\bm{q})}\delta_{(\bm{p'},\bm{q'})}\left(\mathrm{Tr}\left[\tilde{\rho}_{0,B_{l},\bm{q}\bm{p}}^{(1)}\tilde{\rho}_{0,B_{l},\bm{q'}\bm{p'}}^{(1)}\right]-\frac{1}{2^m} \mathrm{Tr}\left[\tilde{\rho}_{0,B_{l},\bm{q}\bm{p}}^{(1)}\right]\mathrm{Tr}\left[\tilde{\rho}_{0,B_{l},\bm{q'}\bm{p'}}^{(1)}\right] \right)\\
    &= \frac{1}{4}\cdot\frac{2^m}{2^{2m}-1} \sum_{\bm{p},\bm{p'}} \left(\mathrm{Tr}\left[\tilde{\rho}_{0,B_{l},\bm{p}\bm{p}}^{(1)}\tilde{\rho}_{0,B_{l},\bm{p'}\bm{p'}}^{(1)}\right]-\frac{1}{2^m} \mathrm{Tr}\left[\tilde{\rho}_{0,B_{l},\bm{p}\bm{p}}^{(1)}\right]\mathrm{Tr}\left[\tilde{\rho}_{0,B_{l},\bm{p'}\bm{p'}}^{(1)}\right] \right),\\
\end{split}
\end{equation}
where we define $\tilde{\rho}_{0,B_{l},\bm{q}\bm{p}}^{(1)}=\mathrm{Tr}_{\bar{S}_{(k,d)}}[(\ket{\bm{p}}\bra{\bm{q}}\otimes \mathbb{I}_{S_{(k,d)}})\tilde{\rho}_{0,B_{l}}^{(1)}]$ with $\tilde{\rho}_{0,B_{l}}^{(1)}= V_{r}(\bm{x},\bm{\theta}) \tilde{B}_{\bm{x'},\theta_{i}}\rho_{0} V_{r}^{\dagger}(\bm{x},\bm{\theta})$ and $B_{\theta_{i},\bm{p}\bm{q}}= \mathrm{Tr}_{\bar{S}_{(k,d)}}[(\ket{\bm{q}}\bra{\bm{p}}\otimes \mathbb{I}_{S_{(k,d)}}) B_{\theta_{i}}]$.
Here the following two equalities are utilized;
$$\mathrm{Tr}\left[B_{\theta_{i},\bm{p}\bm{q}}\right]=\mathrm{Tr}\left[\mathrm{Tr}_{\bar{S}_{(k,d)}}\left[\left(\ket{\bm{q}}\bra{\bm{p}}\otimes \mathbb{I}_{S_{(k,d)}}\right)B_{\theta_{i}}\right]\right]=0,$$
\begin{equation}
\begin{split}
    \mathrm{Tr}\left[B_{\theta_{i},\bm{p}\bm{q}}B_{\theta_{i},\bm{p'}\bm{q'}}\right] &= \mathrm{Tr}\left[\mathrm{Tr}_{\bar{S}_{(k,d)}}\left[\left(\ket{\bm{q}}\bra{\bm{p}}\otimes \mathbb{I}_{S_{(k,d)}}\right)B_{\theta_{i}}\right]\mathrm{Tr}_{\bar{S}_{(k,d)}}\left[\left(\ket{\bm{q'}}\bra{\bm{p'}}\otimes \mathbb{I}_{S_{(k,d)}}\right)B_{\theta_{i}}\right]\right] \\ 
    &= \delta_{(\bm{p},\bm{q})}\delta_{(\bm{p'},\bm{q'})} \mathrm{Tr}\left[B_{\theta_{i}}^2\right]\\
    &= \delta_{(\bm{p},
    \bm{q})}\delta_{(\bm{p'},\bm{q'})} 2^m.
\end{split}
\end{equation}
Thus we focus on the following quantity $\sum_{\bm{p},\bm{p'}} \mathrm{Tr}[\tilde{\rho}_{0,B_{l},\bm{p}\bm{p}}^{(1)}]\mathrm{Tr}[\tilde{\rho}_{0,B_{l},\bm{p'}\bm{p'}}^{(1)}]$ and $\sum_{\bm{p},\bm{p'}} \mathrm{Tr}[\tilde{\rho}_{0,B_{l},\bm{p}\bm{p}}^{(1)}\tilde{\rho}_{0,B_{l},\bm{p'}\bm{p'}}^{(1)}]$.
Note that these terms can be written as
\begin{equation}
\label{eq:qfk_var_rest1}
\begin{split}
    \sum_{\bm{p},\bm{p'}} \mathrm{Tr}\left[\tilde{\rho}_{0,B_{l},\bm{p}\bm{p}}^{(1)}\right]\mathrm{Tr}\left[\tilde{\rho}_{0,B_{l},\bm{p'}\bm{p'}}^{(1)}\right] &= \sum_{\bm{p},\bm{p'}} \mathrm{Tr}\left[\mathrm{Tr}_{\bar{S}_{(k,d)}}\left[\left(\ket{\bm{p}}\bra{\bm{p}}\otimes \mathbb{I}_{S_{(k,d)}}\right)\tilde{\rho}_{0,B_{l}}^{(1)}\right]\right]\mathrm{Tr}\left[\mathrm{Tr}_{\bar{S}_{(k,d)}}\left[\left(\ket{\bm{p'}}\bra{\bm{p'}}\otimes \mathbb{I}_{S_{(k,d)}}\right)\tilde{\rho}_{0,B_{l}}^{(1)}\right]\right] \\
    & = \mathrm{Tr}\left[\tilde{\rho}_{0,B_{l}}^{(1)}\right]\mathrm{Tr}\left[\tilde{\rho}_{0,B_{l}}^{(1)}\right]\\
    & = \mathrm{Tr}\left[V_{r}(\bm{x},\bm{\theta}) \tilde{B}_{\bm{x'},\theta_{i}}\rho_{0} V_{r}^{\dagger}(\bm{x},\bm{\theta})\right]\mathrm{Tr}\left[V_{r}(\bm{x},\bm{\theta}) \tilde{B}_{\bm{x'},\theta_{i}}\rho_{0} V_{r}^{\dagger}(\bm{x},\bm{\theta})\right]\\
    & = \mathrm{Tr}\left[\tilde{B}_{\bm{x'},\theta_{i}}\rho_{0} \right]\mathrm{Tr}\left[\tilde{B}_{\bm{x'},\theta_{i}}\rho_{0} \right],\\
\end{split}
\end{equation}
\begin{equation}
\label{eq:qfk_var_rest2}
\begin{split}
    \sum_{\bm{p},\bm{p'}} \mathrm{Tr}[\tilde{\rho}_{0,B_{l},\bm{p}\bm{p}}^{(1)}\tilde{\rho}_{0,B_{l},\bm{p'}\bm{p'}}^{(1)}] &= \sum_{\bm{p},\bm{p'}} \mathrm{Tr}[\mathrm{Tr}_{\bar{S}_{(k,d)}}[(\ket{\bm{p}}\bra{\bm{p}}\otimes \mathbb{I}_{S_{(k,d)}})\tilde{\rho}_{0,B_{l}}^{(1)}]\mathrm{Tr}_{\bar{S}_{(k,d)}}[(\ket{\bm{p'}}\bra{\bm{p'}}\otimes \mathbb{I}_{S_{(k,d)}})\tilde{\rho}_{0,B_{l}}^{(1)}]] \\
    & = \mathrm{Tr}\left[\mathrm{Tr}_{\bar{S}_{(k,d)}}[\tilde{\rho}_{0,B_{l}}^{(1)}]\mathrm{Tr}_{\bar{S}_{(k,d)}}[\tilde{\rho}_{0,B_{l}}^{(1)}]\right]\\
    & = \mathrm{Tr}\left[\mathrm{Tr}_{\bar{S}_{(k,d)}}[ V_{r}(\bm{x},\bm{\theta}) \tilde{B}_{\bm{x'},\theta_{i}}\rho_{0} V_{r}^{\dagger}(\bm{x},\bm{\theta})]\mathrm{Tr}_{\bar{S}_{(k,d)}}[ V_{r}(\bm{x},\bm{\theta}) \tilde{B}_{\bm{x'},\theta_{i}}\rho_{0} V_{r}^{\dagger}(\bm{x},\bm{\theta})]\right].\\
\end{split}
\end{equation}
This indicates that $V_{r}(\bm{x},\bm{\theta})$ can be excluded from the expectation calculation for the quantity in Eq.~\eqref{eq:qfk_var_rest1}, but not from the calculation for the other in Eq.~\eqref{eq:qfk_var_rest2}. 

Then we integrate the second quantity in Eq.~\eqref{eq:qfk_var_rest2} over the unitary $V_{r}(\bm{x},\bm{\theta})$.
We remind that $V_{r}(\bm{x},\bm{\theta})$ contains all unitary blocks in the light-cone of $W_{k,d}(\bm{x},\bm{\theta}_{k,d})$.
Hence, the quantity is iteratively integrated over every unitary block.
To do so, we consider the following situations: a unitary block $w_s$ acting on (1) a subspace of $S'$, (2) a subspace of $\bar{S'}$, (3) a subspace of both $S'$ and $\bar{S'}$ and (4) $S'$ and a subspace of $\bar{S'}$.
Then, for arbitrary operator $A: S' \otimes \bar{S'} \to S' \otimes \bar{S'}$, the expectation of $ \mathrm{Tr}[\mathrm{Tr}_{\bar{S'}}[w_{s}Aw_{s}^{\dagger}]\mathrm{Tr}_{\bar{S'}}[w_{s}Aw_{s}^{\dagger}]]$ over $w_{s}: S_{s} \to S_{s}$ can be obtained as follows;
\begin{enumerate}
    \item $S_{s} \subseteq S'$ 
    \begin{equation}
    \begin{split} \label{eq:ala_int_vav1}
        \left\langle \mathrm{Tr}\left[\mathrm{Tr}_{\bar{S'}}\left[w_{s}Aw_{s}^{\dagger}\right]\mathrm{Tr}_{\bar{S'}}\left[w_{s}Aw_{s}^{\dagger}\right]\right]\right\rangle_{w_{s}}
        &= \left\langle \mathrm{Tr}\left[w_{s}\mathrm{Tr}_{\bar{S'}}\left[A\right]w_{s}^{\dagger}w_{s}\mathrm{Tr}_{\bar{S'}}\left[A\right]w_{s}^{\dagger}\right]\right\rangle_{w_{s}}\\
        &= \mathrm{Tr}\left[\mathrm{Tr}_{\bar{S'}}\left[A\right]\mathrm{Tr}_{\bar{S'}}\left[A\right]\right]
    \end{split}
    \end{equation}
    \item $S_{s} \subset \bar{S'}$ 
    \begin{equation}
    \begin{split} \label{eq:ala_int_vav2}
        \left\langle \mathrm{Tr}\left[\mathrm{Tr}_{\bar{S'}}\left[w_{s}Aw_{s}^{\dagger}\right]\mathrm{Tr}_{\bar{S'}}\left[w_{s}Aw_{s}^{\dagger}\right]\right]\right\rangle_{w_{s}}
        &= \left\langle \mathrm{Tr}\left[\mathrm{Tr}_{\bar{S'}}\left[Aw_{s}^{\dagger}w_{s}\right]\mathrm{Tr}_{\bar{S'}}\left[Aw_{s}^{\dagger}w_{s}\right]\right]\right\rangle_{w_{s}}\\
        &= \mathrm{Tr}\left[\mathrm{Tr}_{\bar{S'}}\left[A\right]\mathrm{Tr}_{\bar{S'}}\left[A\right]\right]
    \end{split}
    \end{equation}
    \item $S_{s} = S_{h}\otimes S_{\bar{h}}$ with $d^{1/2}$-dimensional spaces $S_{h} \subset S'$ and $S_{\bar{h}} \subset \bar{S'}$ 
    \begin{equation}
    \begin{split} \label{eq:ala_int_vav3}
        &\left\langle \mathrm{Tr}\left[\mathrm{Tr}_{\bar{S'}}\left[w_{s}Aw_{s}^{\dagger}\right]\mathrm{Tr}_{\bar{S'}}\left[w_{s}Aw_{s}^{\dagger}\right]\right]\right\rangle_{w_{s}} \\
        &= \left\langle  \mathrm{Tr}\left[\left(w_{s}Aw_{s}^{\dagger} \otimes w_{s}Aw_{s}^{\dagger} \right)\left( Swap_{S'_{1} \otimes S'_{2}} \otimes \mathbb{I}_{\bar{S'}_{1}\otimes \bar{S'}_{2}}\right)\right] \right\rangle_{w_{s}} \\
        & = \frac{1}{d^2-1}\Bigl( \mathrm{Tr}\left[\left(\mathbb{I}_{S_{s,1} \otimes S_{s,2}} \otimes \mathrm{Tr}_{S_{s,1}}\left[A\right] \otimes \mathrm{Tr}_{S_{s,2}}\left[A\right] \right)\left( Swap_{S'_{1} \otimes S'_{2}} \otimes \mathbb{I}_{\bar{S'}_{1}\otimes \bar{S'}_{2}}\right)\right] \\
        & \qquad  + \mathrm{Tr}\left[\left(Swap_{S_{s,1}\otimes S_{s,2}} \otimes \mathrm{Tr}_{S_{s}} \otimes \mathrm{Tr}_{S_{s,1}\cup S_{s,2} }\left[A \otimes A \left( Swap_{S_{s,1} \otimes S_{s,2}} \otimes \mathbb{I}_{\bar{S_{s,1}}\otimes \bar{S_{s,2}}}\right)\right] \right)\left( Swap_{S'_{1} \otimes S'_{2}} \otimes \mathbb{I}_{\bar{S'_{1}}\otimes \bar{S'_{2}}}\right)\right] \Bigr) \\
        & \quad - \frac{1}{d(d^2-1)}\Bigl( \mathrm{Tr}\left[\left(\mathbb{I}_{S_{s,1} \otimes S_{s,1}} \otimes \mathrm{Tr}_{S_{s,1}\cup S_{s,2} }\left[A \otimes A \left( Swap_{S_{s,1} \otimes S_{s,2}} \otimes \mathbb{I}_{\bar{S_{s,1}}\otimes \bar{S_{s,2}}}\right)\right] \right)\left( Swap_{S'_{1} \otimes S'_{2}} \otimes \mathbb{I}_{\bar{S'}_{1}\otimes \bar{S'}_{2}}\right)\right] \\
        & \quad \qquad + \mathrm{Tr}\left[\left(Swap_{S_{s,1}\otimes S_{s,2}} \otimes \mathrm{Tr}_{S_{s,1}}\left[A\right] \otimes \mathrm{Tr}_{S_{s,2}}\left[A\right] \right)\left( Swap_{S'_{1} \otimes S'_{2}} \otimes \mathbb{I}_{\bar{S'}_{1}\otimes \bar{S'}_{2}}\right)\right] \Bigr) \\
        & = \frac{d^{1/2}}{d+1}\left(\mathrm{Tr}\left[\mathrm{Tr}_{\bar{S'}\cup S_{h}}\left[A\right]\mathrm{Tr}_{\bar{S'}\cup S_{h}}\left[A\right]\right] + \mathrm{Tr}\left[\mathrm{Tr}_{\bar{S'}/ S_{\bar{h}}}\left[A\right]\mathrm{Tr}_{\bar{S'}/ S_{\bar{h}}}\left[A\right]\right] \right)
    \end{split}
    \end{equation}
    \item $S_{s} = S'\otimes S_{\bar{h}}$ with $d^{1/2}$-dimensional spaces $S'$ and $S_{\bar{h}} \subset \bar{S'}$ 
    \begin{equation}
    \begin{split} \label{eq:ala_int_vav4}
        \left\langle \mathrm{Tr}\left[\mathrm{Tr}_{\bar{S'}}\left[w_{s}Aw_{s}^{\dagger}\right]\mathrm{Tr}_{\bar{S'}}\left[w_{s}Aw_{s}^{\dagger}\right]\right]\right\rangle_{w_{s}} &= \left\langle  \mathrm{Tr}\left[\left(w_{s}Aw_{s}^{\dagger} \otimes w_{s}Aw_{s}^{\dagger} \right)\left( Swap_{S'_{1} \otimes S'_{2}} \otimes \mathbb{I}_{\bar{S'}_{1}\otimes \bar{S'}_{2}}\right)\right] \right\rangle_{w_{s}} \\
        & = \frac{d^{1/2}}{d+1}\left(\mathrm{Tr}\left[A\right]\mathrm{Tr}\left[A\right] + \mathrm{Tr}\left[\mathrm{Tr}_{\bar{S'}/ S_{\bar{h}}}\left[A\right]\mathrm{Tr}_{\bar{S'}/ S_{\bar{h}}}\left[A\right]\right] \right)
    \end{split}
    \end{equation}
\end{enumerate}
Here, $\mathbb{I}_{S_{1}\otimes S_{2}}$ and $Swap_{S_{1}\otimes S_{2}}$ denote the identity operator and the swap operator acting on the systems $S_{1}, S_{2}$, respectively.
Also the subspace labeled with the number in the subscript (for example, $S_{s,i}$ with $i\in\{1,2\}$) represents one of the duplicated subsystems.
Note that the swap operation can be expressed as
\begin{equation}
    Swap_{S_{1}\otimes S_{2}} = \sum_{i,j} \ket{i}_{S_1}\bra{j} \otimes \ket{j}_{S_2}\bra{i}. 
\end{equation}

Now we obtain the expectation of the quantity in Eq.~\eqref{eq:qfk_var_rest2} using the above techniques.
We here give an example of the integration for the unitary blocks in the $(d-1)$-th layer, $W_{k-1,d-1}(\bm{x},\bm{\theta}_{k-1,d-1})$ and $W_{k,d-1}(\bm{x},\bm{\theta}_{k,d-1})$, assuming $2\le k \le \kappa-1$.
Then we have
\begin{equation}
\begin{split}
    &\left\langle \mathrm{Tr}\left[\mathrm{Tr}_{\bar{S}_{(k,d)}}[ V_{r}(\bm{x},\bm{\theta}) \tilde{B}_{\bm{x'},\theta_{i}}\rho_{0} V_{r}^{\dagger}(\bm{x},\bm{\theta})]\mathrm{Tr}_{\bar{S}_{(k,d)}}[ V_{r}(\bm{x},\bm{\theta}) \tilde{B}_{\bm{x'},\theta_{i}}\rho_{0} V_{r}^{\dagger}(\bm{x},\bm{\theta})]\right] \right\rangle_{W_{k-1,d-1}(\bm{x},\bm{\theta}_{k-1,d-1}),W_{k,d-1}(\bm{x},\bm{\theta}_{k-1,d-1})} \\
    &= \left(\frac{2^{\frac{m}{2}}}{2^{m}+1}\right)^{2} \Bigl(\mathrm{Tr}\left[ V_{r,d-1} \tilde{B}_{\bm{x'},\theta_{i}}\rho_{0} V_{r,d-1}^{\dagger}\right]\mathrm{Tr}\left[ V_{r,d-1} \tilde{B}_{\bm{x'},\theta_{i}}\rho_{0} V_{r,d-1}^{\dagger}\right] \\
    & \qquad \qquad \qquad \qquad +  \mathrm{Tr}\left[\mathrm{Tr}_{\bar{S}_{(k-1,d-1)}}[ V_{r,d-1} \tilde{B}_{\bm{x'},\theta_{i}}\rho_{0} V_{r,d-1}^{\dagger}]\mathrm{Tr}_{\bar{S}_{(k-1,d-1)}}[ V_{r,d-1} \tilde{B}_{\bm{x'},\theta_{i}}\rho_{0} V_{r,d-1}^{\dagger}]\right]  \\
    & \qquad \qquad \qquad \qquad +  \mathrm{Tr}\left[\mathrm{Tr}_{\bar{S}_{(k,d-1)}}[ V_{r,d-1} \tilde{B}_{\bm{x'},\theta_{i}}\rho_{0} V_{r,d-1}^{\dagger}]\mathrm{Tr}_{\bar{S}_{(k,d-1)}}[ V_{r,d-1} \tilde{B}_{\bm{x'},\theta_{i}}\rho_{0} V_{r,d-1}^{\dagger}]\right] \\
    & \qquad \qquad \qquad \qquad +  \mathrm{Tr}\left[\mathrm{Tr}_{\bar{S}_{(k-1,d-1)} \cap \bar{S}_{(k,d-1)}}[ V_{r,d-1} \tilde{B}_{\bm{x'},\theta_{i}}\rho_{0} V_{r,d-1}^{\dagger}]\mathrm{Tr}_{\bar{S}_{(k-1,d-1)} \cap \bar{S}_{(k,d-1)}}[ V_{r,d-1} \tilde{B}_{\bm{x'},\theta_{i}}\rho_{0} V_{r,d-1}^{\dagger}]\right] \Bigr),
\end{split}
\end{equation}
where $V_{r,l}$ is the set of unitary blocks in $V_{r}(\bm{x},\bm{\theta})$ except for the ones in $l$-th layer. 
That is, for the above case,  $V_{r}(\bm{x},\bm{\theta}) = W_{k-1,d-1}(\bm{x},\bm{\theta}_{k-1,d-1}) W_{k,d-1}(\bm{x},\bm{\theta}_{k-1,d-1})V_{r,d-1}$ is satisfied. 
By iterating the calculation up to the first layer $d=1$, the following result can be obtained;
\begin{equation}
\begin{split}
    &\left\langle \mathrm{Tr}\left[\mathrm{Tr}_{\bar{S}_{(k,d)}}[ V_{r}(\bm{x},\bm{\theta}) \tilde{B}_{\bm{x'},\theta_{i}}\rho_{0} V_{r}^{\dagger}(\bm{x},\bm{\theta})]\mathrm{Tr}_{\bar{S}_{(k,d)}}[ V_{r}(\bm{x},\bm{\theta}) \tilde{B}_{\bm{x'},\theta_{i}}\rho_{0} V_{r}^{\dagger}(\bm{x},\bm{\theta})]\right] \right\rangle_{V_{r}(\bm{x},\bm{\theta})} \\
    &=\sum_{h\in P_{U}(S^{(k_{s}:k_{l},1)})} t_{h} \mathrm{Tr}\left[\mathrm{Tr}_{\bar{h}}[ \tilde{B}_{\bm{x'},\theta_{i}}\rho_{0} ]\mathrm{Tr}_{\bar{h}}[ \tilde{B}_{\bm{x'},\theta_{i}}\rho_{0} ]\right],
\end{split}
\end{equation}
where $t_{h}\in\mathbb{R}^{+}$ and $P_{U}(S^{(k_{s}:k_{l},1)})=\{S_{(k_{s},1)},S_{(k_{s}+1,1)},\ldots S_{(k_{l},1)},S_{(k_{s},1)}\cup S_{(k_{s}+1,1)}, S_{(k_{s},1)}\cup S_{(k_{s}+2,1)},\ldots\}$ is the set of subspace, each element of which is the union of the spaces in a subset of $P(S^{(k_{s}:k_{l},1)})$.
Here, $k_{s}(k_{l})$ is the smallest (largest) label of the unitary blocks in the first layer of $V_{r}(\bm{x},\bm{\theta})$.
Importantly, a set of the coefficients $\{t_{h}\}$ differs depending on the position of the unitary $W_{k,d}(\bm{x},\bm{\theta}_{k,d})$.
In this paper, we consider the following cases: (1) $W_{k,d}(\bm{x},\bm{\theta}_{k,d})$ with $k=1$ or $k=\kappa$ (the first or the last unitary block in a layer) and (2) $W_{k,d}(\bm{x},\bm{\theta}_{k,d})$ with $k$ satisfying $k_s\ge1$ and $k_l\le d$ (a middle block).
Examples of the coefficient $t_{S_{(k_{s}:k_{l},1)}}$ for these cases are as follows;
\begin{enumerate}
    \item $k=1$ or $k=\kappa$
    \begin{equation}
        t_{S_{(k_{s}:k_{l},1)}} = \left(\frac{2^{\frac{m}{2}}}{2^{m}+1}\right)^{\frac{3}{2}(d-1)}  
    \end{equation}
    \item $k$ satisfies $k_s\ge1$ and $k_l\le d$
    \begin{equation}
        t_{S_{(k_{s}:k_{l},1)}} = \left(\frac{2^{\frac{m}{2}}}{2^{m}+1}\right)^{2(d-1)}.  
    \end{equation}
\end{enumerate}
%
Note that every $t_h$ is equal to or greater than $(2^{\frac{m}{2}}/2^{m}+1)^{2(d-1)}$.
Then, the expectation of the first term over $ U_{1:i}(\bm{x},\bm{\theta})$ can be written as
\begin{equation}
\label{eq:fterm_middle}
    \frac{1}{4}\cdot\frac{2^m}{2^{2m}-1}  \left(\sum_{h\in P_{U}(S^{(k_{s}:k_{l},1)})} t_{h} \mathrm{Tr}\left[\mathrm{Tr}_{\bar{h}}[ \tilde{B}_{\bm{x'},\theta_{i}}\rho_{0} ]\mathrm{Tr}_{\bar{h}}[ \tilde{B}_{\bm{x'},\theta_{i}}\rho_{0} ]\right]-\frac{1}{2^m}\mathrm{Tr}\left[\tilde{B}_{\bm{x'},\theta_{i}}\rho_{0} \right]\mathrm{Tr}\left[\tilde{B}_{\bm{x'},\theta_{i}}\rho_{0} \right] \right)
\end{equation}

Next, we compute the expectation of Eq.~\eqref{eq:fterm_middle} over $U_{1:i}(\bm{x'},\bm{\theta})$.
Here we begin with the integration for $\tilde{W}_{k,d}(\bm{x'},\theta_{i})$.
The expectation of $\mathrm{Tr}[\mathrm{Tr}_{\bar{h}}[ \tilde{B}_{\bm{x'},\theta_{i}}\rho_{0} ]\mathrm{Tr}_{\bar{h}}[ \tilde{B}_{\bm{x'},\theta_{i}}\rho_{0} ]]$ in the first term of Eq.~\eqref{eq:fterm_middle} can be calculated as 

\begin{equation}
\begin{split}
\label{eq:fmiddle_rest1}
    & \left\langle\mathrm{Tr}\left[\mathrm{Tr}_{\bar{h}}\left[\tilde{B}_{x',\theta_{i}}\rho_{0}\right]\mathrm{Tr}_{\bar{h}}\left[\tilde{B}_{x',\theta_{i}}\rho_{0} \right]\right]\right\rangle_{\tilde{W}_{k,d}(\bm{x'},\theta_{i})}\\
    &= \biggl\langle \mathrm{Tr}\left[\mathrm{Tr}_{\bar{h}}\Bigl[V_{r}^{\dagger}(\bm{x'},\bm{\theta})\tilde{W}_{k,d}^{\dagger}(\bm{x'},\theta_{i})B_{\theta_{i}}\tilde{W}_{k,d}(\bm{x'},\theta_{i})V_{r}(\bm{x'},\bm{\theta})\rho_{0}\right]\\
    & \qquad \qquad \times \mathrm{Tr}_{\bar{h}}\left[V_{r}^{\dagger}(\bm{x'},\bm{\theta})\tilde{W}_{k,d}^{\dagger}(\bm{x'},\theta_{i})B_{\theta_{i}}\tilde{W}_{k,d}(\bm{x'},\theta_{i})V_{r}(\bm{x'},\bm{\theta})\rho_{0} \right]\Bigr]\biggr\rangle_{\tilde{W}_{k,d}(\bm{x'},\theta_{i})} \\
    &= \biggl\langle \mathrm{Tr}\Bigl[\left(V_{r}^{\dagger}(\bm{x'},\bm{\theta})\tilde{W}_{k,d}^{\dagger}(\bm{x'},\theta_{i})B_{\theta_{i}}\tilde{W}_{k,d}(\bm{x'},\theta_{i})V_{r}(\bm{x'},\bm{\theta})\rho_{0} \otimes V_{r}^{\dagger}(\bm{x'},\bm{\theta})\tilde{W}_{k,d}^{\dagger}(\bm{x'},\theta_{i})B_{\theta_{i}}\tilde{W}_{k,d}(\bm{x'},\theta_{i})V_{r}(\bm{x'},\bm{\theta})\rho_{0} \right) \\
    & \qquad \qquad \qquad \qquad \qquad \qquad \qquad \qquad \qquad \qquad \qquad \qquad \qquad \qquad \qquad \qquad \times \left(Swap_{h_{1}\otimes h_{2}} \otimes \mathbb{I}_{\bar{h}_{1}\otimes \bar{h}_{2}} \right)\biggr\rangle_{\tilde{W}_{k,d}(\bm{x'},\theta_{i})} \\
    &= \frac{1}{2^{2m}-1} \Bigl(\mathrm{Tr}\left[B_{\theta_{i}}\right]\mathrm{Tr}\left[B_{\theta_{i}}\right]\mathrm{Tr}\Bigl[\left(V_{r}^{\dagger}(\bm{x'},\bm{\theta})\otimes V_{r}^{\dagger}(\bm{x'},\bm{\theta})\right) \left(\mathbb{I}_{S_{(k,d),1} \otimes S_{(k,d),2}} \otimes \mathbb{I}_{\bar{S}_{(k,d),1} \otimes \bar{S}_{(k,d),2}}\right) \\
    & \qquad \qquad \qquad \qquad  \qquad \qquad  \qquad \qquad \qquad \qquad \qquad \qquad  \times \left(V_{r}(\bm{x'},\bm{\theta})\otimes V_{r}(\bm{x'},\bm{\theta})\right) \left(\rho_{0} \otimes \rho_{0}\right) \left(Swap_{h_{1}\otimes h_{2}} \otimes \mathbb{I}_{\bar{h}_{1}\otimes \bar{h}_{2}} \right)\Bigr] \\
    & \qquad + \mathrm{Tr}\left[B_{\theta_{i}}^2\right]\mathrm{Tr}\Bigl[\left(V_{r}^{\dagger}(\bm{x'},\bm{\theta})\otimes V_{r}^{\dagger}(\bm{x'},\bm{\theta})\right) \left(Swap_{S_{(k,d),1} \otimes S_{(k,d),2}} \otimes \mathbb{I}_{\bar{S}_{(k,d),1} \otimes \bar{S}_{(k,d),2}}\right) \\
    & \qquad \qquad \qquad \qquad  \qquad \qquad  \qquad \qquad \qquad \qquad \qquad \qquad  \times \left(V_{r}(\bm{x'},\bm{\theta})\otimes V_{r}(\bm{x'},\bm{\theta})\right) \left(\rho_{0} \otimes \rho_{0}\right) \left(Swap_{h_{1}\otimes h_{2}} \otimes \mathbb{I}_{\bar{h}_{1}\otimes \bar{h}_{2}} \right)\Bigr] \Bigr) \\
    & \quad  - \frac{1}{2^{m}(2^{2m}-1)} \Bigl(\mathrm{Tr}\left[B_{\theta_{i}}^2\right]\mathrm{Tr}\Bigl[\left(V_{r}^{\dagger}(\bm{x'},\bm{\theta})\otimes V_{r}^{\dagger}(\bm{x'},\bm{\theta})\right) \left(\mathbb{I}_{S_{(k,d),1} \otimes S_{(k,d),2}} \otimes \mathbb{I}_{\bar{S}_{(k,d),1} \otimes \bar{S}_{(k,d),2}}\right) \\
    & \qquad \qquad \qquad \qquad  \qquad \qquad  \qquad \qquad \qquad \qquad \qquad \qquad  \times \left(V_{r}(\bm{x'},\bm{\theta})\otimes V_{r}(\bm{x'},\bm{\theta})\right) \left(\rho_{0} \otimes \rho_{0}\right) \left(Swap_{h_{1}\otimes h_{2}} \otimes \mathbb{I}_{\bar{h}_{1}\otimes \bar{h}_{2}} \right)\Bigr] \\
    & \qquad + \mathrm{Tr}\left[B_{\theta_{i}}\right]\mathrm{Tr}\left[B_{\theta_{i}}\right]\mathrm{Tr}\Bigl[\left(V_{r}^{\dagger}(\bm{x'},\bm{\theta})\otimes V_{r}^{\dagger}(\bm{x'},\bm{\theta})\right) \left(Swap_{S_{(k,d),1} \otimes S_{(k,d),2}} \otimes \mathbb{I}_{\bar{S}_{(k,d),1} \otimes \bar{S}_{(k,d),2}}\right) \\
    & \qquad \qquad \qquad \qquad  \qquad \qquad  \qquad \qquad \qquad \qquad \qquad \qquad  \times \left(V_{r}(\bm{x'},\bm{\theta})\otimes V_{r}(\bm{x'},\bm{\theta})\right) \left(\rho_{0} \otimes \rho_{0}\right) \left(Swap_{h_{1}\otimes h_{2}} \otimes \mathbb{I}_{\bar{h}_{1}\otimes \bar{h}_{2}} \right)\Bigr] \Bigr) \\
    &= \frac{2^{m}}{2^{2m}-1} \Bigl(\mathrm{Tr}\Bigl[\left(V_{r}^{\dagger}(\bm{x'},\bm{\theta})\otimes V_{r}^{\dagger}(\bm{x'},\bm{\theta})\right) \left(Swap_{S_{(k,d),1} \otimes S_{(k,d),2}} \otimes \mathbb{I}_{\bar{S}_{(k,d),1} \otimes \bar{S}_{(k,d),2}}\right) \\
    & \qquad \qquad \qquad \qquad  \qquad \qquad  \qquad \qquad \qquad \qquad \qquad \qquad  \times \left(V_{r}(\bm{x'},\bm{\theta})\otimes V_{r}(\bm{x'},\bm{\theta})\right) \left(\rho_{0} \otimes \rho_{0}\right) \left(Swap_{h_{1}\otimes h_{2}} \otimes \mathbb{I}_{\bar{h}_{1}\otimes \bar{h}_{2}} \right)\Bigr] \\
    & \quad  -\frac{1}{2^{m}} \mathrm{Tr}\Bigl[\left(V_{r}^{\dagger}(\bm{x'},\bm{\theta})\otimes V_{r}^{\dagger}(\bm{x'},\bm{\theta})\right) \left(\mathbb{I}_{S_{(k,d),1} \otimes S_{(k,d),2}} \otimes \mathbb{I}_{\bar{S}_{(k,d),1} \otimes \bar{S}_{(k,d),2}}\right) \\
    & \qquad \qquad \qquad \qquad  \qquad \qquad  \qquad \qquad \qquad \qquad \qquad \qquad  \times \left(V_{r}(\bm{x'},\bm{\theta})\otimes V_{r}(\bm{x'},\bm{\theta})\right) \left(\rho_{0} \otimes \rho_{0}\right) \left(Swap_{h_{1}\otimes h_{2}} \otimes \mathbb{I}_{\bar{h}_{1}\otimes \bar{h}_{2}} \right)\Bigr]\\
    &= \frac{2^{m}}{2^{2m}-1} \Bigl(\mathrm{Tr}\Bigl[\left(V_{r}^{\dagger}(\bm{x'},\bm{\theta})\otimes V_{r}^{\dagger}(\bm{x'},\bm{\theta})\right) \left(Swap_{S_{(k,d),1} \otimes S_{(k,d),2}} \otimes \mathbb{I}_{\bar{S}_{(k,d),1} \otimes \bar{S}_{(k,d),2}}\right) \\
    & \qquad \qquad \qquad \qquad  \qquad \qquad  \qquad \qquad \qquad \qquad \qquad \qquad  \times \left(V_{r}(\bm{x'},\bm{\theta})\otimes V_{r}(\bm{x'},\bm{\theta})\right) \left(\rho_{0} \otimes \rho_{0}\right) \left(Swap_{h_{1}\otimes h_{2}} \otimes \mathbb{I}_{\bar{h}_{1}\otimes \bar{h}_{2}} \right)\Bigr] \Bigr)\\
    & \qquad \qquad \qquad   -\frac{1}{2^{m}} \mathrm{Tr}\left[\mathrm{Tr}_{\bar{h}}\left[\rho_{0}\right]\mathrm{Tr}_{\bar{h}}\left[\rho_{0}\right]\right] \Bigr),\\
\end{split}
\end{equation}
where we utilize the equality, 
\begin{equation}
\begin{split}
\label{eq:int_tensor_haar}
    &\left\langle V^{\dagger}w_{s}^{\dagger}Aw_{s}VA' \otimes  V^{\dagger}w_{s}^{\dagger}Aw_{s}VA' \right\rangle_{w_{s}} \\
    &= \frac{1}{2^{2m}-1}\Bigl( \left(V^{\dagger}\otimes V^{\dagger}\right) \left(\mathbb{I}_{S_{s,1} \otimes S_{s,2}}\otimes \mathrm{Tr}_{S_{s,1}}\left[A\right] \otimes \mathrm{Tr}_{S_{s,2}}\left[A\right]\right) \left(V\otimes V\right) \left(A'\otimes A'\right) \\
    & \qquad + \left(V^{\dagger}\otimes V^{\dagger}\right) \left(Swap_{S_{s,1} \otimes S_{s,2}}\otimes \mathrm{Tr}_{S_{s,1} \cup S_{s,2}}\left[A\otimes A \left(Swap_{S_{s,1},S_{s,2}} \otimes \mathbb{I}_{\bar{S}_{s,1},\bar{S}_{s,2}} \right)\right] \right) \left(V\otimes V\right) \left(A'\otimes A'\right) \Bigr) \\
    & \quad - \frac{1}{2^{m}(2^{2m}-1)}\Bigl( \left(V^{\dagger}\otimes V^{\dagger}\right) \left(Swap_{S_{s,1} \otimes S_{s,2}}\otimes \mathrm{Tr}_{S_{s,1}}\left[A\right] \otimes \mathrm{Tr}_{S_{s,2}}\left[A\right]\right) \left(V\otimes V\right) \left(A'\otimes A'\right) \\
    & \qquad + \left(V^{\dagger}\otimes V^{\dagger}\right) \left(\mathbb{I}_{S_{s,1} \otimes S_{s,2}}\otimes \mathrm{Tr}_{S_{s,1} \cup S_{s,2}}\left[A\otimes A \left(Swap_{S_{s,1},S_{s,2}} \otimes \mathbb{I}_{\bar{S}_{s,1},\bar{S}_{s,2}} \right)\right] \right) \left(V\otimes V\right) \left(A'\otimes A'\right) \Bigr)
\end{split}
\end{equation}
for arbitrary operator $A,A'$ and  the properties of the Pauli operators, $\mathrm{Tr}[B]=0$ and $\mathrm{Tr}[B^2]=2^{m}$.
Since the first term in Eq\eqref{eq:fmiddle_rest1} still includes $V_{r}(\bm{x'},\bm{\theta})$, we integrate the quantity over all unitary blocks in $V_{r}(\bm{x'},\bm{\theta})$.
Especially, we consider the following quantity,
\begin{equation*}
    \left(V_{r}^{\dagger}(\bm{x'},\bm{\theta})\otimes V_{r}^{\dagger}(\bm{x'},\bm{\theta})\right) \left(Swap_{S_{(k,d),1} \otimes S_{(k,d),2}} \otimes \mathbb{I}_{\bar{S}_{(k,d),1} \otimes \bar{S}_{(k,d),2}}\right) \left(V_{r}(\bm{x'},\bm{\theta})\otimes V_{r}(\bm{x'},\bm{\theta})\right).
\end{equation*}
Then, using the equality in Eq.~\eqref{eq:int_tensor_haar}, we have
\begin{equation}
\begin{split}
\label{eq:int_tensor_haar_res}
    &\left\langle\left(V_{r}^{\dagger}(\bm{x'},\bm{\theta})\otimes V_{r}^{\dagger}(\bm{x'},\bm{\theta})\right) \left(Swap_{S_{(k,d),1} \otimes S_{(k,d),2}} \otimes \mathbb{I}_{\bar{S}_{(k,d),1} \otimes \bar{S}_{(k,d),2}}\right) \left(V_{r}(\bm{x'},\bm{\theta})\otimes V_{r}(\bm{x'},\bm{\theta})\right) \right\rangle_{V_{r}(\bm{x'},\bm{\theta})} \\
    &=\sum_{h'\in P_{U}(S^{(k_{s}:k_{l},1)})} t_{h'}\left(Swap_{h'_{1} \otimes h'_{2}} \otimes \mathbb{I}_{\bar{h'}_{1} \otimes \bar{h'}_{2}}\right),
\end{split}  
\end{equation}
where $t_{h'} \in \mathbb{R}^{+}$.
Note that a set of the coefficients $\{t_{h'}\}$ is the same as $\{t_{h}\}$.
Thus, substituting the above equation into the first term in Eq.~\eqref{eq:fmiddle_rest1}, the following result can be obtained.
\begin{equation}
\begin{split}
    &\Bigl\langle\mathrm{Tr}\Bigl[\left(V_{r}^{\dagger}(\bm{x'},\bm{\theta})\otimes V_{r}^{\dagger}(\bm{x'},\bm{\theta})\right) \left(Swap_{S_{(k,d),1} \otimes S_{(k,d),2}} \otimes \mathbb{I}_{\bar{S}_{(k,d),1} \otimes \bar{S}_{(k,d),2}}\right) \\
    & \qquad \qquad \qquad \qquad \qquad \qquad \qquad \qquad \qquad  \times \left(V_{r}(\bm{x'},\bm{\theta})\otimes V_{r}(\bm{x'},\bm{\theta})\right) \left(\rho_{0} \otimes \rho_{0}\right) \left(Swap_{h_{1}\otimes h_{2}} \otimes \mathbb{I}_{\bar{h}_{1}\otimes \bar{h}_{2}} \right)\Bigr]\Bigr \rangle_{V_{r}(\bm{x'},\bm{\theta})}\\
    &= \mathrm{Tr}\left[\left(\sum_{h'\in P_{U}(S^{(k_{s}:k_{l},1)})} t_{h'}\left(Swap_{h'_{1} \otimes h'_{2}} \otimes \mathbb{I}_{\bar{h'}_{1} \otimes \bar{h'}_{2}}\right)\right) \left(\rho_{0} \otimes \rho_{0} \right)\left(Swap_{h_{1}\otimes h_{2}} \otimes \mathbb{I}_{\bar{h}_{1}\otimes \bar{h}_{2}} \right)\right]\\
    &= \sum_{h'\in P_{U}(S^{(k_{s}:k_{l},1)})} t_{h'} \mathrm{Tr}\left[\left(\rho_{0} \otimes \rho_{0} \right) \left(Swap_{(h_{1} \cup h'_{1}) / (h_{1} \cap h'_{1}) \otimes (h_{1} \cup h'_{1}) / (h_{1} \cap h'_{1})} \otimes \mathbb{I}_{\overline{(h \cup h') / (h \cap h')} \otimes \overline{(h \cup h') / (h \cap h')}} \right)\right] \\
    &= \sum_{h'\in P_{U}(S^{(k_{s}:k_{l},1)})} t_{h'} \mathrm{Tr}\left[\mathrm{Tr}_{\overline{(h \cup h') / (h \cap h')}}\left[\rho_{0}\right]\mathrm{Tr}_{\overline{(h \cup h') / (h \cap h')}}\left[\rho_{0}\right]\right] \\
\end{split}
\end{equation}
Therefore, we can see that the quantity in Eq.~\eqref{eq:fmiddle_rest1} can be represented by the purity of initial state which is patially traced out over a subspace in $P(S_{(k_{s}:k_{l},1)})$.
\begin{equation}
\begin{split}
    & \left\langle\mathrm{Tr}\left[\mathrm{Tr}_{\bar{h}}\left[\tilde{B}_{x',\theta_{i}}\rho_{0}\right]\mathrm{Tr}_{\bar{h}}\left[\tilde{B}_{x',\theta_{i}}\rho_{0} \right]\right]\right\rangle_{U_{1:i}(\bm{x'},\bm{\theta})}\\
    &= \frac{2^{m}}{2^{2m}-1} \Biggl(\Biggl(\sum_{h'\in P_{U}(S^{(k_{s}:k_{l},1)})} t_{h'} \mathrm{Tr}\left[\mathrm{Tr}_{\overline{(h \cup h') / (h \cap h')}}\left[\rho_{0}\right]\mathrm{Tr}_{\overline{(h \cup h') / (h \cap h')}}\left[\rho_{0}\right]\right] \Biggr) -\frac{1}{2^{m}} \mathrm{Tr}\left[\mathrm{Tr}_{\bar{h}}\left[\rho_{0}\right]\mathrm{Tr}_{\bar{h}}\left[\rho_{0}\right]\right] \Biggr).\\
\end{split}
\end{equation}

As for the second term in Eq.~\eqref{eq:fterm_middle}, the integration for $\tilde{W}_{k,d}(\bm{x'},\theta_{i})$ can be calculated in the following way;
\begin{equation}
\begin{split}
\label{eq:aldqfk_var_f}
    & \left\langle\mathrm{Tr}\left[\tilde{B}_{\bm{x'},\theta_{i}}\rho_{0}\right]\mathrm{Tr}\left[\tilde{B}_{\bm{x'},\theta_{i}}\rho_{0}\right]\right\rangle_{\tilde{W}_{k,d}(\bm{x'},\theta_{i})}\\
    & =  \biggl\langle \mathrm{Tr}\left[V_{r}^{\dagger}(\bm{x'},\bm{\theta})\tilde{W}_{k,d}^{\dagger}(\bm{x'},\theta_{i})B_{\theta_{i}}\tilde{W}_{k,d}(\bm{x'},\theta_{i})V_{r}(\bm{x'},\bm{\theta})\rho_{0}\right]\\
    & \qquad \qquad \qquad \qquad \times \mathrm{Tr}\left[V_{r}^{\dagger}(\bm{x'},\bm{\theta})\tilde{W}_{k,d}^{\dagger}(\bm{x'},\theta_{i})B_{\theta_{i}}\tilde{W}_{k,d}(\bm{x'},\theta_{i})V_{r}(\bm{x'},\bm{\theta})\rho_{0}\right]\biggr\rangle_{\tilde{W}_{k,d}(\bm{x'},\theta_{i})} \\
    & =\biggl\langle \mathrm{Tr}\left[\tilde{W}_{k,d}(\bm{x'},\theta_{i})\rho_{V_{x'}}\tilde{W}_{k,d}^{\dagger}(\bm{x'},\theta_{i})B_{\theta_{i}}\right] \mathrm{Tr}\left[\tilde{W}_{k,d}(\bm{x'},\theta_{i})\rho_{V_{x'}}\tilde{W}_{k,d}^{\dagger}(\bm{x'},\theta_{i})B_{\theta_{i}}\right]\biggr\rangle_{\tilde{W}_{k,d}(\bm{x'},\theta_{i})}\\ 
    & = \sum_{\bm{p},\bm{q},\bm{p'},\bm{q'}} \biggl\langle \mathrm{Tr}\left[\tilde{W}_{k,d}(\bm{x'},\theta_{i})\rho_{V_{x'},\bm{q}\bm{p}}\tilde{W}_{k,d}^{\dagger}(\bm{x'},\theta_{i})B_{\theta_{i},\bm{p}\bm{q}}\right] \\
    & \qquad \qquad \qquad \qquad \times \mathrm{Tr}\left[\tilde{W}_{k,d}(\bm{x'},\theta_{i})\rho_{V_{x'},\bm{q'}\bm{p'}}\tilde{W}_{k,d}^{\dagger}(\bm{x'},\theta_{i})B_{\theta_{i},\bm{p'}\bm{q'}}\right]\biggr\rangle_{\tilde{W}_{k,d}(\bm{x'},\theta_{i})}\\ 
    &=  \sum_{\bm{p},\bm{q},\bm{p'},\bm{q'}} \Biggl( \frac{1}{2^{2m}-1} \left(\mathrm{Tr}\left[\rho_{V_{x'},\bm{q}\bm{p}}\right]\mathrm{Tr}\left[B_{\theta_{i},\bm{p}\bm{q}}\right]\mathrm{Tr}\left[\rho_{V_{x'},\bm{q'}\bm{p'}}\right]\mathrm{Tr}\left[B_{\theta_{i},\bm{p'}\bm{q'}}\right]+\mathrm{Tr}\left[\rho_{V_{x'},\bm{q}\bm{p}}\rho_{V_{x'},\bm{q'}\bm{p'}}\right]\mathrm{Tr}\left[B_{\theta_{i},\bm{p}\bm{q}}B_{\theta_{i},\bm{p'}\bm{q'}}\right]\right)\\
    & \qquad -\frac{1}{2^{m}\left(2^{2m}-1\right)}\left(\mathrm{Tr}\left[\rho_{V_{x'},\bm{q}\bm{p}}\right]\mathrm{Tr}\left[\rho_{V_{x'},\bm{q'}\bm{p'}}\right]\mathrm{Tr}\left[B_{\theta_{i},\bm{p}\bm{q}}B_{\theta_{i},\bm{p'}\bm{q'}}\right]+\mathrm{Tr}\left[\rho_{V_{x'},\bm{q}\bm{p}}\rho_{V_{x'},\bm{q'}\bm{p'}}\right]\mathrm{Tr}\left[B_{\theta_{i},\bm{p}\bm{q}}\right]\mathrm{Tr}\left[B_{\theta_{i},\bm{p'}\bm{q'}}\right]\right) \Biggr) \\
    &= \frac{1}{2^{2m}-1} \sum_{\bm{p},\bm{q},\bm{p'},\bm{q'}} \mathrm{Tr}\left[B_{\theta_{i},\bm{p}\bm{q}}B_{\theta_{i},\bm{p'}\bm{q'}}\right] \left(\mathrm{Tr}\left[\rho_{V_{x'},\bm{q}\bm{p}}\rho_{V_{x'},\bm{q'}\bm{p'}}\right]-\frac{1}{2^m}\mathrm{Tr}\left[\rho_{V_{x'},\bm{q}\bm{p}}\right]\mathrm{Tr}\left[\rho_{V_{x'},\bm{q'}\bm{p'}}\right]\right)\\
    &= \frac{2^m}{2^{2m}-1} \sum_{\bm{p},\bm{p'}}  \left(\mathrm{Tr}\left[\rho_{V_{x'},\bm{p}\bm{p}}\rho_{V_{x'},\bm{p'}\bm{p'}}\right]-\frac{1}{2^m}\mathrm{Tr}\left[\rho_{V_{x'},\bm{p}\bm{p}}\right]\mathrm{Tr}\left[\rho_{V_{x'},\bm{p'}\bm{p'}}\right]\right)\\
    &= \frac{2^m}{2^{2m}-1} \left(\mathrm{Tr}\left[\mathrm{Tr}_{\bar{S}_{(k,d)}}\left[\rho_{V_{x'}}\right]\mathrm{Tr}_{\bar{S}_{(k,d)}}\left[\rho_{V_{x'}}\right]\right]-\frac{1}{2^m}\right),\\
\end{split}
\end{equation}
where $\rho_{V_{x'},\bm{pq}}=\mathrm{Tr}_{\bar{S}_{(k,d)}}[(\ket{\bm{p}}\bra{\bm{q}}\otimes \mathbb{I}_{S_{(k,d)}})\rho_{V_{x'}}]$ with $\rho_{V_{x'}}=V_{r}(\bm{x'},\bm{\theta})\rho_{0}V_{r}^{\dagger}(\bm{x'},\bm{\theta})$.
Here we use Lemmas \ref{lem4} and \ref{lem5}, $\mathrm{Tr}\left[B_{\theta_{i},\bm{p}\bm{q}}\right]=0$ and $\mathrm{Tr}\left[B_{\theta_{i},\bm{p}\bm{q}}B_{\theta_{i},\bm{p'}\bm{q'}}\right]=\delta_{(\bm{p},\bm{q})}\delta_{(\bm{p'},\bm{q'})} 2^m$.
Then, integrating the quantity over $V_{r}(\bm{x'},\bm{\theta})$ using Eqs.~\eqref{eq:ala_int_vav1}-~\eqref{eq:ala_int_vav4}, we have
\begin{equation}
\begin{split}
\label{eq:s_in_fterm_middle_integrated}
    &\left\langle\mathrm{Tr}\left[\tilde{B}_{\bm{x'},\theta_{i}}\rho_{0}\right]\mathrm{Tr}\left[\tilde{B}_{\bm{x'},\theta_{i}}\rho_{0}\right]\right\rangle_{U_{1:i}(\bm{x'},\bm{\theta})} = \frac{2^m}{2^{2m}-1} \left(\left(\sum_{h\in P_{U}(S^{(k_{s}:k_{l},1)})} t_{h} \mathrm{Tr}\left[\mathrm{Tr}_{\bar{h}}[ \rho_{0} ]\mathrm{Tr}_{\bar{h}}[\rho_{0} ]\right]\right)-\frac{1}{2^m}\right).
\end{split}
\end{equation}
Therefore we obtain
\begin{equation}
\label{eq:qfk_v_t1}
\begin{split}
    & Var_{a,1} \\
    &=\frac{1}{4}\left(\frac{2^m}{2^{2m}-1}\right)^{2}  \Biggl( \sum_{h\in P_{U}(S^{(k_{s}:k_{l},1)})} \sum_{h'\in P_{U}(S^{(k_{s}:k_{l},1)})} t_{h}t_{h'} \mathrm{Tr}\left[\mathrm{Tr}_{\overline{(h \cup h') / (h \cap h')}}\left[\rho_{0}\right]\mathrm{Tr}_{\overline{(h \cup h') / (h \cap h')}}\left[\rho_{0}\right]\right] \\
    & \qquad  \qquad  \qquad  \qquad  \qquad  \qquad \qquad  \qquad  \qquad \qquad  \qquad   \qquad  \qquad -\frac{2}{2^m} \Biggl(\sum_{h\in P_{U}(S^{(k_{s}:k_{l},1)})} t_{h} \mathrm{Tr}\left[\mathrm{Tr}_{\bar{h}}[ \rho_{0} ]\mathrm{Tr}_{\bar{h}}[\rho_{0} ]\right]\Biggr) + \frac{1}{2^{2m}} \Biggr).\\
\end{split}
\end{equation}

Next we work on the second term  of Eq.\eqref{eq:var_aldqfk_ala_origin}.
We here again integrate the term over $\tilde{W}_{k,d}(\bm{x},\theta_{i})$ first. 
\begin{equation}
\begin{split}
    & Var_{a,2} \\
    & = \frac{1}{2} \biggl\langle \mathrm{Tr}\left[ \tilde{W}_{k,d}(\bm{x},\theta_{i})\tilde{\rho}_{0,B_{l}}^{(1)}\tilde{W}_{k,d}^{\dagger}(\bm{x},\theta_{i})B_{\theta_{i}}  \right]  \mathrm{Tr}\left[ \tilde{W}_{k,d}(\bm{x},\theta_{i})\tilde{\rho}_{0,B_{r}}^{(1)}\tilde{W}_{k,d}^{\dagger}(\bm{x},\theta_{i})B_{\theta_{i}} \right]  \biggr\rangle_{\tilde{W}_{k,d}(\bm{x},\theta_{i})} \\
    &= \frac{1}{2}\sum_{\bm{p},\bm{q},\bm{p'},\bm{q'}} \biggl\langle \mathrm{Tr}\left[ \tilde{W}_{k,d}(\bm{x},\theta_{i})\tilde{\rho}_{0,B_{l},\bm{q}\bm{p}}^{(1)}\tilde{W}_{k,d}^{\dagger}(\bm{x},\theta_{i})B_{\theta_{i},\bm{p}\bm{q}}  \right] \mathrm{Tr}\left[ \tilde{W}_{k,d}(\bm{x},\theta_{i})\tilde{\rho}_{0,B_{r},\bm{q'}\bm{p'}}^{(1)}\tilde{W}_{k,d}^{\dagger}(\bm{x},\theta_{i})B_{\theta_{i},\bm{p'}\bm{q'}}  \right]\biggr\rangle_{\tilde{W}_{k,d}(\bm{x},\theta_{i})} \\
    &= \frac{1}{2} \sum_{\bm{p},\bm{q},\bm{p'},\bm{q'}} \Biggl( \frac{1}{2^{2m}-1} \left(\mathrm{Tr}\left[\tilde{\rho}_{0,B_{l},\bm{q}\bm{p}}^{(1)}\right]\mathrm{Tr}\left[B_{\theta_{i},\bm{p}\bm{q}}\right]\mathrm{Tr}\left[\tilde{\rho}_{0,B_{r},\bm{q'}\bm{p'}}^{(1)}\right]\mathrm{Tr}\left[B_{\theta_{i},\bm{p'}\bm{q'}}\right]+\mathrm{Tr}\left[\tilde{\rho}_{0,B_{l},\bm{q}\bm{p}}^{(1)}\tilde{\rho}_{0,B_{r},\bm{q'}\bm{p'}}^{(1)}\right]\mathrm{Tr}\left[B_{\theta_{i},\bm{p}\bm{q}}B_{\theta_{i},\bm{p'}\bm{q'}}\right]\right)\\
    & \qquad -\frac{1}{2^{m}\left(2^{2m}-1\right)}\left(\mathrm{Tr}\left[\tilde{\rho}_{0,B_{l},\bm{q}\bm{p}}^{(1)}\right]\mathrm{Tr}\left[\tilde{\rho}_{0,B_{r},\bm{q'}\bm{p'}}^{(1)}\right]\mathrm{Tr}\left[B_{\theta_{i},\bm{p}\bm{q}}B_{\theta_{i},\bm{p'}\bm{q'}}\right]+\mathrm{Tr}\left[\tilde{\rho}_{0,B_{l},\bm{q}\bm{p}}^{(1)}\tilde{\rho}_{0,B_{r},\bm{q'}\bm{p'}}^{(1)}\right]\mathrm{Tr}\left[B_{\theta_{i},\bm{p}\bm{q}}\right]\mathrm{Tr}\left[B_{\theta_{i},\bm{p'}\bm{q'}}\right]\right) \Biggr) \\
    &= \frac{1}{2}\cdot \frac{1}{2^{2m}-1} \sum_{\bm{p},\bm{q},\bm{p'},\bm{q'}} \mathrm{Tr}\left[B_{\theta_{i},\bm{p}\bm{q}}B_{\theta_{i},\bm{p'}\bm{q'}}\right] \left(\mathrm{Tr}\left[\tilde{\rho}_{0,B_{l},\bm{q}\bm{p}}^{(1)}\tilde{\rho}_{0,B_{r},\bm{q'}\bm{p'}}^{(1)}\right]-\frac{1}{2^m} \mathrm{Tr}\left[\tilde{\rho}_{0,B_{l},\bm{q}\bm{p}}^{(1)}\right]\mathrm{Tr}\left[\tilde{\rho}_{0,B_{r},\bm{q'}\bm{p'}}^{(1)}\right] \right)\\
    &= \frac{1}{2}\cdot\frac{2^m}{2^{2m}-1} \sum_{\bm{p},\bm{q},\bm{p'},\bm{q'}} \delta_{(\bm{p},\bm{q})}\delta_{(\bm{p'},\bm{q'})}\left(\mathrm{Tr}\left[\tilde{\rho}_{0,B_{l},\bm{q}\bm{p}}^{(1)}\tilde{\rho}_{0,B_{r},\bm{q'}\bm{p'}}^{(1)}\right]-\frac{1}{2^m} \mathrm{Tr}\left[\tilde{\rho}_{0,B_{l},\bm{q}\bm{p}}^{(r)}\right]\mathrm{Tr}\left[\tilde{\rho}_{0,B_{l},\bm{q'}\bm{p'}}^{(1)}\right] \right)\\
    &= \frac{1}{2}\cdot\frac{2^m}{2^{2m}-1} \sum_{\bm{p},\bm{p'}} \left(\mathrm{Tr}\left[\tilde{\rho}_{0,B_{l},\bm{p}\bm{p}}^{(1)}\tilde{\rho}_{0,B_{r},\bm{p'}\bm{p'}}^{(1)}\right]-\frac{1}{2^m} \mathrm{Tr}\left[\tilde{\rho}_{0,B_{l},\bm{p}\bm{p}}^{(1)}\right]\mathrm{Tr}\left[\tilde{\rho}_{0,B_{r},\bm{p'}\bm{p'}}^{(1)}\right] \right),\\
\end{split}
\end{equation}
where $\tilde{\rho}_{0,B_{r}}^{(1)}= V_{r}(\bm{x},\bm{\theta}) \rho_{0}\tilde{B}_{\bm{x'},\theta_{i}} V_{r}^{\dagger}(\bm{x},\bm{\theta})$.
Also we define $\tilde{\rho}_{0,B_{i},\bm{q}\bm{p}}^{(1)}=\mathrm{Tr}_{\bar{S}_{(k,d)}}\left[\left(\ket{\bm{p}}\bra{\bm{q}}\otimes \mathbb{I}_{S_{(k,d)}}\right)\tilde{\rho}_{0,B_{i}}^{(1)}\right]$ for $i\in\{l,r\}$ and $B_{\theta_{i},\bm{p}\bm{q}}=\mathrm{Tr}_{\bar{S}_{(k,d)}}\left[\left(\ket{\bm{q}}\bra{\bm{p}}\otimes \mathbb{I}_{S_{(k,d)}}\right)B_{\theta_{i}}\right]$.
Note that the terms containing the remaining unitary blocks are rewritten as 
\begin{equation}
\label{eq:qfk_var_rest3}
\begin{split}
    \sum_{\bm{p},\bm{p'}} \mathrm{Tr}\left[\tilde{\rho}_{0,B_{l},\bm{p}\bm{p}}^{(1)}\right]\mathrm{Tr}\left[\tilde{\rho}_{0,B_{r},\bm{p'}\bm{p'}}^{(1)}\right] &= \sum_{\bm{p},\bm{p'}} \mathrm{Tr}\left[\mathrm{Tr}_{\bar{S}_{(k,d)}}\left[\left(\ket{\bm{p}}\bra{\bm{p}}\otimes \mathbb{I}_{S_{(k,d)}}\right)\tilde{\rho}_{0,B_{l}}^{(1)}\right]\right]\mathrm{Tr}\left[\mathrm{Tr}_{\bar{S}_{(k,d)}}\left[\left(\ket{\bm{p'}}\bra{\bm{p'}}\otimes \mathbb{I}_{S_{(k,d)}}\right)\tilde{\rho}_{0,B_{r}}^{(1)}\right]\right] \\
    & = \mathrm{Tr}\left[\tilde{\rho}_{0,B_{l}}^{(1)}\right]\mathrm{Tr}\left[\tilde{\rho}_{0,B_{r}}^{(1)}\right]\\
    & = \mathrm{Tr}\left[V_{r}(\bm{x},\bm{\theta}) \tilde{B}_{\bm{x'},\theta_{i}}\rho_{0} V_{r}^{\dagger}(\bm{x},\bm{\theta})\right]\mathrm{Tr}\left[V_{r}(\bm{x},\bm{\theta}) \rho_{0}\tilde{B}_{\bm{x'},\theta_{i}} V_{r}^{\dagger}(\bm{x},\bm{\theta})\right]\\
    & = \mathrm{Tr}\left[\tilde{B}_{\bm{x'},\theta_{i}}\rho_{0} \right]\mathrm{Tr}\left[\tilde{B}_{\bm{x'},\theta_{i}}\rho_{0} \right]\\
\end{split}
\end{equation}
and 
\begin{equation}
\label{eq:qfk_var_rest4}
\begin{split}
    \sum_{\bm{p},\bm{p'}} \mathrm{Tr}\left[\tilde{\rho}_{0,B_{l},\bm{p}\bm{p}}^{(1)}\tilde{\rho}_{0,B_{r},\bm{p'}\bm{p'}}^{(1)}\right] &= \sum_{\bm{p},\bm{p'}} \mathrm{Tr}\left[\mathrm{Tr}_{\bar{S}_{(k,d)}}\left[\left(\ket{\bm{p}}\bra{\bm{p}}\otimes \mathbb{I}_{S_{(k,d)}}\right)\tilde{\rho}_{0,B_{l}}^{(1)}\right]\mathrm{Tr}_{\bar{S}_{(k,d)}}\left[\left(\ket{\bm{p'}}\bra{\bm{p'}}\otimes \mathbb{I}_{S_{(k,d)}}\right)\tilde{\rho}_{0,B_{r}}^{(1)}\right]\right] \\
    & = \mathrm{Tr}\left[\mathrm{Tr}_{\bar{S}_{(k,d)}}\left[\tilde{\rho}_{0,B_{l}}^{(1)}\right]\mathrm{Tr}_{\bar{S}_{(k,d)}}\left[\tilde{\rho}_{0,B_{r}}^{(1)}\right]\right]\\
    & = \mathrm{Tr}\left[\mathrm{Tr}_{\bar{S}_{(k,d)}}\left[V_{r}(\bm{x},\bm{\theta}) \tilde{B}_{\bm{x'},\theta_{i}}\rho_{0} V_{r}^{\dagger}(\bm{x},\bm{\theta})\right]\mathrm{Tr}_{\bar{S}_{(k,d)}}\left[V_{r}(\bm{x},\bm{\theta}) \rho_{0}\tilde{B}_{\bm{x'},\theta_{i}} V_{r}^{\dagger}(\bm{x},\bm{\theta})\right]\right].\\
\end{split}
\end{equation}
From the results, we can see that Eq.~\eqref{eq:qfk_var_rest3} is the same as Eq.~\eqref{eq:qfk_var_rest1}.
As for Eq.~\eqref{eq:qfk_var_rest4}, by the integration over $V_{r}(\bm{x},\bm{\theta})$, we have
\begin{equation}
\begin{split}
    &\left\langle \mathrm{Tr}\left[\mathrm{Tr}_{\bar{S}_{(k,d)}}[ V_{r}(\bm{x},\bm{\theta}) \tilde{B}_{\bm{x'},\theta_{i}}\rho_{0} V_{r}^{\dagger}(\bm{x},\bm{\theta})]\mathrm{Tr}_{\bar{S}_{(k,d)}}[ V_{r}(\bm{x},\bm{\theta}) \rho_{0}\tilde{B}_{\bm{x'},\theta_{i}} V_{r}^{\dagger}(\bm{x},\bm{\theta})]\right] \right\rangle_{V_{r}(\bm{x},\bm{\theta})} \\
    &=\sum_{h\in P_{U}(S^{(k_{s}:k_{l},1)})} t_{h} \mathrm{Tr}\left[\mathrm{Tr}_{\bar{h}}[ \tilde{B}_{\bm{x'},\theta_{i}}\rho_{0} ]\mathrm{Tr}_{\bar{h}}[ \rho_{0}\tilde{B}_{\bm{x'},\theta_{i}} ]\right].
\end{split}
\end{equation}

Subsequently, we integrate the quantity over $U_{1:i}(\bm{x'},\bm{\theta})=\tilde{W}_{k,d}(\bm{x'},\theta_{i})V_{r}(\bm{x'},\bm{\theta})$.
Then we obtain

\begin{equation}
\begin{split}
\label{eq:smiddle_rest1}
    & \left\langle\mathrm{Tr}\left[\mathrm{Tr}_{\bar{h}}\left[\tilde{B}_{x',\theta_{i}}\rho_{0}\right]\mathrm{Tr}_{\bar{h}}\left[\rho_{0}\tilde{B}_{x',\theta_{i}} \right]\right]\right\rangle_{\tilde{W}_{k,d}(\bm{x'},\theta_{i}), V_{r}(\bm{x'},\bm{\theta})}\\
    &= \biggl\langle \mathrm{Tr}\left[\mathrm{Tr}_{\bar{h}}\Bigl[V_{r}^{\dagger}(\bm{x'},\bm{\theta})\tilde{W}_{k,d}^{\dagger}(\bm{x'},\theta_{i})B_{\theta_{i}}\tilde{W}_{k,d}(\bm{x'},\theta_{i})V_{r}(\bm{x'},\bm{\theta})\rho_{0}\right]\\
    & \qquad \qquad \times \mathrm{Tr}_{\bar{h}}\left[\rho_{0}V_{r}^{\dagger}(\bm{x'},\bm{\theta})\tilde{W}_{k,d}^{\dagger}(\bm{x'},\theta_{i})B_{\theta_{i}}\tilde{W}_{k,d}(\bm{x'},\theta_{i})V_{r}(\bm{x'},\bm{\theta}) \right]\Bigr]\biggr\rangle_{\tilde{W}_{k,d}(\bm{x'},\theta_{i}),V_{r}(\bm{x'},\bm{\theta})} \\
    &= \biggl\langle \mathrm{Tr}\Bigl[\left(V_{r}^{\dagger}(\bm{x'},\bm{\theta})\tilde{W}_{k,d}^{\dagger}(\bm{x'},\theta_{i})B_{\theta_{i}}\tilde{W}_{k,d}(\bm{x'},\theta_{i})V_{r}(\bm{x'},\bm{\theta})\rho_{0} \otimes \rho_{0}V_{r}^{\dagger}(\bm{x'},\bm{\theta})\tilde{W}_{k,d}^{\dagger}(\bm{x'},\theta_{i})B_{\theta_{i}}\tilde{W}_{k,d}(\bm{x'},\theta_{i})V_{r}(\bm{x'},\bm{\theta}) \right) \\
    & \qquad \qquad \qquad \qquad \qquad \qquad \qquad \qquad \qquad \qquad \qquad \qquad \qquad \qquad \times \left(Swap_{h_{1}\otimes h_{2}} \otimes \mathbb{I}_{\bar{h}_{1}\otimes \bar{h}_{2}} \right)\biggr\rangle_{\tilde{W}_{k,d}(\bm{x'},\theta_{i}),V_{r}(\bm{x'},\bm{\theta})} \\
    &= \frac{2^{m}}{2^{2m}-1} \biggl(\Bigl\langle\mathrm{Tr}\Bigl[\left(\mathbb{I} \otimes \rho_{0}\right)\left(V_{r}^{\dagger}(\bm{x'},\bm{\theta})\otimes V_{r}^{\dagger}(\bm{x'},\bm{\theta})\right) \left(Swap_{S_{(k,d),1} \otimes S_{(k,d),2}} \otimes \mathbb{I}_{\bar{S}_{(k,d),1} \otimes \bar{S}_{(k,d),2}}\right) \\
    & \qquad \qquad  \qquad \qquad \qquad \qquad \qquad \qquad  \times \left(V_{r}(\bm{x'},\bm{\theta})\otimes V_{r}(\bm{x'},\bm{\theta})\right)\left(\rho_{0} \otimes \mathbb{I}\right) \left(Swap_{h_{1}\otimes h_{2}} \otimes \mathbb{I}_{\bar{h}_{1}\otimes \bar{h}_{2}} \right)\Bigr] \Bigr\rangle_{V_{r}(\bm{x},\bm{\theta})} \biggr)\\
    & \qquad \qquad \qquad   -\frac{1}{2^{m}} \mathrm{Tr}\left[\mathrm{Tr}_{\bar{h}}\left[\rho_{0}\right]\mathrm{Tr}_{\bar{h}}\left[\rho_{0}\right]\right] \Bigr)\\
    &= \frac{2^{m}}{2^{2m}-1} \Bigl(\mathrm{Tr}\Bigl[\left(\mathbb{I} \otimes \rho_{0}\right)\left(\sum_{h'\in P_{U}(S^{(k_{s}:k_{l},1)})} t_{h'}\left(Swap_{h'_{1} \otimes h'_{2}} \otimes \mathbb{I}_{\bar{h'}_{1} \otimes \bar{h'}_{2}}\right)\right)\left(\rho_{0} \otimes \mathbb{I}\right) \left(Swap_{h_{1}\otimes h_{2}} \otimes \mathbb{I}_{\bar{h}_{1}\otimes \bar{h}_{2}} \right)\Bigr] \Bigr)\\
    & \qquad \qquad \qquad   -\frac{1}{2^{m}} \mathrm{Tr}\left[\mathrm{Tr}_{\bar{h}}\left[\rho_{0}\right]\mathrm{Tr}_{\bar{h}}\left[\rho_{0}\right]\right] \Bigr)\\
    &= \frac{2^{m}}{2^{2m}-1} \Biggl(\Biggl(\sum_{h'\in P_{U}(S^{(k_{s}:k_{l},1)})} t_{h'}\mathrm{Tr}\Bigl[\left(\mathbb{I} \otimes \rho_{0}\right)\left(Swap_{h'_{1} \otimes h'_{2}} \otimes \mathbb{I}_{\bar{h'}_{1} \otimes \bar{h'}_{2}}\right)\left(\rho_{0} \otimes \mathbb{I}\right) \left(Swap_{h_{1}\otimes h_{2}} \otimes \mathbb{I}_{\bar{h}_{1}\otimes \bar{h}_{2}} \right)\Bigr] \Biggr)\Biggr)\\
    & \qquad \qquad \qquad   -\frac{1}{2^{m}} \mathrm{Tr}\left[\mathrm{Tr}_{\bar{h}}\left[\rho_{0}\right]\mathrm{Tr}_{\bar{h}}\left[\rho_{0}\right]\right] \Bigr),\\
\end{split}
\end{equation}
where we utilize Eq.~\eqref{eq:int_tensor_haar_res} for the fourth equality.
Here, the quantity $A_{\rho,Swap}\equiv\mathrm{Tr}[(\mathbb{I} \otimes \rho_{0})(Swap_{h'_{1} \otimes h'_{2}} \otimes \mathbb{I}_{\bar{h'}_{1} \otimes \bar{h'}_{2}})(\rho_{0} \otimes \mathbb{I})(Swap_{h_{1}\otimes h_{2}} \otimes \mathbb{I}_{\bar{h}_{1}\otimes \bar{h}_{2}})] $ can be rewritten as
\begin{equation}
    A_{\rho,Swap}=\mathrm{Tr}\left[\mathbb{I}_{h \cap h'}\right] \mathrm{Tr}\left[\mathrm{Tr}_{\overline{h \cup h'}}\left[\rho_{0}\right]\mathrm{Tr}_{\overline{h \cup h'}}\left[\rho_{0}\right]\right].
\end{equation}
Note that we here define $\mathrm{Tr}[\mathbb{I}_{\emptyset}]=1$.
Thus, substituting the equality into Eq.\eqref{eq:smiddle_rest1}, we get
\begin{equation}
\begin{split}
\label{eq:f_sterm_middle_integrated}
    & \left\langle\mathrm{Tr}\left[\mathrm{Tr}_{\bar{h}}\left[\tilde{B}_{x',\theta_{i}}\rho_{0}\right]\mathrm{Tr}_{\bar{h}}\left[\rho_{0}\tilde{B}_{x',\theta_{i}} \right]\right]\right\rangle_{U_{1:i}(\bm{x'},\bm{\theta})}\\
    &=\frac{2^{m}}{2^{2m}-1} \Biggl(\Biggl(\sum_{h'\in P_{U}(S^{(k_{s}:k_{l},1)})} t_{h'}\mathrm{Tr}\left[\mathbb{I}_{h \cap h'}\right] \mathrm{Tr}\left[\mathrm{Tr}_{\overline{h \cup h'}}\left[\rho_{0}\right]\mathrm{Tr}_{\overline{h \cup h'}}\left[\rho_{0}\right]\right]\Biggr) -\frac{1}{2^{m}} \mathrm{Tr}\left[\mathrm{Tr}_{\bar{h}}\left[\rho_{0}\right]\mathrm{Tr}_{\bar{h}}\left[\rho_{0}\right]\right] \Biggr).
\end{split}
\end{equation}
Hence, using Eqs.~\eqref{eq:s_in_fterm_middle_integrated} and \eqref{eq:f_sterm_middle_integrated}, we have
\begin{equation}
\label{eq:qfk_v_t2}
\begin{split}
    & Var_{a,2} = \frac{1}{2}\left(\frac{2^m}{2^{2m}-1}\right)^{2}  \Biggl( \sum_{h\in P_{U}(S^{(k_{s}:k_{l},1)})} \sum_{h'\in P_{U}(S^{(k_{s}:k_{l},1)})} t_{h}t_{h'} \mathrm{Tr}\left[\mathbb{I}_{h \cap h'}\right]\mathrm{Tr}\left[\mathrm{Tr}_{\overline{h \cup h'}}\left[\rho_{0}\right]\mathrm{Tr}_{\overline{h \cup h'}}\left[\rho_{0}\right]\right] \\
    & \qquad  \qquad  \qquad  \qquad  \qquad  \qquad \qquad  \qquad  \qquad \qquad  \qquad   \qquad  \qquad  \qquad  -\frac{2}{2^m} \Biggl(\sum_{h\in P_{U}(S^{(k_{s}:k_{l},1)})} t_{h} \mathrm{Tr}\left[\mathrm{Tr}_{\bar{h}}[ \rho_{0} ]\mathrm{Tr}_{\bar{h}}[\rho_{0} ]\right]\Biggr) + \frac{1}{2^{2m}} \Biggr). \\
\end{split}
\end{equation}

Lastly, the expectation of the third term in Eq.\eqref{eq:var_aldqfk_ala_origin} is the same as that of the first term, due to the symmetry.
Thus, we have
\begin{equation}
\label{eq:qfk_v_t3}
\begin{split}
    & Var_{a,3} \\
    &=\frac{1}{4}\left(\frac{2^m}{2^{2m}-1}\right)^{2}  \Biggl( \sum_{h\in P_{U}(S^{(k_{s}:k_{l},1)})} \sum_{h'\in P_{U}(S^{(k_{s}:k_{l},1)})} t_{h}t_{h'} \mathrm{Tr}\left[\mathrm{Tr}_{\overline{(h \cup h') / (h \cap h')}}\left[\rho_{0}\right]\mathrm{Tr}_{\overline{(h \cup h') / (h \cap h')}}\left[\rho_{0}\right]\right] \\
    & \qquad  \qquad  \qquad  \qquad  \qquad  \qquad \qquad  \qquad  \qquad \qquad  \qquad   \qquad  \qquad -\frac{2}{2^m} \Biggl(\sum_{h\in P_{U}(S^{(k_{s}:k_{l},1)})} t_{h} \mathrm{Tr}\left[\mathrm{Tr}_{\bar{h}}[ \rho_{0} ]\mathrm{Tr}_{\bar{h}}[\rho_{0} ]\right]\Biggr) + \frac{1}{2^{2m}} \Biggr).\\
\end{split}
\end{equation}
Consequently, by summing up Eqs.~\eqref{eq:qfk_v_t1},~\eqref{eq:qfk_v_t2} and ~\eqref{eq:qfk_v_t3}, the variance of the ALDQFK is expressed as
\begin{equation}
\begin{split}
\label{eq:var_ala_aldqfk}
    &Var\left[k_{QF}^{A}\right] \\
    &= Var_{a,1}+Var_{a,2}+Var_{a,3} \\
    &=\frac{1}{2}\left(\frac{2^m}{2^{2m}-1}\right)^{2}   \Biggl( \sum_{h\in P_{U}(S^{(k_{s}:k_{l},1)})} \sum_{h'\in P_{U}(S^{(k_{s}:k_{l},1)})} t_{h}t_{h'} \biggl( \mathrm{Tr}\left[\mathrm{Tr}_{\overline{(h \cup h') / (h \cap h')}}\left[\rho_{0}\right]\mathrm{Tr}_{\overline{(h \cup h') / (h \cap h')}}\left[\rho_{0}\right]\right] \\
    & \qquad \qquad \qquad \qquad  +\mathrm{Tr}\left[\mathbb{I}_{h \cap h'}\right] \mathrm{Tr}\left[\mathrm{Tr}_{\overline{h \cup h'}}\left[\rho_{0}\right]\mathrm{Tr}_{\overline{h \cup h' }}\left[\rho_{0}\right]\right]\biggr)   -\frac{4}{2^m} \Biggl(\sum_{h\in P_{U}(S^{(k_{s}:k_{l},1)})} t_{h} \mathrm{Tr}\left[\mathrm{Tr}_{\bar{h}}[ \rho_{0} ]\mathrm{Tr}_{\bar{h}}[\rho_{0} ]\right]\Biggr) + \frac{2}{2^{2m}} \Biggr).
\end{split}
\end{equation}

Further, we will obtain the lower bound of Eq.~\eqref{eq:var_ala_aldqfk}, assuming the initial state satisfies the following equalities;
\begin{equation}
\begin{split}
    &\mathrm{Tr}\left[\mathrm{Tr}_{\overline{(h \cup h') / (h \cap h')}}\left[\rho_{0}\right]\mathrm{Tr}_{\overline{(h \cup h') / (h \cap h')}}\left[\rho_{0}\right]\right] \ge \mathrm{Tr}\left[\mathrm{Tr}_{\bar{h}}\left[\rho_{0}\right]\mathrm{Tr}_{\bar{h}}\left[\rho_{0}\right]\right]\mathrm{Tr}\left[\mathrm{Tr}_{\bar{h'}}\left[\rho_{0}\right]\mathrm{Tr}_{\bar{h'}}\left[\rho_{0}\right]\right], \\
    & \mathrm{Tr}\left[\mathrm{Tr}_{\overline{h \cup h'}}\left[\rho_{0}\right]\mathrm{Tr}_{\overline{h \cup h' }}\left[\rho_{0}\right]\right] \ge \mathrm{Tr}\left[\mathrm{Tr}_{\bar{h}}\left[\rho_{0}\right]\mathrm{Tr}_{\bar{h}}\left[\rho_{0}\right]\right]\mathrm{Tr}\left[\mathrm{Tr}_{\bar{h'}}\left[\rho_{0}\right]\mathrm{Tr}_{\bar{h'}}\left[\rho_{0}\right]\right].
\end{split}
\end{equation}
Note that the initial states that satisfy above conditions include the tensor product states of arbitrary single-qubit pure states $\{\rho_{0,i}\}_{i=1}^{n}$, i.e., $\rho_{0} = \rho_{0,1} \otimes \rho_{0,2} \otimes \ldots \otimes \rho_{0,i} \otimes \ldots \otimes \rho_{0,n} $, and the completely mixed states, while it is unclear if any quantum states fulfill the properties.
Then Eq.~\eqref{eq:var_ala_aldqfk} can be expressed as
\begin{equation}
\begin{split}
    &Var\left[k_{QF}^{A}\right] \\
    & \ge \frac{1}{2}\left(\frac{2^m}{2^{2m}-1}\right)^{2}   \Biggl( \sum_{h\in P_{U}(S^{(k_{s}:k_{l},1)})} \sum_{h'\in P_{U}(S^{(k_{s}:k_{l},1)})} t_{h}t_{h'}  \mathrm{Tr}\left[\mathrm{Tr}_{\bar{h}}\left[\rho_{0}\right]\mathrm{Tr}_{\bar{h}}\left[\rho_{0}\right]\right]\mathrm{Tr}\left[\mathrm{Tr}_{\bar{h'}}\left[\rho_{0}\right]\mathrm{Tr}_{\bar{h'}}\left[\rho_{0}\right]\right] \left(1+\mathrm{Tr}\left[\mathbb{I}_{h \cap h'}\right]\right) \\
    & \qquad  \qquad \qquad  \qquad \qquad  \qquad \qquad  \qquad \qquad  \qquad \qquad  \qquad   -\frac{4}{2^m} \Biggl(\sum_{h\in P_{U}(S^{(k_{s}:k_{l},1)})} t_{h} \mathrm{Tr}\left[\mathrm{Tr}_{\bar{h}}[ \rho_{0} ]\mathrm{Tr}_{\bar{h}}[\rho_{0} ]\right]\Biggr) + \frac{2}{2^{2m}} \Biggr)\\
    &=\frac{1}{2}\left(\frac{2^m}{2^{2m}-1}\right)^{2}   \Biggl( 2 \Biggl( \Biggl( \sum_{h\in P_{U}(S^{(k_{s}:k_{l},1)})} t_{h}\mathrm{Tr}\left[\mathrm{Tr}_{\bar{h}}[ \rho_{0} ]\mathrm{Tr}_{\bar{h}}[\rho_{0} ]\right] \Biggr) -\frac{1}{2^m} \Biggr)^{2} \\
    & \qquad \qquad \qquad \qquad +  \sum_{h\in P_{U}(S^{(k_{s}:k_{l},1)})} \sum_{h'\in P_{U}(S^{(k_{s}:k_{l},1)})} t_{h}t_{h'} \mathrm{Tr}\left[\mathrm{Tr}_{\bar{h}}\left[\rho_{0}\right]\mathrm{Tr}_{\bar{h}}\left[\rho_{0}\right]\right]\mathrm{Tr}\left[\mathrm{Tr}_{\bar{h'}}\left[\rho_{0}\right]\mathrm{Tr}_{\bar{h'}}\left[\rho_{0}\right]\right] \left(\mathrm{Tr}\left[\mathbb{I}_{h \cap h'}\right]-1\right) \Biggr) \\
    & \ge \frac{1}{2}\left(\frac{2^m}{2^{2m}-1}\right)^{2} \Biggl(\sum_{h\in P_{U}(S^{(k_{s}:k_{l},1)})} \sum_{h'\in P_{U}(S^{(k_{s}:k_{l},1)})} t_{h}t_{h'} \mathrm{Tr}\left[\mathrm{Tr}_{\bar{h}}\left[\rho_{0}\right]\mathrm{Tr}_{\bar{h}}\left[\rho_{0}\right]\right]\mathrm{Tr}\left[\mathrm{Tr}_{\bar{h'}}\left[\rho_{0}\right]\mathrm{Tr}_{\bar{h'}}\left[\rho_{0}\right]\right] \left(\mathrm{Tr}\left[\mathbb{I}_{h \cap h'}\right]-1\right) \Biggr) \\
    & \ge \frac{1}{2}\left(\frac{2^m}{2^{2m}-1}\right)^{2}  t_{S_{(k_{s}:k_{l},1)}}^{2}\mathrm{Tr}\left[\mathrm{Tr}_{\overline{S_{(k_{s}:k_{l},1)}}}\left[\rho_{0}\right]\mathrm{Tr}_{\overline{S_{(k_{s}:k_{l},1)}}}\left[\rho_{0}\right]\right]^{2} \left(\mathrm{Tr}\left[\mathbb{I}_{S_{(k_{s}:k_{l},1)}}\right]-1\right).\\
\end{split}
\end{equation}
We remind that $t_{S_{(k_{s}:k_{l},1)}}$ differs depending on the position of the unitary $W_{k,d}(\bm{x},\bm{\theta}_{k,d})$.
Actually, the lowest value is attained when $k$ satisfies $k_s\ge1$ and $k_l\le d$ (i.e., $W_{k,d}(\bm{x},\bm{\theta}_{k,d})$ is located in the middle of a layer).
Thus, the lower bound of the variance for the ALDQFK using the ALA reads
\begin{equation}
    Var\left[k_{QF}^{A}\right]  \ge \frac{ 2^{md} -1}{2\left(2^{2m}-1\right)^{2}\left(2^{m}+1\right)^{4(d-1)}}.
\end{equation}
Here we also utilize the bound of the purity, $1/d\le \mathrm{Tr}[\rho^{2}]\le 1$ with the $d$-dimensional quantum state.

We also give a lower bound on the variance when the initial state is a tensor product of arbitrary single-qubit pure states, which is commonly and practically used as the initial state.
In this case, the lower bound on the variance of the ALDQFK with the ALA can be written as
\begin{equation}
    Var\left[k_{QF}^{A}\right]  \ge \frac{2^{2md} \left(2^{md} -1\right)}{2\left(2^{2m}-1\right)^{2}\left(2^{m}+1\right)^{4(d-1)}}.
\end{equation}

\section{Details of the Numerical experiments}
This section presents the details of the numerical experiments shown in the main text.
Specifically, we describe the settings and additional results for the numerical study in the following subsections of the main test: “Motivating examples”, “Numerical experiments” and “Expressivity comparison of the fidelity-based QK”. 
Additionally, we show the geometric difference of the Gram matrices given by the fidelity-based QK and the ALDQFK to see the difference when high-dimensional data is used. 
We note that Cirq \cite{cirq_developers_2022_6599601} is used to compute the quantum states for all numerical experiments in our study.

\subsection{Numerical study in “Motivating examples”}
We show the setup of the numerical experiment that illustrates examples of the vanishing similarity issue for the fidelity-based QK. 
To demonstrate the issue, we calculate the expectation and the variance of the fidelity-based QK using two types of quantum circuits with the number of qubits $n$:
\begin{itemize}
    \item Tensor-product quantum circuits
    \begin{equation}
    U_{TP}(\bm{\alpha})=\otimes_{i=1}^{n} \exp(-i\alpha_{i}Y_{i}/2)\exp(-i\alpha_{i}Z_{i}/2)
    \end{equation}
    with the Pauli operators acting on the $i$-th qubit, $Y_i$ and $Z_i$,
    \item IQP-type quantum circuits \cite{havlivcek2019supervised}
    \begin{equation}
    U_{IQP}(\bm{\alpha})=U_{\phi}(\bm{\alpha})H^{\otimes n}    
    \end{equation}
    with $U_{\phi}(\bm{\alpha})=\exp(\sum_{i=1}^{n} \phi_{i}(\bm{\alpha}) Z_{i} + \sum_{j=1}^{n-1} \phi_{j,j+1}(\bm{\alpha}) Z_{j}Z_{j+1})$ and tensor-product of the Hadamard gates $H^{\otimes n}$.
\end{itemize}
The diagrams of these quantum circuits are shown in Figure~\ref{fig:q_circuit} (a) and (b).
To be more specific, we use these quantum circuits with depth $L=2$, where each layer is composed of an input-embedded circuit and a PQC, and the data re-uploading technique is employed \cite{perez2020data} i.e., $$U(\bm{x},\bm{\theta})=\prod_{d=1}^{L} U_{k}(\bm{\theta}_d)U_{k}(\bm{x}), \qquad k\in\{TP, IQP\}.$$ 
Here, $\bm{\theta}_d$ represents parameters in the $d$-th PQC layer.
As for the IQP-type quantum circuits, $\phi_{i}(\bm{x})=x_i,\phi_{j,j+1}(\bm{x})=(x_j x_{j+1})/\pi$ in the input layers, and $\phi_{i}(\bm{\theta})=\theta_i,\phi_{j,j+1}(\bm{\theta})=\theta_{n+j}$ in the PQC layers.
Note that $\phi_{j,j+1}(\bm{x})$ is slightly modified from the original proposal \cite{havlivcek2019supervised} in such a way that $\phi_{j,j+1}(\bm{x}) \in [-\pi,\pi)$.

\begin{figure}[t]
    \centering
    \includegraphics[keepaspectratio, scale=0.8]{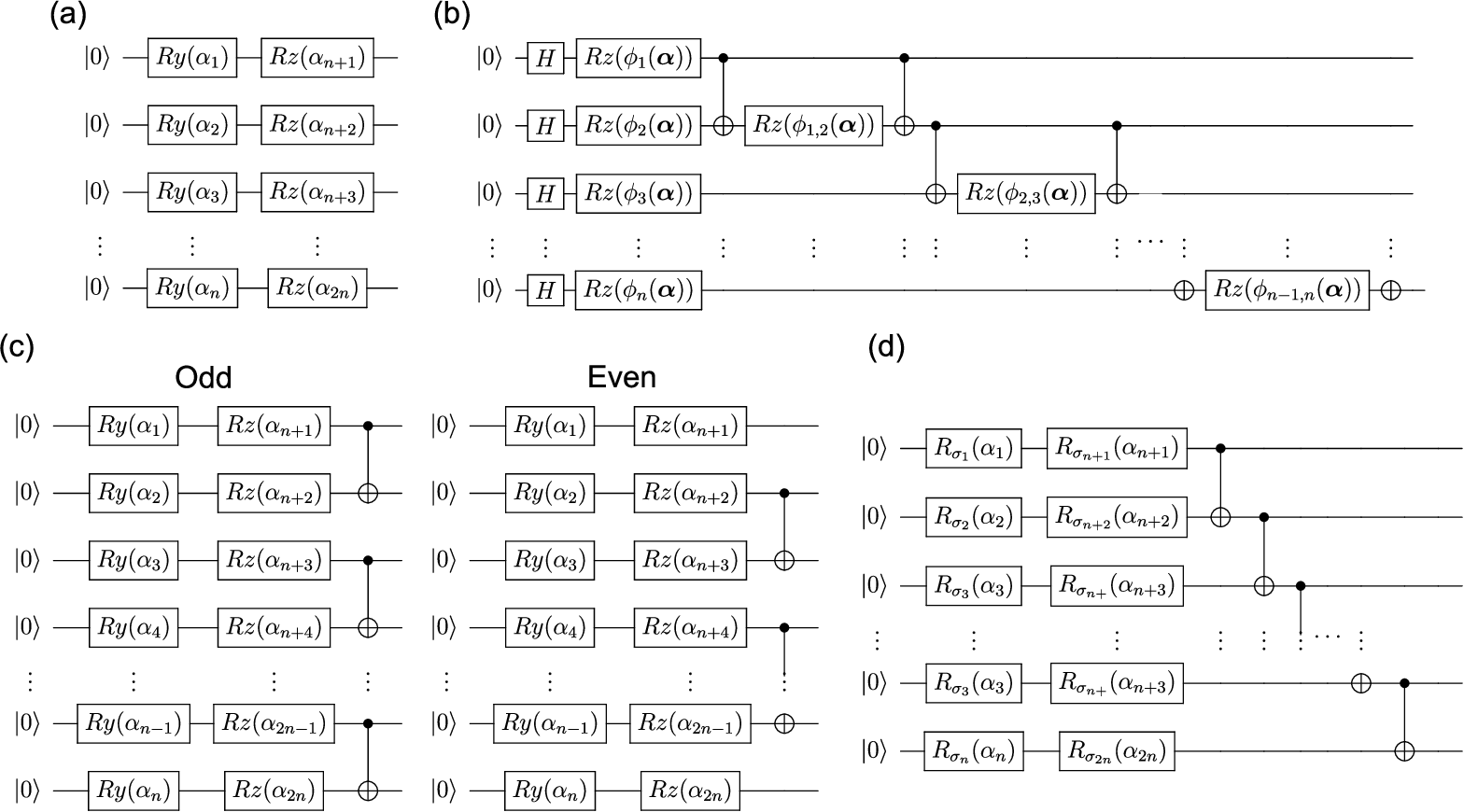}
    \caption{Quantum circuits used in this study. Here, we show one layer of (a) a tensor-product quantum circuit, (b) an IQP-type quantum circuit, (c) an ALA with $2$-qubit local unitary blocks and (d) a HEA serving as a random quantum circuit. As for the ALA, different entanglers are used in the odd and even layers; CNOT gates act on a pair of neighbor qubits alternately. We also note that the circuit (d) contains single-qubit rotation gates whose axes are chosen at random.}
    \label{fig:q_circuit}
\end{figure}

In the experiment, we prepare five sets of 100 data points $\{\bm{x}_i\}_{i=1}^{100}$ randomly generated from the range $[-\pi,\pi)$, where the dimension of the input data is equal to the number of qubits used. 
Likewise, five sets of the parameters in the PQC $\bm{\theta}$ are randomly generated from the same range. 
Throughout this paper, each element of randomly generated input data and parameters in the PQCs ranges form $-\pi$ to  $\pi$, unless otherwise mentioned.
Then we calculate the QK $k_{Q}(\bm{x},\bm{x’})$ with $\bm{x}\neq\bm{x’}$ for all 25 combinations of the input data set and the parameter set, which are used to derive the expectation and the variance.

\subsection{Numerical study in “Numerical experiments”}

In this subsection, we provide the settings of the numerical experiments to validate the Proposition and Theorem in the main text. 
As in "Motivating examples", we calculate the variance of the fidelity-based QK and the ALDQFK using different quantum circuits. 
Here, each layer of the quantum circuits consists of an input-embedded circuit and a PQC, that is, $U(\bm{x},\bm{\theta})=\prod_{d=1}^{L}U (\bm{\theta}_d)U(\bm{x})$, where the tensor-product quantum circuits are used for the input circuit $U(\bm{x})=U_{TP}(\bm{x})$. 
As for the PQCs, we choose the tensor-product quantum circuits, the ALAs composed of $2$-qubit local unitary blocks in Figure~\ref{fig:q_circuit} (c) and the hardware efficient ansatzs (HEAs) in Figure~\ref{fig:q_circuit} (d). 
For the HEAs, we use the fixed entangling gates and randomly chosen single-qubit rotation gates $R_{\sigma_{i}}(\theta)=\exp(-i\theta \sigma_{i}/2), \sigma_{i}\in\{X,Y,Z\}$ to make the circuits serve as the random quantum circuits.
Note that the random quantum circuits and the local unitary blocks in the ALAs might not be $2$-designs. 
Moreover, as is the case for the motivating examples, we prepare five datasets containing randomly generated 100 data points, and five sets of parameters in PQC.
Then the QKs with $\bm{x}\neq\bm{x'}$ calculated for a total of 25 combinations of them are used to obtain the variance.

\subsection{Numerical study in “Expressivity comparison of the fidelity-based QK”}
We give the details of the Fourier analysis and the binary classification tasks using one-dimensional synthesized datasets.

First, we present the details and numerical settings of the Fourier analysis. In the Fourier analysis, the key idea is to utilize the Fourier representation of QKs to numerically obtain the non-zero Fourier coefficients, which can be interpreted as the expressivity of the model \cite{schuld2021effect,schuld2021supervised}. 
The Fourier representation can be expressed as 
\begin{equation}
\label{eq:fourier_repr}
    k(\bm{x},\bm{x'}) = \sum_{\omega,\omega'\in \Omega} e^{i\omega \bm{x}}e^{i\omega' \bm{x}} c_{\omega,\omega'},
\end{equation}
where $c_{\omega,\omega'}\in\mathbb{C}$ satisfying $c_{\omega,\omega'}=c_{-\omega,-\omega'}^{*}$ and $\Omega$ is the set of the integer-valued frequencies. 
Then the non-zero Fourier coefficients $\{c_{\omega,\omega'}| c_{\omega,\omega'}\neq 0\}$ are used to see the expressivity of the QK. 
Namely, the QK with many non-zero coefficients could represent a large class of functions.
However, the Fourier coefficients cannot be obtained analytically. 
Thus, we use the “curve\_fit” function in SciPy \cite{2020SciPy-NMeth} to fit the Gram matrix given by a QK to its Fourier representation, from which the coefficients were obtained. 

In this numerical study, as in the “Numerical experiments”, we use quantum circuits, each layer of which is composed of an input-embedded circuit and a PQC.
Here we use the tensor-product quantum circuits for the input layer, and the ALAs with $2$-qubit local unitary blocks and the HEAs for the PQC layers, with the setting $n=1,2,3$ and $L=2,3,4$.
Also, due to the computational difficulty, 100 data points uniformly distributed in one dimension, and the truncated set of frequencies $\tilde{\Omega}$ ranging from $-12$ to $12$ are used.
Note that the input data $x$ is embedded into the angles of all single-rotation gates of the tensor-product quantum circuits, i.e., $\alpha_i = x$ for all $i$ of $U_{TP}(\bm{\alpha})$. 
In addition, we try different 10 sets of parameters in the PQC for the calculation of the QKs because of the difference in expressivity depending on the parameters. 

Figure~\ref{fig:fourier_analysis} (b) shows the amplitudes of the Fourier coefficients for each QK using the HEAs.
We remark that each Fourier coefficient is aligned on the horizontal axis as in an example for the case $\tilde{\omega}=\{-2,-1,0,1,2\}$ shown in Figure~\ref{fig:fourier_analysis} (a).
As we can see, the fidelity-based QK and the ALDQFK have approximately the same number of non-zero amplitudes up to the order of $10^{-3}$, while there is a slight difference in the amplitude of the coefficients with respect to high frequencies.
This means that these QKs have almost the same expressivity. 
We note that the reason why the amplitudes less than $10^{-3}$ are truncated is discussed later. 
Also, we show the mean absolute error between each QK and its reconstructed one in Table~\ref{tab:mae_fourier}, to exhibit the validity of the numerically obtained Fourier coefficients.

\begin{figure}[htbp]
    \centering
    \includegraphics[keepaspectratio, scale=0.8]{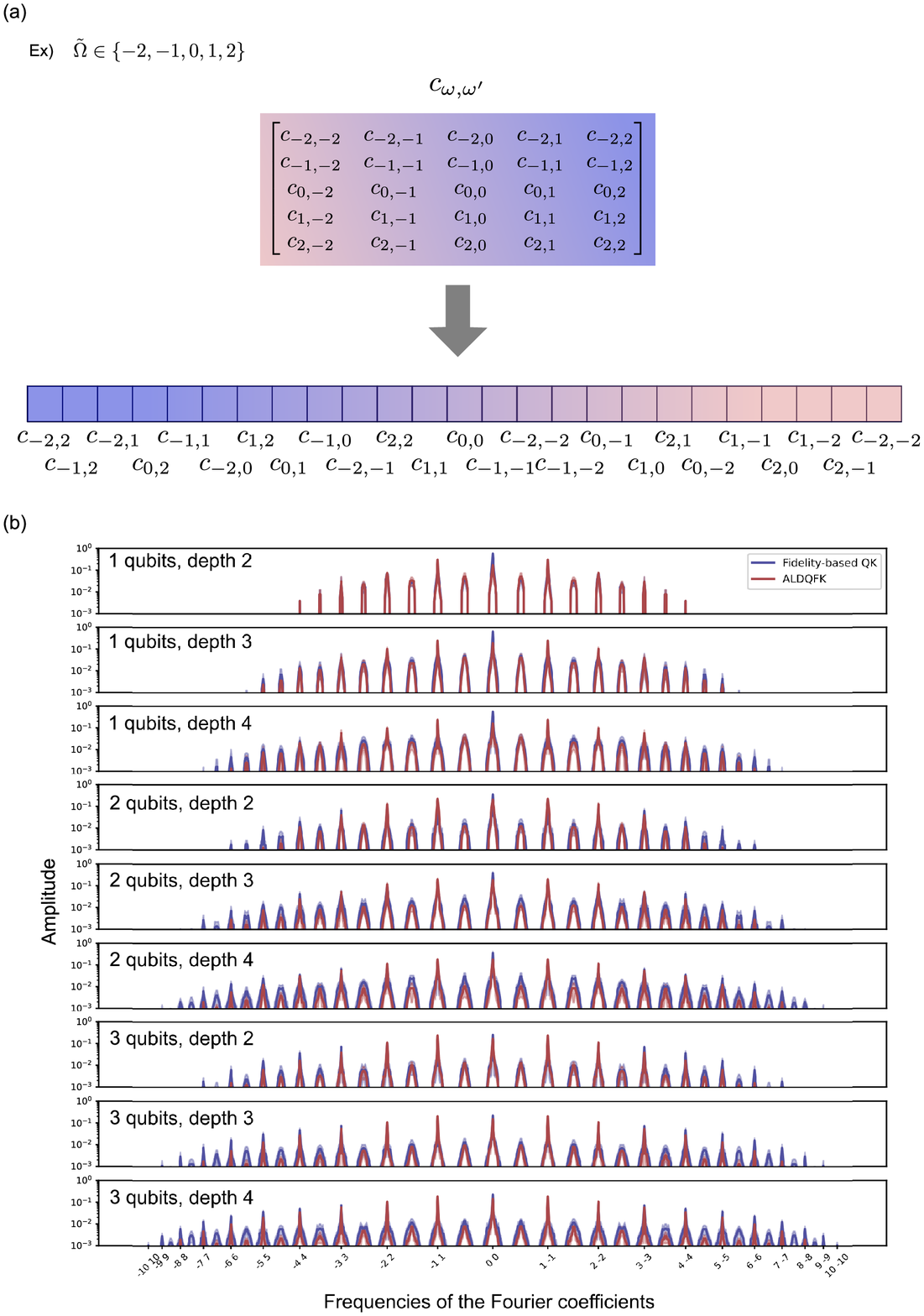}
    \caption{Comparison of the QKs from the perspective of Fourier analysis. (a) How the Fourier coefficients are aligned in one axis for a plot is shown, taking the case for $\tilde{\Omega} \in \{-2,-1,0,1,2\}$ as an example. (b) The amplitudes of the coefficients for the fidelity-based QK and the ALDQFK are shown. Here the HEAs with $n=1,2,3$ and $L=2,3,4$ are used. The shaded regions represent the standard deviation over 10 trials. }
    \label{fig:fourier_analysis}
\end{figure}

\begin{table}[t]
\caption{\label{tab:mae_fourier} List of the mean absolute errors (MAEs) between QKs and their reconstructed ones. Here $(L,n)$ represents a pair of depth and the number of qubits used in the ALA with $2$-qubit local unitary blocks (ALA2) and the HEAs (HEA).}

    \centering
    \begin{tabular}{c|c|c|c|c|c|c|c|c}
    \hline
        & \multicolumn{4}{c|}{Fidelity-based QK} & \multicolumn{4}{|c}{ALDQFK} \\ \hline
        & \multicolumn{2}{|c|}{ALA2} & \multicolumn{2}{|c|}{HEA} & \multicolumn{2}{|c|}{ALA2} & \multicolumn{2}{|c}{HEA} \\ \hline
        $(L,n)$ & Mean & Std & Mean & Std & Mean & Std & Mean & Std \\ \hline
        $(2,1)$ & $2.6\times10^{-8}$ & $2.4\times10^{-9}$ & $2.5\times10^{-8}$ & $2.1\times10^{-9}$ & $1.6\times10^{-8}$ & $2.7\times10^{-9}$ & $1.8\times10^{-8}$ & $2.2\times10^{-9}$ \\ 
        $(2,2)$ & $9.6\times10^{-9}$ & $7.2\times10^{-10}$ & $9.7\times10^{-9}$ & $2.0\times10^{-9}$ & $1.2\times10^{-8}$ & $1.7\times10^{-9}$ & $1.2\times10^{-8}$ & $2.7\times10^{-9}$ \\ 
        $(2,3)$ & $5.3\times10^{-9}$ & $2.1\times10^{-9}$ & $4.6\times10^{-9}$ & $1.4\times10^{-9}$ & $1.1\times10^{-8}$ & $1.7\times10^{-9}$ & $7.2\times10^{-9}$ & $1.7\times10^{-9}$ \\ 
        $(3,1)$ & $2.2\times10^{-8}$ & $1.9\times10^{-9}$ & $2.3\times10^{-8}$ & $9.2\times10^{-10}$ & $1.3\times10^{-8}$ & $1.9\times10^{-9}$ & $1.4\times10^{-8}$ & $1.7\times10^{-9}$ \\ 
        $(3,2)$ & $5.2\times10^{-9}$ & $2.1\times10^{-9}$ & $5.0\times10^{-9}$ & $1.7\times10^{-9}$ & $6.4\times10^{-9}$ & $1.4\times10^{-9}$ & $5.6\times10^{-9}$ & $1.6\times10^{-9}$ \\ 
        $(3,3)$ & $2.4\times10^{-4}$ & $9.0\times10^{-5}$ & $4.2\times10^{-4}$ & $2.3\times10^{-4}$ & $5.6\times10^{-9}$ & $1.6\times10^{-9}$ & $1.6\times10^{-6}$ & $9.2\times10^{-7}$ \\ 
        $(4,1)$ & $1.4\times10^{-8}$ & $2.5\times10^{-9}$ & $1.4\times10^{-8}$ & $1.5\times10^{-9}$ & $8.3\times10^{-9}$ & $1.7\times10^{-9}$ & $1.0\times10^{-8}$ & $1.5\times10^{-9}$ \\ 
        $(4,2)$ & $4.1\times10^{-4}$ & $2.6\times10^{-4}$ & $4.4\times10^{-4}$ & $3.4\times10^{-4}$ & $3.4\times10^{-5}$ & $2.5\times10^{-5}$ & $2.3\times10^{-5}$ & $1.4\times10^{-5}$ \\ 
        $(4,3)$ & $1.3\times10^{-3}$ & $6.0\times10^{-4}$ & $1.4\times10^{-3}$ & $7.4\times10^{-4}$ & $1.2\times10^{-4}$ & $7.5\times10^{-5}$ & $1.7\times10^{-4}$ & $5.1\times10^{-5}$ \\ \hline
    \end{tabular}
\end{table}

Next, we describe the details of the classification tasks using one-dimensional synthesized datasets. 
As described in the main text, the dataset consists of one-dimensional input data $\{x_{i}\}_{i=1}^{100}$ and labels $\{y_{i}\}_{i=1}^{100}$ defined by the sign of the sine function $$y_i = \text{sign} (\sin (w x_i + b))$$ with frequency $w$ and phase $b$. 
Here 80 data points are used for the training and the rest are used for the test. 
We use the datasets because the performance for the dataset is strongly linked to whether the QK has non-zero coefficients corresponding to the frequency used to determine the labels; namely, there is a relationship between the frequency of the dataset $w$ and the frequency of the Fourier coefficients $\omega,\omega'$.

We use support vector machines (SVMs) as classifiers to perform the tasks. 
We implement the SVM using SVC in scikit-learn \cite{buitinck2013api}, where an optimal hyperparameter C is chosen from $\{2^{t}|t=-8,-7,-6,\ldots,7,8,9 \}$ using 5-cross validation.
Note that the inverse of the hyperparameter acts as a regularization in the objective function of the SVMs,
\begin{equation}
\begin{split}
    & L(\bm{a}) = -\sum_{i=1}^{N} a_i +\frac{1}{2} \sum_{i=1}^{N}\sum_{j=1}^{N} a_i a_j y_i y_j k(\bm{x}_i,\bm{x}_j) \\
    & \text{subject to} \quad \sum_{i} y_i a_i = 0\\
    & \qquad \qquad \quad 0 \le a_i \le C \quad \text{for any } i=1,\ldots, N
\end{split}
\end{equation}
where $\bm{a}=\{a_{i}\}_{i=1}^{N}$ denotes the parameters to be optimized with the number of training data points $N$ and $k(\bm{x}_i,\bm{x}_j)$ is a kernel function.
This means that the Fourier coefficients of the QKs whose absolute values are less than $1/2^{9}\approx2\times 10^{-3}$ might not be taken into account in the minimization of the objective function. 
Thus values less than $10^{-3}$ are not shown in Figure~\ref{fig:fourier_analysis} (b).

Again, we prepare 10 sets of the randomly generated parameters in the PQC to evaluate the performance via the average of the misclassification rate.
We notice that the misclassification rate is defined as the number of wrong predictions over the number of total test data. 
As in the Fourier analysis, we used the ALAs and the HEAs with the number of qubits $n=1,2,3$ and depth $L=2,3,4$ to check the performance of the fidelity-based QK and the ALDQFK.
Figure~\ref{fig:1Dsynthesized} shows the averaged misclassification rate of each QK for datasets with the frequency $w\in\{2,\ldots,12\}$ and a fixed phase $b=0.3$, and the amplitudes of the Fourier coefficients $c_{\omega,-\omega}$ with $\omega\in\{2,\ldots,12\}$. 
We can obviously see that, for the fixed frequency $\omega=w$, QKs with small amplitudes of the Fourier coefficients $c_{\omega,-\omega}$ result in high misclassification rate for the synthesized dataset.
Thus, the expressivity comparison in terms of Fourier analysis can be considered effective.

\begin{figure}[tbp]
    \centering
    \includegraphics[keepaspectratio, scale=0.7]{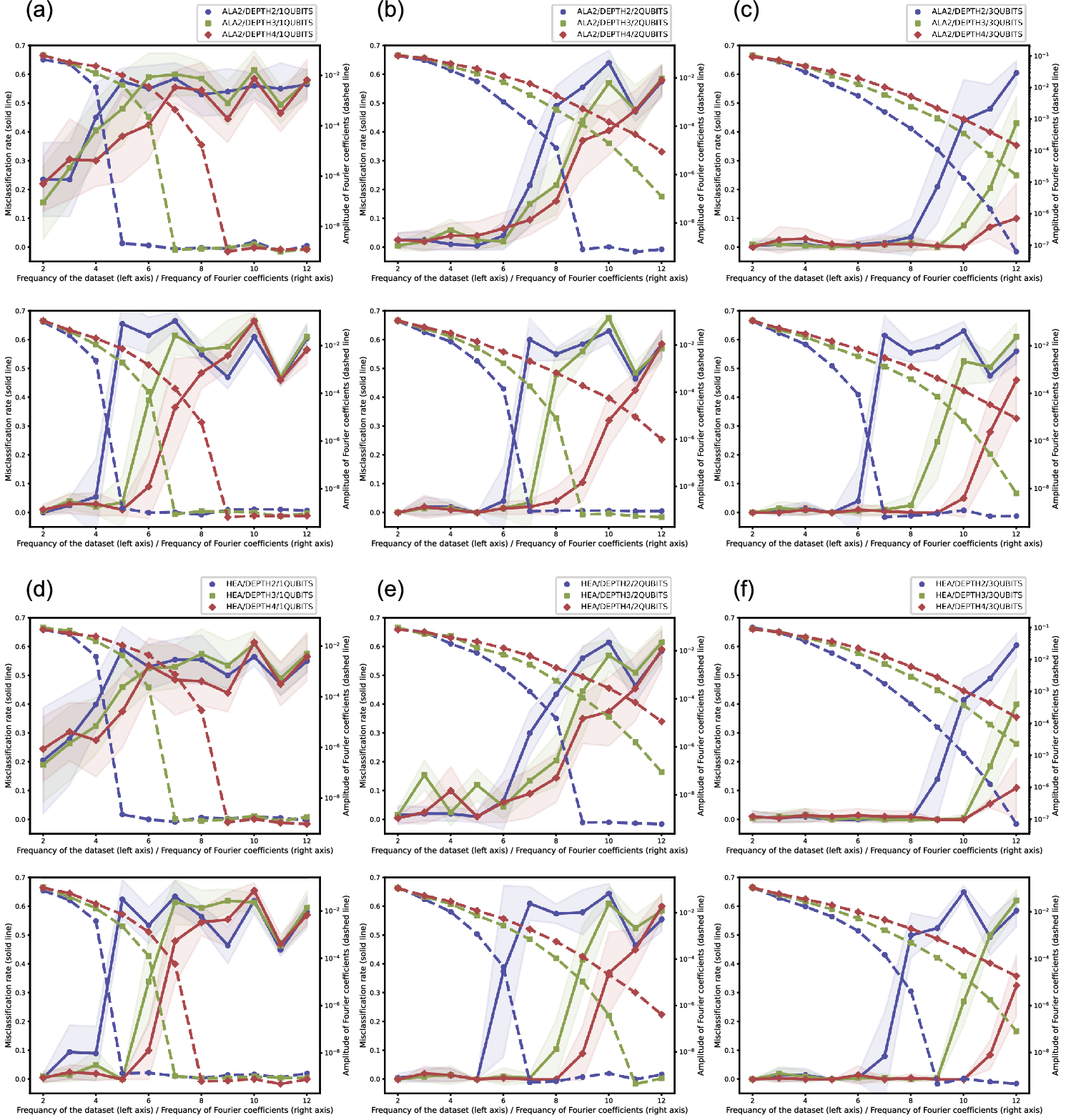}
    \caption{Misclassification rates of QKs for the synthesized datasets and the amplitude of the Fourier coefficients $c_{\omega,-\omega}$ given by the corresponding QKs. The misclassification rate for the datasets with the frequency $w\in\{2,\ldots,12\}$ and the amplitude of the Fourier coefficients $c_{\omega,-\omega}$ with $\omega\in\{2,\ldots,12\}$ are denoted by the solid line and the dashed line, respectively. Here the ALAs with the number of qubit $n=1,2,3$ are used for (a)-(c) and the HEAs with $n=1,2,3$ for (d)-(f). For each panel, the upper figure shows the result for the fidelity-based QK and the lower shows the one for the ALDQFK. Also, the standard deviation of the misclassification rates is indicated by the shaded areas. }
    \label{fig:1Dsynthesized}
\end{figure}

\subsection{Geometric difference between the fidelity-based QK and the ALDQFK}
In addition, we compare a geometric difference between the fidelity-based QK and the ALDQFK to investigate the case where high-dimensional input data is used.
The geometric difference is a measure to see the difference between two Gram matrices, as introduced in the flowchart for screening QKs with possible quantum advantage \cite{huang2021power}. 
This is expressed as  
\begin{equation}
    g_{a,b} = g(K_{a}||K_{b}) = \sqrt{\| \sqrt{K_{b}} \left(K_{a}\right)^{-1} \sqrt{K_{b}}  \|_{S}},
\end{equation}
with two Gram matrices $K_{a},K_{b}$ and the spectral norm $\|\cdot\|_{S}$.
Although it is not possible for the measure to directly assess the performance difference, it allows us to see the similarity between two Gram matrices regardless of the dimension of the data.
For this reason, we compute the geometric difference between these QKs using Fashion-MNIST datasets where the dimension of each data is reduced to the number of qubits by principle component analysis and then the reduced data is standardized.
As for the quantum circuits, the setups are the same as the case in “Numerical experiments”. 

Figure~\ref{fig:geometric_difference} shows the normalized geometric differences, i.e., $g_{a,b}/\sqrt{N}$ with the number of data points $N$, for the following situations; the ALDQFK against the fidelity-based QK ($[K_{a}]_{i,j}=k_{Q}(\bm{x}_{i},\bm{x}_{j})$, $[K_{b}]_{i,j}=\bar{k}_{QF}^{A}(\bm{x}_{i},\bm{x}_{j})$), the fidelity-based QK against the Kronecker delta ($[K_{a}]_{i,j}=\delta_{i,j}$, $[K_{b}]_{i,j}=k_{Q}(\bm{x}_{i},\bm{x}_{j})$) and the ALDQFK against the Kronecker delta ($[K_{a}]_{i,j}=\delta_{i,j}$, $[K_{b}]_{i,j}=\bar{k}_{QF}^{A}(\bm{x}_{i},\bm{x}_{j})$).
Here we use the normalized ALDQFK so that the trace of all Gram matrices in the comparison is the number of data points.
Also, note that the Gram matrix given by the Kronecker delta is exactly the identity matrix.
As a result, we can see that the difference between the fidelity-based QK and the ALDQFK gradually gets larger as the number of qubits increases, while there is a peak for the case $n=4$.
The tendency can be interpreted from Figure~\ref{fig:geometric_difference} (b) and (c); that is, the Gram matrices given by the fidelity-based QK get closer to the identity matrices with respect to the number of qubits, while the difference between those given by the ALDQFK and the identity matrices level off.
Actually, the difference between the fidelity-based QK and the ALDQFK is not reliable as in the case for $n=4$.
This is because the inverse of a Gram matrix $K_{a}$ is numerically instable. 
However, the different trends are shown in the geometric differences of the fidelity-based QK and the ALDQFK against the identity matrix where the inverse matrix calculation is stable.  
Thus, the results still indicate that the fidelity-based QK and the ALDQFK differ when the number of qubits increase, due to the vanishing similarity issue.

\begin{figure}[tbp]
    \centering
    \includegraphics[keepaspectratio, scale=0.8]{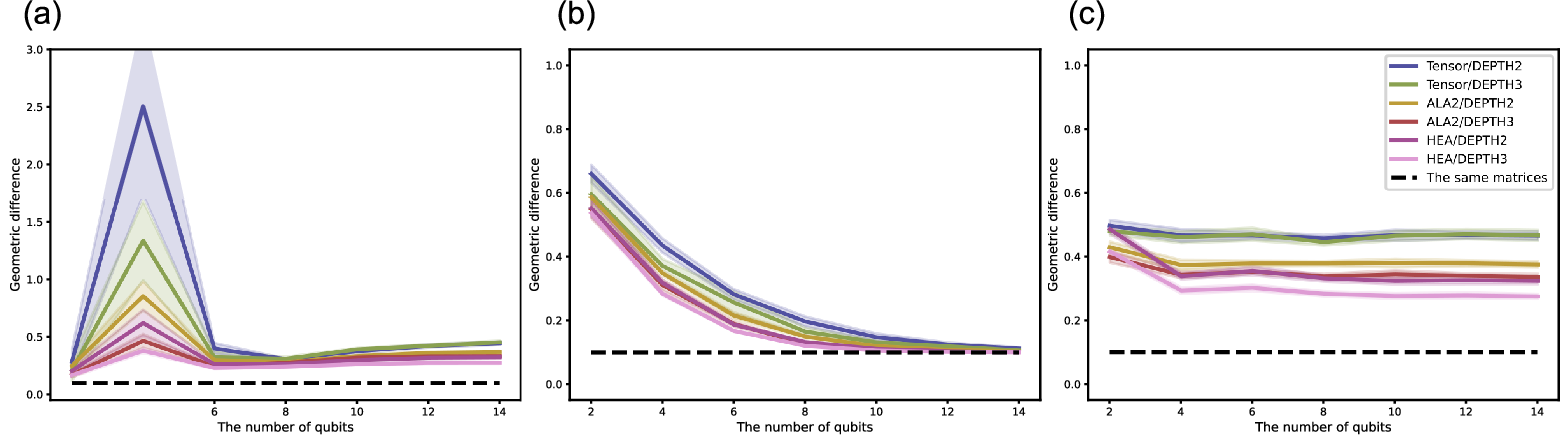}
    \caption{Comparison of the normalized geometric difference for (a) the ALDQFK against the fidelity-based QK, (b) the fidelity-based QK against the Kronecker delta and (c) the ALDQFK against the Kronecker delta. The differences of QKs against the Kronecker delta are also shown to see the difference between the Gram matrices given by the QKs and the identity matrices. For these panels the dashed lines show the case where two Gram matrices are the same. The shaded regions are shown to represent the standard deviation of the difference over 10 trials.}
    \label{fig:geometric_difference}
\end{figure}

\section{Connection between the SLDQFK and quantum neural tangent kernels}

In this section, we show a link between the SLDQFK and the quantum neural tangent kernel (QNTK).

First, let us restate the definition of the SLDQFK. 
The SLDQFK is constructed using the quantity called the SLD obtained from the follwoing equation 
\begin{equation}
\label{eq:sld}
    \partial_{\theta_l} \rho_{\bm{x},\bm{\theta}} = \frac{1}{2}\left(\rho_{\bm{x},\bm{\theta}}L_{\bm{x},\theta_{l}}^{S}+L_{\bm{x},\theta_{l}}^{S}\rho_{\bm{x},\bm{\theta}}\right),
\end{equation}
where $\rho_{\bm{x},\bm{\theta}}=U(\bm{x},\bm{\theta})\rho_0 U^{\dagger}(\bm{x},\bm{\theta})$ and the partial derivative with respect to the parameter $\theta_{l}$ is denoted as $\partial_{\theta_l}\equiv \partial/\partial_{\theta_{l}}$.
Although the analytical solution is not uniquely determined, we can get a closed form for the case where the initial state is pure. 
One solution can be expressed as 
\begin{equation}
\label{eq:sld_pure}
    L_{\bm{x},\theta_{l}}^{S} = 2\partial_{\theta_l} \rho_{\bm{x},\bm{\theta}}.
\end{equation}
Then, the SLDQFK is defined as follows:
\begin{equation}
\label{eq:def_sldqfk}
    k_{QF}^{S}(\bm{x},\bm{x'}) =  \sum_{i,j} \mathcal{F}_{S,i,j}^{-1} \mathrm{Tr}\left[ L_{\bm{x},\theta_i}^{S} L_{\bm{x'},\theta_j}^{S} \right],
\end{equation}
with the SLD-based QFIM $\mathcal{F}_{S}$.

On the other hand, the QNTK is a quantum analogue of the classical neural tangent kernel that helps us to understand the training dynamics of the quantum neural networks (QNNs) analytically. 
So far, several definitions of the QNTKs have been proposed \cite{liu2022representation,nakaji2021quantum,shirai2021quantum}; however, we follow the one in Ref \cite{liu2022representation} here. 
Suppose the cost function of a QNN is defined as
\begin{equation}
    L(\bm{\theta}) = \frac{1}{2}\sum_{i}\left( \bm{\hat{y}}_{i} - \bm{y}_{i}  \right)^2, 
\end{equation}
where $\bm{\hat{y}}_{i}= [ \mathrm{Tr}[\rho_{\bm{x}_{i},\bm{\theta}} O_{1} ],\mathrm{Tr}[\rho_{\bm{x}_{i},\bm{\theta}} O_{2} ],\ldots \mathrm{Tr}[\rho_{\bm{x}_{i},\bm{\theta}} O_{n_{o}} ] ] $ represents the outputs of the QNN with $n_{o}$ observables $\{O_{m}\}_{m=1}^{n_{o}}$ and the target values $\bm{y}_{i}$.
Here, $\rho_{\bm{x}_{i},\bm{\theta}}$ denotes a input- and parameter-dependent quantum states, as is in the SLDQFK.
Then, the QNTK associated with the gradient descent equation on $\hat{y}_{i,m} = \mathrm{Tr}[\rho_{\bm{x}_{i},\bm{\theta}} O_{m} ]$ is defined as follows
\begin{equation}
    k_{i,i'}^{m,m'} = \sum_{l} \frac{d\hat{y}_{i,m} }{d\theta_{l}} \frac{d\hat{y}_{i',m'} }{d\theta_{l}}.
\end{equation}
For the detailed explanation of the QNTK, please refer to Ref. \cite{liu2022representation}.

Indeed, using the SLD for pure states in Eq.~\eqref{eq:sld_pure}, the QNTK can be rewritten as
\begin{equation}
     k_{i,i'}^{m,m'} = \frac{1}{4} \sum_{l}  \mathrm{Tr}[L_{\bm{x}_{i},\theta_{l}}^{S} O_{m}] \mathrm{Tr}[L_{\bm{x}_{i'},\theta_{l}}^{S} O_{m'}],
\end{equation}
where we utilize the following equality
\begin{equation}
    \frac{d\hat{y}_{i,m} }{d\theta_{l}} = \mathrm{Tr}[\partial_{\theta_{l}}\rho_{\bm{x}_{i},\bm{\theta}} O_{m}] = \frac{1}{2} \mathrm{Tr}[L_{\bm{x}_{i},\theta_{l}}^{S} O_{m}].
\end{equation}
Then the SLDQFK with $\mathcal{F}=\mathbb{I}$ can be rewritten as
\begin{equation}
\label{eq:qntk_sldqfk}
    k_{QF}^{S}(\bm{x}_{i},\bm{x}_{i'}) = \frac{1}{2^{n-2}} \sum_{m\in\{k|O_{m}\in\{I,X,Y,Z\}^{\otimes n}\}} k_{i,i'}^{m,m}.  
\end{equation}
Thus, it would be interesting to investigate the SLDQFK as a theoretical tool to analyze the performance of the QNN.
It will also be important to check the capability of the QNTK in terms of the vanishing similarity issue; since the SLDQFK suffers from the issue, the QNTK may be subjected to the same issue, which probably hinders performance analysis. 
Moreover, it is worth exploring the case where $\mathcal{F}\neq \mathbb{I}$, because the SLDQFK has the property that it is invariant under invertible differentiable transformations of parameters, i.e., $k_{QF,\bm{\theta}}^{S}(\bm{x}_{i},\bm{x}_{i'})=k_{QF,\Psi(\bm{\theta})}^{S}(\bm{x}_{i},\bm{x}_{i'})$ under the transformation $\Psi: \bm{\theta} \to \Psi(\bm{\theta})$, which might be related to the training dynamics of the QNN.

\bibliographystyle{naturemag}
\bibliography{supp_info.bib}